\documentclass[11pt]{article}
\usepackage[papersize={8.5in,11in}, margin=1in]{geometry}	
\usepackage{lmodern}										

\usepackage[utf8]{inputenc}
\usepackage[T1]{fontenc}
\usepackage{graphicx}
\usepackage{soul}
\usepackage{amsmath,amsfonts,amsthm,amssymb}
\usepackage{thm-restate}
\usepackage{authblk}
\usepackage{xfrac}                                        
\usepackage{xspace}
\usepackage{enumerate}
\usepackage[algoruled, vlined, linesnumbered]{algorithm2e}
\usepackage{setspace}

\usepackage{wrapfig}					

\usepackage[dvipsnames]{xcolor}

\usepackage{tikz}
\usetikzlibrary{arrows.meta}
\usepackage[colorlinks=true, urlcolor=blue, linkcolor=blue, citecolor=ForestGreen]{hyperref}
\usepackage[noabbrev, nameinlink]{cleveref}					
															
\crefname{ineq}{inequality}{inequalities}					
\creflabelformat{ineq}{#2{(\upshape#1)}#3}					

\newtheorem{theorem}{Theorem}
\newtheorem{corollary}[theorem]{Corollary}
\newtheorem{lemma}[theorem]{Lemma}

\newtheorem{case}{Case}
\counterwithin*{case}{theorem}								
\theoremstyle{definition}
\newtheorem{definition}[theorem]{Definition}

\newcommand*{\G}{\mathcal{G}}

\newcommand*{\ball}{\operatorname{ball}}
\newcommand*{\rank}{\operatorname{rank}}

\DeclareMathOperator{\tr}{tr}
\DeclareMathOperator{\seg}{seg}

\newcommand*{\Otilde}{\widetilde{O}}
\newcommand*{\poly}{\textsf{poly}}

\newcommand*{\length}[1]{|#1|}

\newcommand*{\nwspace}{\hspace*{.1em}} 

\renewcommand{\leq}{\leqslant}
\renewcommand{\geq}{\geqslant}
\renewcommand{\le}{\leqslant}
\renewcommand{\ge}{\geqslant}
\renewcommand{\epsilon}{\varepsilon}
\newcommand{\eps}{\varepsilon}

\let\oldsqrt\sqrt
\def\hksqrt{\mathpalette\DHLhksqrt}
\def\DHLhksqrt#1#2{\setbox0=\hbox{$#1\oldsqrt{#2\,}$}\dimen0=\ht0
   \advance\dimen0-0.2\ht0
   \setbox2=\hbox{\vrule height\ht0 depth -\dimen0}%
   {\box0\lower0.4pt\box2}}
\renewcommand\sqrt\hksqrt

\title{Improved Distance (Sensitivity) Oracles with Subquadratic Space}

\date{}

\author[1]{Davide Bilò}
\author[2]{Shiri Chechik}
\author[3]{Keerti Choudhary}
\author[4]{\\\vspace*{.5em} Sarel Cohen}
\author[5]{Tobias Friedrich}
\author[6]{Martin Schirneck}

\affil[1]{Department of Information Engineering, Computer Science and Mathematics, 
	
	University of L'Aquila

	\texttt{davide.bilo@univaq.it}
	\vspace*{.5em}
}
\affil[2]{Department of Computer Science, Tel Aviv University

	\texttt{shiri.chechik@gmail.com}
	\vspace*{.5em}
}
\affil[3]{Department of Computer Science and Engineering, Indian Institute of Technology Delhi,

	\texttt{keerti@iitd.ac.in}
	\vspace*{.5em}
}
\affil[4]{School of Computer Science, The Academic College of Tel Aviv-Yaffo

	\texttt{sarelco@mta.ac.il}
	\vspace*{.5em}
}
\affil[5]{Hasso Plattner Institute, University of Potsdam

	\texttt{tobias.friedrich@hpi.de}
	\vspace*{.5em}
}
\affil[6]{Faculty of Computer Science, University of Vienna

	\texttt{martin.schirneck@univie.ac.at}
}

\begin{document}
\maketitle

\thispagestyle{empty}

\begin{abstract}
A \emph{distance oracle} (DO) for a graph $G$ is a data structure that, when queried with vertices $s,t$, returns an estimate $\widehat{d}(s,t)$ of their distance in $G$. The oracle has stretch  $(\alpha, \beta)$ if the estimate satisfies 
$d(s,t) \le \widehat{d}(s,t) \le \alpha \cdot d(s,t) + \beta$. An $f$-\emph{edge fault-tolerant distance sensitivity oracle} ($f$-DSO) additionally receives a set $F$ of up to $f$ edges 
and estimates the distance in $G{-}F$.

Our first contribution is the design of new distance oracles with subquadratic space for undirected graphs. We show that introducing a small additive stretch $\beta > 0$ allows one to make the multiplicative stretch $\alpha$ arbitrarily small.
This sidesteps a known lower bound of $\alpha \ge 3$ (for $\beta = 0$ and subquadratic space) [Thorup \& Zwick, JACM 2005]. We present a DO for graphs with edge weights in $[0,W]$ that, for any positive integer $\ell$ and any $c \in (0, \ell/2]$, has stretch $(1{+}\frac{1}{\ell}, 2W)$, space $\Otilde(n^{2-\frac{c}{\ell}})$, and query time $O(n^c)$.
These are the first subquadratic-space DOs with $(1+\epsilon, O(1))$-stretch generalizing Agarwal and Godfrey's results for sparse graphs [SODA 2013] to general undirected graphs. We also construct alternative DOs with even smaller space at the cost of a higher additive stretch. For any integer $k \ge 1$, the DOs have a stretch $(2k{-}1{+}\frac{1}{\ell},4kW)$, space $O(n^{1+\frac{1}{k}(1- \frac{c}{8\ell})})$, and query time $O(n^{c})$.

Our second contribution is a framework that turns any $(\alpha,\beta)$-stretch DO for unweighted graphs into an $(\alpha (1{+}\varepsilon),\beta)$-stretch \(f\)-DSO
with sensitivity $f = o(\log(n)/\log\log n)$ and retains subquadratic space.
This generalizes a result by Bilò, Chechik, Choudhary, Cohen, Friedrich, Krogmann, and Schirneck [STOC 2023, TheoretiCS 2024] for the special case of stretch $(3,0)$ and $f = O(1)$. 
We also derandomize the entire construction. 
By combining the framework with our new distance oracle, we obtain an \(f\)-DSO that, for any $\gamma \in (0, (\ell{+}1)/2]$, has stretch $((1{+}\frac{1}{\ell}) (1{+}\varepsilon), 2)$, space $n^{ 2- \frac{\gamma}{(\ell+1)(f+1)} + o(1)}/\varepsilon^{f+2}$, and query time $\Otilde(n^{\gamma} /{\varepsilon}^2)$. This is the first deterministic $f$-DSO with subquadratic space, near-additive stretch, and sublinear query time.
\end{abstract}

\section{Introduction}
\label{sec:intro}

A \emph{distance oracle} (DO) is a data structure to retrieve (exact or approximate) distances between any pair of vertices $s,t$ in an undirected graph $G = (V,E)$ upon query.
The problem of designing distance oracles has attracted a lot of attention in recent years
due to the wide applicability in domains like network routing, traffic engineering, and distributed computing.
These oracles are used in settings where one cannot afford to store the entire graph,
but still wants to be able to quickly query graph distances.
A DO has \emph{stretch} $(\alpha,\beta)$ if, for any pair $s$ and $t$, the value $\widehat{d}(s,t)$ returned by the DO satisfies $d(s,t) \leq \widehat{d}(s,t) \leq \alpha \cdot d(s,t)+\beta$, where $d(s,t)$ denotes the exact distance between $s$ and $t$ in $G$.
As networks in most real-life applications are prone to transient failures, researchers have also studied the problem of designing oracles that additionally tolerate multiple edge failures in $G$. 
An $f$-\emph{edge fault-tolerant distance sensitivity oracle} ($f$-DSO) with stretch $(\alpha, \beta)$ is a data structure that, when queried on a triple $(s,t,F)$, where $F\subseteq E$ has size at most $f$, outputs an estimate $\widehat{d}(s,t,F)$ of the distance $d(s,t,F)$ from $s$ to $t$ in $G{-}F$ such that $d(s,t,F) \leq \widehat{d}(s,t,F) \leq \alpha \cdot d(s,t,F)+\beta$.

Several DOs and $f$-DSOs with different size-stretch-time trade-offs have been developed in the last decades. See, for example, \cite{Agarwal14SpaceStretchTimeTradeoffDOs,Baswana09Nearly2ApproximateJournal,ChCo20,ChCoFiKa17,
dey2024nearly,DuanP09a,DuRe22,DemetrescuT02,GrandoniVWilliamsFasterRPandDSO_journal,
PatrascuRoditty14BeyondThorupZwick,Ren22Improved,ThorupZ05,WY13}.
Our focus is on providing new distance oracles with a subquadratic space usage  for both static and error-prone graphs. 
A special focus of this work is how static distance oracles
can be converted into fault-tolerant distance \emph{sensitivity} oracles.

\subsection{Approximate Distance Oracles for Static Graphs}
\label{sec:intro_DO}

We first discuss distance oracles.
Extensive research has been dedicated in that direction in the past two decades. 
In their seminal paper \cite{ThorupZ05}, Thorup and Zwick showed that,
for any positive integer $k$ (possibly depending on the input size), 
an undirected graph with $n$ vertices and $m$ edges can be preprocessed in time $O(mn^{1/k})$ 
to obtain an oracle with multiplicative stretch $2k{-}1$, 
space\footnote{%
  The space of the data structures is measured in the number of machine words on $O(\log n)$ bits.
} $O(kn^{1+1/k})$, and query time $O(k)$. 
Subsequent works \cite{Chechik14,Chechik15,wulff2013approximate}
further improved the space of the construction to $O(n^{1+1/k})$ and the query time to $O(1)$.

In the case of $k=2$, this results in a $3$-approximate oracle taking subquadratic space $O(n^{3/2})$.
A simple information theoretic lower bound using bipartite graphs \cite{ThorupZ05}
shows that distance oracles with a purely multiplicative stretch \emph{below} $3$ 
require space that is at least quadratic in $n$.
P\v{a}tra\c{s}cu and Roditty~\cite{PatrascuRoditty14BeyondThorupZwick} 
were arguably the first to introduce an additive stretch to 
simultaneously reduce the space and the multiplicative stretch in general dense graphs.
They proposed a distance oracle 
for unweighted graphs with stretch $(2,1)$ that takes $O(n^{5/3})$ space,
has a a constant query time and can be constructed in time $O(mn^{2/3})$.
They also showed that DOs with multiplicative stretch $\alpha \le 2$ 
and constant query time require $\Omega(n^2)$ space,
assuming a conjecture on the hardness of set intersection queries.

Agarwal and Godfrey~\cite{AgarwalGodfrey13DOsStretchLessThan2} 
studied $(1+\varepsilon,O(1))$-stretch DOs for sparse graphs.
However, when transferred to dense graphs the space of their construction becomes $\Omega(n^2)$
and the query time is $\Omega(n)$.
To the best of our knowledge, no distance oracles have been constructed that simultaneously have
subquadratic space and a multiplicative stretch better than 2 for general undirected graphs.
If we want to reduce $\alpha$ while retaining low space,
we necessarily have to introduce some additive stretch $\beta$
and the query time must rise beyond constant. 
This raises the following natural question.

\pagebreak

\noindent
\textbf{Question.}
\emph{Is there a distance oracle with a $1+\varepsilon$ multiplicative and constant additive stretch that takes subquadratic space and has a query time that is sublinear in $n$?}
\vspace*{.75em}

We provide an affirmative answer by presenting the first oracle with stretch $(1{+}\varepsilon,O(1))$ for general graphs. 
The construction is surprisingly simple, has subquadratic space, and sublinear query time, 
improving over Agarwal and Godfrey's work for dense graphs.

\begin{restatable}{theorem}{distanceoracle}
\label{thm:distance_oracle}
	Let $W$ be a non-negative real number,
	and $G$ an undirected graph with $n$ vertices and edge weights 
	in a $\textup{\poly}(n)$-sized subset of $[0,W]$.
	For every positive integer $K \le \sqrt{n}$ and any $\varepsilon > 0$,
	there exists a path-reporting distance oracle for $G$
	that has stretch $(1{+}\varepsilon, 2W)$,
	space $\Otilde(n^2/K)$,
	and query time $O(K^{\lceil 1/\varepsilon \rceil})$
	for the distance and an additional $O(1)$ time per reported edge.
	The data structure can be constructed in APSP time.
\end{restatable}

The restriction of the edge weights to a polynomial-sized subset of $[0,W]$ is to ensure
that any graph distance can be encoded in a constant number of $O(\log n)$-bit words.
We remark that our construction in fact guarantees a stretch of $(1{+}\varepsilon, 2w_{s,t})$ where $w_{s,t}$ is the maximum edge weight along a shortest path from $s$ to $t$ in $G$.
The stretch thus depend locally on the queried vertices rather than the global maximum edge weight. 

One cannot reduce the additive stretch in \Cref{thm:distance_oracle} to $1$ (if $\varepsilon < 2$) 
even in unweighted graphs as the unconditional lower bound for bipartite graphs \cite{ThorupZ05}
stated earlier rules out any data structure that can distinguish between 
distances $1$ and $3$ in subquadratic space.

By setting $\varepsilon = 1/\ell$ and $K^\ell = n^c$, we get a trade-off between 
stretch space and time.

\begin{corollary}
\label{cor:distance_oracle}
	For every positive integer $t$ and any $0 < c \le \ell/2$,
	there exists a distance oracle for undirected graphs with stretch $(1{+}\frac{1}{\ell}, 2W)$,
	space $\Otilde(n^{2-\frac{c}{\ell}})$,
	and query time $O(n^c)$.
\end{corollary}

An extension of our construction allows to trade a higher stretch for a lower space.
Namely, we present a family of DOs with multiplicative stretch of $2k{-}1+\varepsilon$ and $o(n^{1+1/k})$ space.

\begin{restatable}{theorem}{hierarchydo}
\label{thm:hierarchy_DO}
	Let $W$ be a non-negative real number,
	and $G$ an undirected graph with $n$ vertices and edge weights 
	in a $\textup{\poly}(n)$-sized subset of $[0,W]$.
	For all positive integers $k$ and $K$ with $K = O(n^{1/(2k+1)})$, 
	and every $\varepsilon > 0$,
	there exits a distance oracle for $G$ that
	has stretch $(2k{-}1{+}\varepsilon, 4kW)$,
	space $O((\frac{n}{K})^{1+1/k} \, \log^{1+1/k} n)$,
	and query time $O(K^{2\lceil 4k/\varepsilon \rceil})$.
	The data structure can be constructed in APSP time.
\end{restatable}

For $k=1$, the oracle in \Cref{thm:hierarchy_DO} has the same multiplicative stretch of $1+\varepsilon$
as the one in \Cref{thm:distance_oracle} and a better space of $\Otilde(n^2/K^2)$,
but the additive stretch of $4W$ is larger and so is the query time.
We obtain the following corollary for general $k$, $\varepsilon = 1/\ell$, and $K^{8k\ell} = n^{c}$.

\begin{corollary}
\label{cor:2k-1+eps-oracle}
	For all positive integers $k$ and $\ell$, and any $0 < c \le (4{-}\frac{4}{2k+1}) \, \ell$,
	there exists a distance oracle with stretch $(2k{-}1{+}\frac{1}{\ell},4kW)$,
	space $\Otilde(n^{1+\frac{1}{k}(1-\frac{c}{8\ell})})$, and query time $O(n^{c})$.
\end{corollary}

Probably closest to hierarchy we present in \Cref{thm:hierarchy_DO}
is the distance labeling scheme of Abraham and Gavoille~\cite{Abraham11AffineStretch}.
Seen as a DO for unweighted graphs, for any integer $k \ge 2$,
it has a stretch of $(2k{-}2,1)$
space $\Otilde(n^{1+\frac{2}{2k-1}})$, and query time $O(k)$.

\subsection{Distance Sensitivity Oracles}

Most of the proposed distance sensitivity oracles that treat the sensitivity $f$ 
(the number of tolerated edge failures) as a parameter
require $\Omega(n^2)$ space, have a stretch depending on $f$, or an $\Omega(n)$ query time. 
See \Cref{subsec:related_work} for more details.
Bilò, Chechik, Choudhary, Cohen, Friedrich, Krogmann, and Schirneck~\cite{Bilo24ApproxDSOSubquadraticTheoretiCS} were the first to introduce an $f$-DSO with subquadratic space,
a constant multiplicative stretch of $3+\epsilon$ (for any $f$),
and a query time that can be made an arbitrarily small polynomial.
More precisely, for any unweighted graph $G$ with unique shortest paths,
every integer constant $f \ge 2$, any $0 < \gamma < \sfrac{1}{2}$, and $\varepsilon > 0$,
they devised an $f$-DSO for $G$ with stretch $3{+}\eps$, space 
  $\Otilde(n^{2-\frac{\gamma}{f+1}}) \,{\cdot}\, O(\log n/\varepsilon)^{f+2}$,
query time $O(n^{\gamma}/\varepsilon^2)$,
  and preprocessing time
  $\Otilde(mn^{2-\frac{\gamma}{f+1}})
  	\cdot O(\log n/\varepsilon)^{f+1}$.
Even more than in the case of static distance oracles, a multiplicative stretch better than $3$
(let alone close to $1$) remains a barrier for subquadratic-space \mbox{$f$-DSOs}.
We explore whether the introduction of a small additive stretch can help here as well.
\vspace*{.75em}

\noindent
\textbf{Question.}
\emph{Is there a distance sensitivity oracle for general sensitivity $f$ 
with a $1+\varepsilon$ multiplicative stretch, possibly at the expense of a constant additive stretch,
that takes subquadratic space and has a query time that is sublinear in $n$?}
\vspace*{.75em}

We also answer this question affirmatively in this paper by presenting an $f$-DSO with stretch
$(1{+}\epsilon, 2)$ while still having subquadratic space and small polynomial query time.

\begin{restatable}{theorem}{chained}
\label{thm:chained}
	Let $\ell$ be a positive integer constant
	and $G$ be an undirected and unweighted graph with $n$ vertices and $m$ edges
	and unique shortest paths.
	For any $0 < \gamma \le (\ell{+}1)/2$, sensitivity $2 \,{\le}\, f \,{=}\, o(\log(n)/\log\log n)$, and 
	approximation parameter $\varepsilon = \omega(\sqrt{\log(n)}/n^{\frac{\gamma}{2(\ell+1)(f+1)}})$,
	there exists an $f$-edge fault-tolerant distance sensitivity oracle for $G$ 
	that has
	\begin{itemize}
		\item stretch $((1{+}\frac{1}{\ell})(1{+}\varepsilon), \nwspace 2)$,\vspace*{-.25em}
		\item space $n^{2- \frac{\gamma}{(\ell+1)(f+1)} + o(1)}/\varepsilon^{f+2}$,\vspace*{-.25em}
		\item query time $O(n^\gamma/\varepsilon^2)$,
		\item preprocessing time $n^{2+\gamma + o(1)} 
			+ mn^{2- \frac{\gamma}{(\ell+1)(f+1)} + o(1)}/\varepsilon^{f+1}$.
	\end{itemize}
\end{restatable}

We observe that the additive stretch of $\beta=2$ is necessary due to the unconditional lower bound of $\Omega(n^2)$ on the size of every distance oracle with a purely multiplicative stretch of $\alpha <3$, as discussed in~\Cref{sec:intro_DO}. The assumption of unique shortest paths can be achieved, for example, by perturbing the edge weights with random small values. This means that an unweighted graph $G$ becomes weighted and its edge weights are very close to 1. As an alternative, we can compute, in time $O(mn + n^2 \log^2 n)$, a set of unique shortest paths via \emph{lexicographic pertubation}~\cite{CaChEr13}. 

Our main technique to obtaain the new distance sensitivity oracle
is to develop a reduction that, given a path-reporting DO with stretch $(\alpha, \beta)$,
constructs a $((1{+}\epsilon)\alpha, \beta)$-stretch $f$-DSO.
Crucially, the reduction results in a subquadratic-space $f$-DSO 
provided that the initial DO also takes only subquadratic space.
Although the problem of designing static as well as fault-tolerant distance (sensitivity) oracles has been studied extensively in the past two decades, there has been no substantial progress to obtain black-box conversions from compact DOs to compact $f$-DSOs.
The work by Bilò et al.~\cite{Bilo24ApproxDSOSubquadraticTheoretiCS} comes close,
but their construction of the $f$-DSO is entangled with the inner workings
of the DO of Thorup and Zwick~\cite{ThorupZ05}.
Also, their analysis relies heavily on the input distance oracle having
a purely multiplicative stretch of $3$.
In contrast, we develop algorithms that can work with any distance oracle
as long as it is able to report an approximate shortest path that adheres to the stretch bound.
Our analysis is able to incorporate any multiplicative stretch $\alpha \ge 1$ 
as well as additive stretch $\beta$ that satisfies very mild technical assumptions.
(See \Cref{thm:general_DSO} for the precise statement.)
Plugging in the distance oracle with stretch $(1+\frac{1}{\ell}, 2)$ from \Cref{cor:distance_oracle}
then allows us to achieve the $((1+\frac{1}{\ell})(1+\eps), 2)$-stretch $f$-DSO with subquadratic space 
and small polynomial query time from \Cref{thm:chained}. 

\paragraph{Emulating Searches in Fault-Tolerant Trees.} 
Our transformation that makes distance oracles fault tolerant has two major steps.
In the first one, we only treat \emph{hop-short} replacement paths, these are shortest paths in $G{-}F$
whose number of edges is bounded by some cut-off parameter $L$.
We transform any $(\alpha,\beta)$-stretch distance oracle into an $f$-distance sensitivity oracle
for hop-short paths ($f$-DSO$^{\leq L}$) with the same stretch.
The second step then combines the solutions for hop-short paths 
into a general distance sensitivity oracle ($f$-DSO).

For the the first step, we develop \emph{fault-tolerant tree oracles},
which are a new way to retrieve hop-short replacement paths from a previously known data structure called fault-tolerant trees (FT-trees), but without actually storing those trees.
FT-trees were originally introduced by Chechik, Cohen, Fiat, and Kaplan~\cite{ChCoFiKa17}.
There is an FT-tree $FT(s,t)$ required for every pair of vertices and a query traverses along a root-to-leaf path in the respective tree, 
until a shortest $s$-$t$-path in $G-F$ is found.
This solution uses super-quadratic $\Omega(n^2 L^f)$ space in total and is thus too large for our purpose.
We devise a technique to emulate a search in an FT-tree without access to the tree.
We use those searches to generate a carefully chosen family of subgraphs of $G$
and apply the input distance oracle to each of them.
This ensures that any query $(s,t,F)$ that satisfies $d(s,t,F)\leq L$, is answered with an $(\alpha, \beta)$-approximate path.

\begin{restatable}{theorem}{shortdistanceDSO}
\label{thm:short_distance_DSO}
Let $G$ be an unweighted (possibly directed) graph, with $n$ vertices, 
and let $f$ and $L$ be positive integer parameters possibly depending on $n$.
Assume access to a path-reporting distance oracle that, on any spanning subgraph of $G$, 
has stretch $(\alpha,\beta)$, takes space $\mathsf{S}$, query time $\mathsf{Q}$, 
and preprocessing time $\mathsf{T}$.
Then, there is an $f$-edge fault-tolerant distance sensitivity oracle for replacement paths in $G$
with at most $L$ edges that has 
\begin{itemize}
	\item stretch $(\alpha,\beta)$,\vspace*{-.25em}
	\item space $O(fL \log n)^f \cdot \mathsf{S}$,\vspace*{-.25em}
	\item query time $\Otilde(f \cdot (\mathsf{Q} + \alpha L + \beta + f^2 L))$,\vspace*{-.25em}
	\item preprocessing time $\Otilde(n^2 (\alpha L + \beta)^f) \cdot ( O(fL \log n)^f + \mathsf{Q} )
				+ O(fL \log n)^f \cdot \mathsf{T}$.
\end{itemize}
\end{restatable}

\noindent
Note that the behavior of $f$-DSOs (hop-short or general) are usually only discussed for valid queries,
that is, triples $(s,t,F)$ where the edges in $F$ are actually present in $G$.
\emph{Checking} for validity requires $\Omega(n^2)$ space,
which would make subquadratic-space $f$-DSO impossible.
Our query algorithm never uses the assumption $F \subseteq E$
and the data structure can be queried with any triplet $(s,t,F)$ 
where $F \subseteq \binom{V}{2}$ is a set of at most $f$ pairs of vertices.
In this case it returns the approximate distance for the valid query $(s,t, F \cap E)$.\footnote{
	See \Cref{lem:correctness_query_short-distance_DSO,lem:weaker_guarantee} for the technical details.
}

\paragraph{General Distance Sensitivity Oracles.} 
Bilò et al.~\cite{Bilo24ApproxDSOSubquadraticTheoretiCS} described how to apply an $f$-DSO$^{\le L}$
(with the restriction to hop-short paths)
to construct a general distance sensitivity oracle without the constraint.
They used two sets of pivots, which are vertices at with the hop-short solutions are combined.
However, their approach only works for $f$-DSO$^{\le L}$ with purely multiplicative stretch of $3$,
resulting in an $f$-DSO with multiplicative stretch $3+\varepsilon$.
We generalize this by introducing a new data structure called \emph{pivot trees}, as well as pinpointing the places in their construction that need to be adapted to accompany multiplicative stretch $\alpha \neq 3$ an additive stretch $\beta > 0$.
The new pivot trees allow us to quickly find relevant pivots in $G{-}F$ 
that are close enough to the endpoints $s$ and $t$ in the query. 
An additional difference to \cite{Bilo24ApproxDSOSubquadraticTheoretiCS} is a more involved analysis of the stretch.
The issue is that the additive part $\beta$ accumulates while the answers of the $f$-DSO$^{\le L}$
are aggregated.
Fortunately, this happens only if the paths in questions are hop-long,
that is, if they have more than $L$ edges.
This allows us to introduce an inductive argument controlling the stretch accumulation
and charge the overhead to the multiplicative part instead.
As a result, we are able to turn a $f$-DSO$^{\le L}$ with stretch $(\alpha, \beta)$ 
into a general $f$-DSO that has stretch $(\alpha (1{+}\epsilon), \beta)$,
all this while keeping the total space subquadratic.

To simplify the presentation, we first provide a randomized solution in \Cref{sec:framework_DSO} and then show how to 
derandomize it in \Cref{sec:derandomization}.
For the randomized construction, the guarantees hold \emph{with high probability} (w.h.p.),
which we define as with probability at least $1- n^{-c}$ for some constant $c >0$.
In fact, $c$ can be made arbitrarily large without affecting the asymptotics.
We later show how to derandomize the oracle without any loss in its features,
under very mild assumption on the parameters $f$ and $L$.
The main source of randomization is the creation of the pivot sets.
The aforementioned pivot trees not only allow us to get better bounds for the randomized data structure 
but even help with derandomizing it. 

To the best of our knowledge, we provide the first deterministic $f$-DSO with subquadratic space, near-additive stretch, and sublinear query time.

\begin{restatable}{theorem}{generaldso}
\label{thm:general_DSO}
	Let $G$ be an undirected, unweighted graph with $n$ vertices, $m$ edges,
	and unique shortest paths.
	Let $f$ and $L$ be positive integer parameters possibly depending on $n$ and $m$
	such that $2 \le f \le L \le n$ as well as $L = \omega(\log n)$.
	Assume access to an $f$-edge fault-tolerant distance sensitivity oracle for replacement paths in $G$
with at most $L$ edges,
	that has stretch $(\alpha, \beta)$, space $\mathsf{S}_L$, query time $\mathsf{Q}_L$,
	and preprocessing time $\mathsf{T}_L$.
    Then, for every
	$\varepsilon = \varepsilon(n, m, f, L) > 0$,
	and $\beta = o\!\left(\frac{\varepsilon^2 L}{f^3 \log n}\right)$,
	there is a randomized (general) $f$-edge fault-tolerant distance sensitivity oracle for $G$ 
	that with high probability has
	\begin{itemize}
	\item stretch $(\alpha (1{+}\varepsilon), \beta)$,\vspace*{-.25em}
	\item query time $\Otilde(f^5 L^{f-1} (\mathsf{Q}_L+f)/\varepsilon^2)$,\vspace*{-.25em}
	\item preprocessing time 
		$\mathsf{T}_L  + O(L^{3f}n) + \Otilde(f^2 mn^2/L) \cdot O(\log n/\varepsilon)^{f+1}$.
	\end{itemize}
	The space of the data structure is w.h.p.\
	\begin{equation*}
		\mathsf{S}_L + \Otilde(fL^{2f-1} n)
		+ \Otilde\!\left(f^2 \frac{n^2}{L}\right) \cdot
			O\!\left(\frac{\log n}{\varepsilon}\right)^{f+2}.
	\end{equation*}
	If additionally $f \ge 4$ and $L = \Otilde(\sqrt{f^3m/\varepsilon})$,
	the data structure can be made deterministic with the same
	stretch, query time, preprocessing time, and space.
\end{restatable}

If $\Otilde(fL^{2f}) = \Otilde(n)$, the space in \Cref{thm:general_DSO} simplifies to
$\mathsf{S}_L + \Otilde(f^2 n^2/L) \cdot O(\log n/\varepsilon)^{f+2}$.
That means our construction has subquadratic space as long as the space requirement $\mathsf{S}_L$
of the input $f$-DSO for short hop-distances is subquadratic.
We would like to add some context to the restriction on the additive stretch $\beta$.
Assume for simplicity that $f$ is a constant.
In most cases in the literature where an $f$-DSO$^{\leq L}$ is used to build an $f$-DSO,
$L$ is of order $n^{\Theta(1/f)}$.
We also use such a value in this work.
In this case, our construction supports an additive stretch that can be as large as a small polynomial,
as long as it is asymptotically smaller than $\varepsilon^2 \nwspace n^{O(1/f)}$.
Conversely, for a fixed $\beta$, the restriction can be interpreted as bounding how fast
the approximation parameter $\varepsilon$ can approach $0$ (we do not require $\varepsilon$ to be constant).
In \Cref{thm:chained}, we have $\beta = 2$,
hence $\varepsilon = \omega(\sqrt{\log(n)}/n^{O(1/f)})$.
Finally, the derandomization requires $f \ge 4$ and $L = \Otilde(\sqrt{f^3 m/\varepsilon})$.
Even in the unfavorable case in which both $f$ and $\varepsilon$ are constants,
this allows for a cut-off parameter up to $\Otilde(\sqrt{m})$.
In very dense graphs, this is no restriction at all.

\subsection{Related Work}
\label{subsec:related_work}

\paragraph{Distance Oracles.}
Thorup and Zwick~\cite{ThorupZ05} showed that, for any positive integer $k$, 
any DO for undirected graphs with a multiplicative stretch strictly less than $2k + 1$ must take $\Omega(n^{1+1/k})$ bits off space,
assuming the Erd\H{o}s girth conjecture \cite{Erd64extremal}.
The lower bound only applies
to graphs that are sufficiently dense and to queries that involve pairs of neighboring vertices,
leading to several attempts to bypass it in different settings.
For example, there is a line of work on improved distance oracles for sparse graphs.
Porat and Roditty~\cite{PoratR11} showed that 
for unweighted graphs and any $\varepsilon>0$, one can construct a DO 
with multiplicative stretch $1{+}\varepsilon$ 
and query time $\widetilde O(m^{1-\frac{\varepsilon}{4+2\varepsilon}})$.
The space of the data structure is $O(nm^{1-\frac{\varepsilon}{4+2\varepsilon}})$,
which is subquadratic for $m = o(n^{1+\frac{\varepsilon}{4+\varepsilon}})$.
P\v{a}tra\c{s}cu, Roditty, and Thorup~\cite{Patrascu12NewInfinity}
obtained a series of DOs with fractional multiplicative stretches for sparse graphs.
For general dense graphs, P\v{a}tra\c{s}cu and Roditty~\cite{PatrascuRoditty14BeyondThorupZwick} 
devised a distance oracle  for unweighted graphs with stretch $(2,1)$ that has $O(1)$ query time, $O(n^{5/3})$ space, 
and can be constructed in time $O(mn^{2/3})$.
They also showed that $(\alpha,\beta)$-approximate DOs with $2\alpha+\beta < 4$
require $\Omega(n^2)$ space,
assuming conjecture on the hardness of set intersection queries.
Compared to the P\v{a}tra\c{s}cu and Roditty upper bound~\cite{PatrascuRoditty14BeyondThorupZwick},
Baswana, Goyal, and Sen~\cite{Baswana09Nearly2ApproximateJournal}
marginally increased the stretch to $(2,3)$ and space to $\Otilde(n^{5/3})$
in order to reduce the preprocessing time to $\Otilde(n^2)$.
The stretch was later reset again to $(2,1)$ by Sommer~\cite{sommer2016all}, 
keeping the improved time complexity.
A successive work by Knudsen~\cite{Knudsen17AdditiveSpanner} removed all additional poly-logarithmic factors in both the construction time and space.
	
Agarwal and Godfrey~\cite{AgarwalGodfrey13DOsStretchLessThan2,Agarwal14SpaceStretchTimeTradeoffDOs} 
investigated the possibility of constructing a distance oracle with a stretch less than 2,
albeit at the expense of slower query times.
They showed that, for any positive integer $\ell$ and any real number $c \in (0,1]$, it is possible to design a DO of size $O(m + n^{2-c})$
and multiplicative stretch $1 + \frac{1}{\ell}$.
The query time is $O((n^c\mu)^{\ell})$, where $\mu = \frac{2m}{n}$ is the average degree of the graph.
Furthermore, they also showed that the query time can be reduced to $O((n^c+\mu)^{2\ell-1})$ at the cost of a small additive stretch  $\frac{2\ell-1}{t}\, W$, with $W$being the maximum edge weight.
Though the constructions in~\cite{AgarwalGodfrey13DOsStretchLessThan2,Agarwal14SpaceStretchTimeTradeoffDOs} have a multiplicative stretch better than 2, 
their DOs have two main drawbacks.
The subquadratic space only holds for sparse graphs, 
and, while they achieve a very low stretch, the query time is super-linear for dense graphs.

Akav and Roditty~\cite{akav2020almost} proposed, for any $\varepsilon \in (0,\frac{1}{2})$,
an $O(m+n^{2-\Omega(\epsilon)})$-time algorithm that computes a DO with stretch $(2+\epsilon, 5)$ and $O(n^{11/6})$ space,
thus breaking the quadratic time barrier for multiplicative stretch below $3$.
Chechik and Zhang~\cite{ChechikZ:22} improved this by offering
both a DO with stretch $(2,3)$ that can be built in $\Otilde(m + n^{1.987})$ time
and a DO with stretch $(2+\varepsilon, c(\varepsilon))$ that can be built in $O(m + n^{\frac{5}{3}-\varepsilon})$ time,
where $c(\varepsilon)$ is exponential in $1/\varepsilon$.
Both data structures have space $\Otilde(n^{5/3})$ and a constant query time.

\paragraph{Distance Sensitivity Oracles.}
Most of the work on distance sensitivity oracles is about handling a very small number $f \in \{1,2\}$ of failures~\cite{DemetrescuT02, BeKa08,  DeThChRa08, DuanP09a, BeKa09, BaswanaK13, ChCo20, BCFS21SingleSourceDSO_ESA,   GrandoniVWilliamsFasterRPandDSO_journal, GuRen21, Ren22Improved}. Here, we focus on related work with sensitivity $f\geq 3$ as this is the setting of the second problem in this paper. 
In their seminal work, Weimann and Yuster~\cite{WY13} designed a randomized $f$-DSO for exact distances 
introducing a size-time trade-off that is controlled by a parameter $\alpha \in (0,1)$.
More precisely, their oracle w.h.p.\ has space  $\Otilde(n^{3-\alpha})$, query time of $\Otilde(n^{2-2(1-\alpha)/f})$,
and can be built in time $\Otilde(mn^{2-\alpha})$.
Van den Brand and Saranurak \cite{BrSa19} and Karczmarz and Sankowski \cite{KarczmarzS23} presented $f$-DSOs using algebraic algorithms.
However, their space requirement is at least quadratic and their query time is at least linear.
Duan and Ren~\cite{DuRe22} provided an alternative $f$-DSO for exact distances with $O(fn^4)$ space, $f^{O(f)}$ query time,
that is, constant whenever $f$ is a constant.
The preprocessing time for building their oracle is exponential in $f$, namely, $n^{\Omega(f)}$. 
Recently, Dey and Gupta \cite{dey2024nearly} developed an $f$-DSO for undirected graphs
where each edge has an integral weight from $\{1\ldots W \}$ with $O((c f \log(nW))^{O(f^2)})$ query time,
where $c>1$ is some constant.
It has near-quadratic space $O(f^4 n^2 \log^2 (nW))$.
A drawback of their oracles is again the preprocessing time of $\Omega(n^f)$,
and their oracle requires at least $\Omega(n^2)$ space.

When allowing approximation, the $f$-DSO of Chechik et al.~\cite{ChCoFiKa17}
guarantees a multiplicative stretch of $1+\eps$ with a space requirement of $O(n^{2+o(1)}\log W)$,
where $\eps >0$ $W$ is constant and $W$ is the weight of the heaviest edge.
Their oracle can handle up to $f = o(\log n / \log \log n)$ failures,
has a query time of $O(f^5\log n \log \log W)$, and  can be build in $O(n^{5+o(1)} (\log W)/\varepsilon^f)$ time.
In fact, the preprocessing time has recently been reduced to $O(mn^{2+o(1)}/\varepsilon^f)$~\cite{Bilo24ApproxDSOSubquadraticTheoretiCS}

Besides the general Thorup-Zwick bound~\cite{ThorupZ05} that assumes the girth conjecture,
they also showed \emph{unconditionally} that for undirected graphs with a multiplicative stretch better than $3$ must take $\Omega(n^2)$ bits of space.
This of course also applies in the presence of failures and is the reason why all the above $f$-DSOs have at least quadratic space. 
Like for distance oracles, there has been a line of work focusing on the design of $f$-DSOs with subquadratic space, sidestepping the bound.
These oracles must have stretch $(\alpha,\beta)$ with $\alpha+\beta\geq 3$ since a stretch $(\alpha, \beta)$ 
can also be stated as $(\alpha+\beta,0)$.
Chechik, Langberg, Peleg, and Roditty~\cite{CLPR12} designed an $f$-DSO that, for any integer parameter $k\geq 1$,  has space of $O(fkn^{1+1/k}\log(nW))$, query time $\Otilde(|F| \log\log d_{G-F}(s,t))$, and guarantees a multiplicative stretch of $(8k-2)(f+1)$.
Note that the stretch depends on the sensitivity parameter $f$. 
Recently, two $f$-DSOs with subquadratic space and stretch independent from $f$ have been developed.
The first one is by Bilò, Choudhary, Cohen, Friedrich, Krogmann, and Schirneck~\cite{Bilo23CompactDO} that can handle up to $f = o(\log n/\log\log n)$ edge failures and, for every integer $k \ge 2$, guarantees a stretch of $2k-1$.
The size and the query time depend on some trade-off parameter $\alpha \in (0,1/k)$.
They are equal to $kn^{1+\alpha + \frac{1}{k} +o(1)}$ and $\Otilde(n^{1 + \frac{1}{k} -\frac{\alpha}{k(f+1)}})$, respectively. 
The second $f$-DSO by Bilò et al.~\cite{Bilo24ApproxDSOSubquadraticTheoretiCS} works for unweighted graphs $G$ with unique shortest paths.
For every constants $f \ge 2$,  $0 < \gamma < \sfrac{1}{2}$, and any $\varepsilon > 0$ (possibly depending on $m$, $n$, and $f$),
their oracle has stretch $3 + \eps$,
space   $\Otilde(n^{2-\frac{\gamma}{f+1}}) \,{\cdot}\, O(\log n/\varepsilon)^{f+2}$,
query time $O(n^{\gamma}/\varepsilon^2)$,
  and preprocessing time
  $\Otilde(mn^{2-\frac{\gamma}{f+1}})
  	\cdot O(\log n/\varepsilon)^{f+1}$.

\subsection{Organization of the Paper}
\label{subsec:outline}
The next section provides an overview of the techniques we use to prove our main theorems. 
In \Cref{sec:prelims}, we recall the central definitions.
\Cref{sec:DOs} treats the distance oracles.
This includes the $(1{+}\varepsilon, 2W)$-stretch DO, 
as well as a hierarchy of DOs that trade a smaller size for a higher stretch.
It follows a two-step black-box reduction turning a (static) DO 
into a fault-tolerant distance sensitivity oracle.
The first step, in \Cref{sec:framework_DSO_short}, 
constructs an $f$-DSO$^{\leq L}$ for hop-short distances as an intermediate data structure.
Then, in \Cref{sec:framework_DSO}, the general $f$-DSOs is computed. 
The results in \Cref{sec:DOs,sec:framework_DSO} are initially phrased as random algorithms,
\Cref{sec:derandomization} derandomizes them.
The paper is concluded in \Cref{sec:chained}, 
where we combine the DO and the reduction to prove \Cref{thm:chained}.

\section{Technical Overview}
\label{sec:overview}

We now give a more detailed overview how we achieve the results stated above.
The sections corresponds to the two main parts of the paper.
The first is about our new construction of distance oracles.
The second part describes the framework that turns (static)
DOs into distance \emph{sensitivity} oracles.
This consists of two steps:
first obtain an $f$-DSO$^{\leq L}$ and then use if for the general $f$-DSO.
In the third part of the overview, we briefly sketch our derandomization approach.

\subsection{Improved Distance Oracles}
\label{subsec:overview_DO}

Our distance oracles introduce a small additive stretch
in order to make the multiplicative stretch arbitrarily small 
while, at the same time, keep the space subquadratic.
We assume the input graph to be undirected,
for it is known that subquadratic-space DOs are impossible for digraphs~\cite{ThorupZ05}.
The edges may have non-negative weights from a domain of polynomial size\footnote{%
	The domain size is such that any graph distance
	can be encoded in a constant number of machine words.
}
with maximum weight $W$.
A common pattern in the design of distance oracles is to designate a subset 
(or hierarchy of subsets) of vertices
called \emph{centers}~\cite{ThorupZ05}, \emph{landmark vertices}~\cite{AgarwalGodfrey13DOsStretchLessThan2} or \emph{pivots}~\cite{Chechik14,Chechik15}.
The data structure stores, for each vertex $v$ in the graph,
the distance from $v$ to all pivots.
Also, $v$ knows its closest pivot $p(v)$.
When given two query vertices $s,t \in V$, the oracle first checks whether $s$ and $t$
are sufficiently ``close'' to work out the exact distance $d(s,t)$.
The definition of ``close'' varies among the different constructions in the literature.
If $s$ and $t$ are instead ``far'' from each other, the estimate $d(s,p(s))+d(p(s),t)$ is returned,
which is at most $d(s,t) + 2 \nwspace d(s,p(s))$.
Since $s$ and $t$ are ``far'' compared to $d(s,p(s))$, 
the estimate has a good stretch.
This observation, of course, is in no way confined to the vertex $s$.
For any vertex $v$ on a shortest $s$-$t$-path,
$d(s,p(v))+d(p(v),t)$ incurs an error of at most $2 \nwspace d(v,p(v))$.
This gives some freedom on how much storage space and query time
one is willing to spend
on finding such a $v$ with a small distance to its closest pivot.  

Our twist to that method is to look at the vicinity of a vertex not in terms of a fixed radius,
but by an absolute bound on the number of considered vertices.
Namely, we define a cut-off value $K$ and store, for each vertex $v$, the $K$ nearest vertices,
regardless of the actual distance to $v$.
Choosing $\Otilde(n/K)$ pivots
ensures that every vertex has a pivot in its $K$-vicinity.
Storing the distances from every vertex to every pivot takes $\Otilde(n^2/K)$ space,
so $K$ is our saving over quadratic space.

One of two things can happen when searching along the shortest path from $s$ to $t$ 
for a suitable vertex $v$.
If all vertices in the list $K[v]$ have a small graph distance to $v$,
also the closest pivot $p(v)$ must be nearby.
Otherwise, there are elements in $K[v]$ that have a large graph distance,
which we can use to skip ahead in the path.
This sets up a win-win strategy.
Consider the auxiliary graph $H$ on the same vertex set as $G$
in which any vertex $v$ has an edge to each member of $K[v]$.
Given a query $(s,t)$ and an approximation parameter $\varepsilon > 0$,
we conduct a bidirectional breath-first search
in $H$ starting from both $s$ and $t$, trimming the search at hop-distance $1/\varepsilon$.
The two searches meeting in one or more vertices is our definition of $s$ and $t$ being ``close''.
We can then compute $d(s,t)$ exactly by minimizing $d(s,v)+d(v,t)$ over the intersecting vertices.
Otherwise, we prove that the reason why the searches remained disjoint
was that we could not skip ahead fast enough.
There must have been a vertex $v$ on the shortest $s$-$t$-path
for which all neighbors in $K[v]$ have a small distance to $v$, including $p(v)$.
We take ``small'' to mean at most  $\frac{\varepsilon}{2} \, d(s,t) + W$,
where $W$ is the maximum edge weight.
The sum $d(s,p(v))+ d(p(v),t)$ thus overestimates
the true distance by at most $2 \nwspace d(v,p(v))$, resulting in an $1+\varepsilon$
multiplicative stretch and $2W$ additive.

Spacewise, the bottleneck is to store all the distances between vertices and pivots.
In a second construction, we devise a way to further reduce the space,
at the cost of increasing both the multiplicative and additive stretch.
We are now looking for two vertices $u$ and $v$ on the $s$-$t$-path
that both have small distance to their respective pivots $p(u)$ and $p(v)$.
The portion of the distance between $p(u)$ and $p(v)$ is not stored directly but instead
estimated at query time by another, internal, distance oracle.
Since the inner data structure only needs to answer queries between pivots,
we can get a $2k-1$ multiplicative stretch (for this part) 
with only $\Otilde( (n/K)^{1+1/k})$ space.
This results in a hierarchy of new DOs with ever smaller space.

\subsection{From Hop-Short Distance Oracles to Distance Sensitivity Oracles}
\label{subsec:overview_hop-short}

An $f$-DSO is a data structure 
that receives a query $(s,t,F)$, 
consisting of vertices $s$ and $t$ as well as a set $F \subseteq E$ of at most $f$ edges,
and answers with an approximation of the $s$-$t$-distance $d(s,t,F)$
in the graph $G{-}F$.
Let $L$ be an integer cut-off parameter. 
A query $(s,t,F)$ is \emph{hop-short},
if $s$ and $t$ are joined by a path on at most $L$ edges in $G{-}F$.
An $f$-DSO for hop-short distances ($f$-DSO$^{\leq L}$) 
only guarantees a good stretch for hop-short queries.
Such oracles are used as stepping stones towards
general $f$-DSOs~\cite{AlonChechikCohen19CombinatorialRP,KarthikParter21DeterministicRPC,Ren22Improved,WY13}.

Recently, Bilò et al.~\cite{Bilo24ApproxDSOSubquadraticTheoretiCS}
presented a new approximate $f$-DSO for unweighted graphs with unique shortest paths taking only subquadratic space.
Implicit in their work is a pathway that takes the (static) distance oracle of Thorup and Zwick
with multiplicative stretch $3$
and makes it fault-tolerant increasing the stretch to $3 + \varepsilon$.
The parameter $\varepsilon > 0$ influences
the space, query time, and preprocessing time of the data structure.
The transformation has two major steps.
In the first one, the DO is used to build a hop-short $f$-DSO$^{\leq L}$.
The second step then combines the answer for hop-short paths into
good approximations for arbitrary queries.
Their ad-hoc method is highly tailored towards the DO of Thorup and Zwick~\cite{ThorupZ05}
and does not readily generalize.

We take the same two-step approach
but give an entirely new construction for the hop-short $f$-DSO$^{\leq L}$.
This is necessary to make it compatible with other distance oracles.
Our stand-alone framework
works with any path-reporting DO as a black box.
The input oracle can have an arbitrary multiplicative stretch $\alpha$
and may even have an additive component $\beta$,
requiring new techniques.
The resulting distance sensitivity oracle has stretch $(\alpha (1{+}\varepsilon),\beta)$.
The key feature of the reduction is that, as long as the input DO has subquadratic space, 
so does the $f$-DSO.
\vspace*{.5em}

\noindent
\textbf{Fault tolerance for hop-short paths.}
Naively, an oracle with sensitivity $f$ must be able to handle $O(n^2 \nwspace m^f)$ queries,
one for each pair of vertices and any set $F$ of up to $f$ edge failures.
Clearly, not every failure set is relevant for every pair of vertices.
As a first key tool we use \emph{fault-tolerant trees} (FT-trees) to zero in on the relevant queries.
The were originally introduced by Chechik et al.~\cite{ChCoFiKa17}. 
We combine it with hop-short paths.
There is a tree $FT(s,t)$ for every pair of vertices $s$ and $t$
whose shortest path in the original graph $G$ has at most $L$ edges.
In the root that path is stored.
For any edge $e_1$ along that path, the root has a child node 
which in turn stores a shortest $s$-$t$-path in the graph $G-\{e_1\}$.
This corresponds to failing $e_1$ and looking for a \emph{replacement path} $P(s,t, \{e_1\})$.
The construction is iterated for each child node until depth $f$ is reached.
That means, for any node $x$ at level $k$ of $FT(s,t)$, let $F_x = \{e_1, \ldots, e_k\}$ be the edges associated with the path the root to $x$.
The node $x$ stores a shortest $s$-$t$-path $P_x$ in $G{-}F_x$.
If the shortest $s$-$t$-replacement path has more than $L$ edges, $x$ is made a leaf,
marked by setting $P_x = \perp$.

To handle a query $(s,t,F)$, in any node $x$ starting with the root, 
it is checked whether $E(P_x) \cap F \neq \emptyset$.
If so, the search continues with the first child node associated with an edge in $E(P_x) \cap F$.
Otherwise, there are two cases.
Either $x$ stores some shortest $s$-$t$-path disjoint from $F$ or $P_x = \perp$.
In the first case, the length $|P_x|$ is the desired \emph{replacement distance} $d(s,t,F)$.
In the second, it is enough to report that $d(s,t,F) > L$.
This reduces the number of relevant queries to $O(n^2 L^f)$.

The second key tool is an
$(L,f)$-\emph{replacement path covering} (RPC)~\cite{WY13,KarthikParter21DeterministicRPC}.
This is a family of $\G$ subgraphs of $G$ such that 
for any set $F$ of at most $f$ edges and pair of vertices $s$ and $t$ 
for which there is a replacement path $P(s,t,F)$ of at most $L$ edges,
there exists a subgraph $G_i \in \G$ such that  $E(G_i)\cap F = \emptyset$
and $G_i$ contains $P(s,t,F)$.
RPCs are common in the design of $f$-DSO$^{\leq L}$
because, if one can find $G_i$ quickly,
the replacement distance is just $d_{G_i}(s,t)$.
Recently, Karthik and Parter~\cite{KarthikParter21DeterministicRPC}
gave a construction with $O(f L \log n)^{f+1}$ graphs.
We shave an $(fL\log n)$-factor from that.
In the full version of~\cite{KarthikParter21DeterministicRPC},
they give an algorithm that takes a list $(F_1, P_1), \dots (F_\ell, P_\ell)$ of
pairs of edge sets, with $|F_k| \le f$ and $|P_k| \le L$.
It computes a family of subgraphs that only for pairs $(F_k,P_k)$ from the list 
guarantees some $G_i$ avoiding $F_k$ but containing $P_k$.
We use nodes of the fault-tolerant trees to compute a list that allows us to
cover all relevant queries but get a smaller family 
with only $O(f L \log n)^f$ graphs.
This improvement may be of independent interest.

Unfortunately, in the subquadratic-space regime, 
both the FT-trees as well as the graphs of the replacement-path covering are too large.
The former have size $O(n^2 L^{f+1})$ (each node stores up to $L$ edges),
and the latter $O(m) \cdot O(f L \log n)^f$.
Instead, we build a data structure that \emph{emulates} queries in the FT-trees
without actually storing them.
The graphs in the RPC are replaced by a distance oracles.
When using an $(\alpha, \beta)$-approximate DO, this gives a saving over quadratic space.
The main difficulty is that the emulated query procedure must follow the exact same path as in the actual (discarded) tree $FT(s,t)$ in order to find the correct distance oracle.
\vspace*{.5em}

\noindent
\textbf{From hop-short to general distance sensitivity oracles.}
The second step of the transformation takes the $f$-DSO$^{\le L}$ for hop-short distances
and combines it with techniques to stitch together the answer for hop-long paths
from those for hop-short ones.
This step follows the pathway in~\cite{Bilo24ApproxDSOSubquadraticTheoretiCS} more closely.
The key differences are the introduction of what we call \emph{pivot trees}
as well as a more thorough analysis that is needed to deal with additive stretch.

We reuse the idea of FT-trees but, as before, they are too large to be stored.
Above, we emulated a query without access to the actual trees.
Here, we do store some of the trees but in smaller versions  and not for all pairs of vertices.
Since we must also handle hop-long queries,
any node of an FT-tree may now store a path of up to $n$ edges.
If we were to create a child node for each edge, 
we would end up with a prohibitive space of $O(n^{f+1})$ for a single tree.
Instead, we use larger segments and each child node now corresponds to the failure of a whole segment.
Of course, this has the danger that failing a segment may inadvertently destroy replacement paths
that would still exists in $G-F$, where only the failing edges are removed.
We have to strike a balance between the space reduction of larger segments
and the introduced inaccuracies.
This leads to the concept of granularity.
In an FT-tree with granularity $\lambda$, the first and last $\lambda/2$ edges form their own segments.
Beyond that, we use segments whose size increase exponentially towards the middle of the stored path.
The base of this exponential is $1+ \Theta(\varepsilon)$,
where $\varepsilon > 0$ is the approximation parameter.
A higher granularity means much larger trees but with higher accuracy.

We offset this by building the larger trees only for very few pairs of vertices.
We sample a set $B_1$ of $\Otilde(fn/L^f)$ pivots and build an FT-tree 
with granularity $\lambda = \Theta(\varepsilon L)$ for each pair of vertices in $B_1 \times B_1$.
For some vertex $v$ and failure set $F$, let $\ball_{G-F}(v,\lambda/2)$
be the ball of hop-distance $\lambda/2$ around $v$ in $G{-}F$.
We show that if both $s$ and $t$ have respective pivots 
$p_s \in \ball_{G-F}(s,\lambda/2) \cap B_1$ and $p_t \in \ball_{G-F}(t,\lambda/2) \cap B_1$,
then the FT-tree $FT_{\lambda}(p_s,p_t)$
with granularity $\lambda$ gives a very good estimate of the replacement distance $d(s,t,F)$.
However, since there are so few pivots in $B_1$ this can only be guaranteed
if the balls around $s$ and $t$ are very dense.
We also have to give a fall-back solution in case one of the balls is sparse.
We sample a set $B_2$ of now $\Otilde(fn/L)$ pivots (much more than $B_1$)
and build an FT-tree \emph{without} granularity 
for pairs of vertices in $B_2 \times V$.
These trees are much smaller and we show that together with the $f$-DSO$^{\le L}$ 
they are still sufficiently accurate.

The problem is to quickly find the pivots that are close to $s$ and $t$ in $G{-}F$ at query time.
Simply scanning over $B_1$ and $B_2$ is way too slow.
Instead, we use yet another kind of tree data structure, \emph{pivot trees}.
They are inspired by FT-trees but are not the same.
We have one pivot tree per vertex in $V$ (instead of an FT-tree for every pair).
The tree belonging to $s$ stores in each node $x$
some shortest path $P_x$ starting from $s$ as before,
and each child of node $x$ represents the failure of one edge of $P_x$.
The key difference is that the other endpoint of $P_x$ can vary now,
the path ends in the closest pivot $p_s \in B_1$ of the first type.
Note that in the graph $G{-}F_x$ that pivot may not be the same as in $G{-}(F_x \cup \{e_{k+1}\})$.
We only care about the vicinity $\ball_{G-F}(s,\lambda/2)$,
so the paths $P_x$ have at most $\lambda/2$ edges.
If the closest pivot in $B_1$ is too far away, $\ball_{G-F}(s,\lambda/2)$ must be sparse.
That means, also $\ball_{G-F}(s,\lambda/2) \cap B_2$ is small and it is feasible to store all pivots of the second type.

Unfortunately, this is still not enough.
Due to our sparsification via segments with multiple edges,
the answer of the FT-trees with and with out granularity are only accurate 
if the replacement path $P(s,t,F)$ is ``far away'' from all failures in $F$.
That roughly means that any vertex of $P(s,t,F)$ is more than an $\Theta(\varepsilon)$-multiple 
from any endpoint of some edge in $F$.
If this safety area is free of failures, no segment can accidentally contain an edge of $F$, 
which would disturb the return value of the FT-tree too much.
If, however, there is a failure too close to the path,
Chechik et al.~\cite{ChCoFiKa17} showed that there is always some surrogate target 
$t' \in V(F)$ such that the replacement path $P(s,t',F)$
is indeed far away from all failures
and also not much of a detour compared to going directly from $s$ to $t$.
This causes an $1+\varepsilon_1$ factor in the multiplicative part of the stretch,
where $\varepsilon_1 = \Theta(\varepsilon)$.
We build an auxiliary \emph{weighted} complete graph $H$
on the vertex set $V(F) \cup \{s,t\}$ to exploit this detour structure.
We use the FT-trees and $f$-DSO$^{\le L}$ to compute the weight of all edges in the graph.
The eventual answer of our oracle is the $s$-$t$-distance in $H$.

The main obstacle is to prove that $d_H(s,t)$
is actually an $(\alpha(1{+}\varepsilon), \beta)$-approximation of $d(s,t,F)$. 
This is also the other decisive difference to~\cite{Bilo24ApproxDSOSubquadraticTheoretiCS}
in that we need to handle both the multiplicative and the additive stretch.
Consider an edge $\{u,v\} \in E(H)$.
If $P(u,v,F)$ is hop-short, the $f$-DSO$^{\le L}$ trivially gives an $(\alpha,\beta)$-approximation
of $d(u,v,F)$ which we use to compute the weight $w_H(u,v)$.
If $P(u,v,F)$ has more than $L$ edges and is ``far away'' from all failures,
we show that $w_H(u,v)$ computed by the FT-trees is only an $(\alpha, X\beta)$-approximation.
The blow-up $X = O(f \log (n)/\varepsilon)$ stems from the segments of increasing size.
This is only exacerbated by the fact that the shortest $s$-$t$-path in the weighted graph $H$
has $O(f^2)$ edges, each one contributing their own distortion to the additive stretch.
However, the \emph{additive} blow-up only happens if $P(u,v,F)$ has many edges,
that is, if $d(u,v,F)$ is large.
The idea is to then charge most of the additive stretch to the multiplicative part, increasing
it only by another $1+\varepsilon_2 = 1+ \Theta(\varepsilon)$ factor
which is then combined with the $1+\varepsilon_1$ stemming from the auxiliary graph $H$.
To make this work, we carefully analyze the interplay of the edges in $H$,
using an induction over $E(H)$ in the order of increasing replacement distance $d(u,v,F)$.

\subsection{Derandomization}
\label{subsec:overview_derand}

For the derandomization, we use the framework of critical paths proposed by Alon, Chechik, and Cohen~\cite{AlonChechikCohen19CombinatorialRP}.
It consists of identifying a small set of shortest paths in $G$,
(greedily) compute a deterministic hitting set for them,
and then build the fault-tolerant data structure from there.
The number of paths needs to be small in order to keep the derandomization efficient.
This is the main difficulty in the approach.
Our threshold of efficiency is that making the data structures deterministic
should not increase their preprocessing time by more than a constant factor.

In the context of derandomization, we view paths as mere sets of vertices without any further structure.
This allows us to apply the approach also to other subsets of $V$ in a unified fashion.
Consider the (static) distance oracle for weighted graphs in \Cref{thm:distance_oracle,thm:hierarchy_DO}.
The list $K[v]$,
containing the $K$ vertices closest to $v$,
are key components in the construction. 
It is indeed enough to hit all the $\{K[v]\}_{v \in V}$ to derandomize the DO.
This also incurs hardly any extra work as those lists are compiled anyway in the preprocessing.

The construction of the general $f$-DSO 
from the one restricted to hop-short distances (\Cref{thm:general_DSO})
builds on the pivot sets $B_1$ and $B_2$ whose derandomization is more involved.
Recall that we use $\lambda$ for the granularity of the FT-trees.
During the proof of correctness, it becomes apparent that the pivots of the second type in $B_2$
need to hit the length-$\lambda/2$ prefixes and suffixes of all
concatenations of up to two replacement paths, provided that those concatenations are hop-long 
(have more than $L$ edges).
We use a structural result by Afek, Bremler-Barr, Kaplan, Cohen, and Merritt~\cite{Afek02RestorationbyPathConcatenation_journal} 
that states that replacement paths themselves are concatenations of $O(f)$ shortest paths in the 
original graph $G$.
This allows us to show that computing a deterministic hitting set of all shortest paths in $G$
of length $\Omega(\lambda/f)$ suffices to guarantee the same covering properties as $B_2$.
The set $B_1$, in turn, need to hit all the sets $\ball_{G-F}(u,\lambda/2)$ that are sufficiently dense
(more than $L^f$ vertices).
In principle, a similar approach as for $B_2$ would work
but that would either produce way too many pivots or take much to long.
Instead, we interleave the level-wise construction of the pivot trees with derandomization phases
to find a deterministic stand-in for $B_1$.

\section{Preliminaries}
\label{sec:prelims}

We let $G = (V,E)$ denote the base graph with $n$ vertices and $m$ edges.
We tacitly assume $m \ge n$.
Depending on the setting $G$ is either directed or undirected, may be edge-weighted by non-negative reals or unweighted,
as stated in the respective sections. 
In case $G$ is weighted and undirected, we assume that the weight function $w$
maps into an subset of the interval $[1, W]$ of at most $m$ elements.
This does not lose generality as we can contract (undirected) weight-zero edges 
and scaling all other weights by $1/\min_{e \in E} w(e)$.
Note that any distance in the graph can be encoded in a constant number
of machine words on $O(\log n)$ bits.
This is the reason why none of the asymptotic bounds depend on $W$.

For a graph $H$, which may differ from the input $G$,
we denote by $V(H)$ and $E(H)$ the set of its vertices and edges, respectively.
Let $P = (u_1, \dots, u_i)$ and $Q = (v_1, \dots, v_j)$ be two paths in $H$.
Their \emph{concatenation} is $P \circ Q = (u_1, \dots, u_i, v_1, \dots, v_j)$,
which is well-defined if $u_i = v_1$ or $\{u_i,v_1\} \in E(H)$.
For vertices $u,v \in V(P)$, we let $P[u..v]$ denote the subpath of $P$ from $u$ to $v$.
The \emph{length} of path $P$ is $w(P) = \sum_{e \in E(P)} w(e)$.
We use $|P| = |E(P)|$ for the number of edges.
For $s,t \in V(H)$, the \emph{distance} $d_H(s,t)$ 
is the minimum length of any $s$-$t$-path in $H$;
and $d_H(s,t) =+ \infty$ if no such path exists.
We drop the subscripts for the input graph $G$.

A \emph{spanning} subgraph of a graph $H$
is one with the same vertex set as $H$ but possibly any subset of its edges.
(This should not be confused with a spanner.)
Let $\alpha,\beta$ be two non-negative real numbers with $\alpha \ge 1$.
A \emph{distance oracle} (DO) for $H$ is a data structure
that reports, upon query $(s,t)$, an approximation of the distance $d_H(s,t)$.
It has \emph{stretch} $(\alpha,\beta)$ if the reported value $\widehat{d}_H(s,t)$
satisfies $d_H(s,t) \le \widehat{d}_H(s,t) \le \alpha \cdot d_H(s,t) + \beta$.
We say the stretch is \emph{multiplicative} if $\beta = 0$,
\emph{additive} if $\alpha = 1$
and \emph{near-additive} if $\alpha = 1{+}\varepsilon$
for some parameter $\varepsilon > 0$
that can be made arbitrarily small.

For a set $F \subseteq E$ of edges,
let $G{-}F$ be the graph obtained from $G$ by removing all edges in $F$.
For any two $s,t \in V$, a \emph{replacement path} $P(s,t,F)$ 
is a shortest path from $s$ to $t$ in $G{-}F$. 
Its length $d(s,t,F) = d_{G-F}(s,t)$ is the \emph{replacement distance}.
For a positive integer $f$, an $f$-\emph{edge fault-tolerant distance sensitivity oracle}
($f$-DSO)
answers queries $(s,t,F)$, with $|F| \le f$, with an approximation of the replacement distance $d(s,t,F)$.
The stretch of an $f$-DSO is defined as for DOs.
The maximum number $f$ of supported failures is the \emph{sensitivity}.
Let $L \le n$ be a positive integer parameter.
We also discuss \emph{$f$-DSO for hop-short distances} bounded by $L$ ($f$-DSO$^{\leq L}$).
On query $(s,t,F)$, their estimate $\widehat{d^{\le L}}(s,t,F)$
must always satisfy $\widehat{d^{\le L}}(s,t,F) \ge d(s,t,F)$.
The upper bound of the stretch requirement
$\widehat{d^{\le L}}(s,t,F) \le \alpha \cdot d(s,t,F) + \beta$
only needs if there is a shortest path from $s$ to $t$ in $G{-}F$ with at most $L$ edges.

We measure the space complexity of distance (sensitivity) oracles and all other data structures
in the number of $O(\log n)$-bit words.
The size of the graph $G$ does not count against the space, unless it is stored explicitly.

\section{Distance Oracles for Static Graphs}
\label{sec:DOs}

We now present the construction of our (static) distance oracles in subquadratic space,
trading a small additive stretch for an improved multiplicative one.
We first discuss a near-additive oracle 
with a multiplicative stretch $1+\varepsilon$ and additive stretch $2W$.
Upon request, it also reports the approximate shortest path,
which we use in \Cref{sec:framework_DSO_short}.
It has a space of $\Otilde(n^2/K)$ for some integer parameter $K$.
We then show how to reduce the space to $\Otilde((n/K)^{1+1/k})$ at the expense
of raising the multiplicative stretch to $2k-1+\varepsilon$,
and the additive to $4kW$.

\subsection{Near-Additive Distance Oracle}
\label{subsec:DOs_near-additive}

For convenience, we restate \Cref{thm:distance_oracle} below.
In this section we prove a randomized version in which the features 
of the oracle are only guaranteed with high probability.
The derandomization is deferred to \Cref{subsec:derand_DO}.
We first treat the space, query time, and stretch of the DO.
Subsequently, we describe how to make the oracle path-reporting.

\distanceoracle*

\noindent
\textbf{Preprocessing and space.}
Let $G$ denote the underlying graph.
First, we make a simplifying assumption.
It is easy to compute the connected components of $G$ in time $O(m)$ (a fortiori in APSP time)
and store an $O(n)$-sized table of the component IDs for all vertices.
This allows us to check in constant time whether two vertices have a path between them,
and otherwise correctly answer $d(s,t) = +\infty$.
We can thus assume to only receive queries for vertex pairs in the same component.

\begin{algorithm}[t]
\setstretch{1.33}
\vspace*{.25em}
    $A_{1/\varepsilon}(s) \gets \{v \in V \mid d^{\text{hop}}_H(s,v) \le \lceil 1/\varepsilon \rceil\}$\; 
    $A_{1/\varepsilon}(t) \gets \{v \in V \mid d^{\text{hop}}_H(v,t) \le \lceil 1/\varepsilon \rceil\}$\;  
   $\widehat{d}_1 \gets \min \{d_{1/\varepsilon}(s,v) + d_{1/\varepsilon}(v,t) \mid v \in A_{1/\varepsilon}(s) \cap A_{1/\varepsilon}(t)\}$\; \label{line:d1_minimum}
    $\widehat{d}_2 \gets \min \{d(s,p(v)) + d(p(v),t) \mid v \in A_{1/\varepsilon}(s) \cup A_{1/\varepsilon}(t)\}$\; \label{line:d2_minimum}
    \Return $\min(\widehat{d}_1,\widehat{d}_2)$\; \label{line:global_minimum}
\caption{Query algorithm of the DO in \Cref{thm:distance_oracle} for the query $(s,t)$.
	$d^{\text{hop}}_H$ is the hop-distance in $H$,
	$d_{1/\varepsilon}$ is the minimum length of all
	paths with at most $\lceil 1/\varepsilon \rceil$ hops in $H$,
	and $p(v) \in B$ is the pivot closest to $v$ in $G$.\vspace*{.25em}}
\label{alg:query_DO}
\end{algorithm}

In APSP time, we compute, for every vertex $v \in V$, the list $K[v]$ of its $K$ closest vertices
in $G$ (including $v$ itself) where ties are broken arbitrarily.
Each element of $K[v]$ is annotated with its distance to $v$.
We also sample a set $B$ of vertices, called \emph{pivots}, by including any vertex in $B$
independently with probability $C (\log n)/K$ for a sufficiently large constant $C > 0$.
It is easy to show using Chernoff bounds that with high probability
$B$ has $O(\frac{n}{K} \log n)$ elements.
We use $p(v)$ to denote the pivot closest to $v$.
Our data structure stores, for each vertex $v$, the list $K[v]$, the closest pivot $p(v)$,
and the distance $d(v,p(v))$.
Moreover, for each pivot $p \in B$, it additionally stores the distance from $p$ 
to \emph{every} vertex in $G$.
The data structure takes space $O(|B|n + nK) = O(\frac{n^2}{K} \log n + nK)$,
which is $\Otilde(n^2/K)$ for $K = O(\sqrt{n})$.
\vspace*{.5em}

\noindent
\textbf{Query algorithm.}
We use an auxiliary graph $H$ to simplify the presentation of the query algorithm
and the subsequent reasoning.
The graph $H$ has the same vertex set as $G$ and, for each $v \in V$,
an edge from $v$ to any $v' \in K[v]{\setminus}\{v\}$ whose weight is $d(v,v')$.
The \emph{hop-distance} between two vertices $u,v \in V$ in $H$,
is the minimum number of edges on any $u$-$v$-path.

\Cref{alg:query_DO} summarizes how a query $(s,t) \in V^2$ is processed.
Let $A_{1/\varepsilon}(s), A_{1/\varepsilon}(t)$ be the sets of vertices
with hop-distance at most $\lceil 1/\varepsilon \rceil$ from $s$ and $t$, respectively, 
in the graph $H$.
To compute the set $A_{1/\varepsilon}(s)$ we use a slightly modified breath-first search from $s$ (analogously for $A_{1/\varepsilon}(t)$ starting from $t$).
The modification consists in exploring, in each step with current vertex $v$,
\emph{all} neighbors in $K[v]{\setminus}\{v\}$ 
as long as fewer than $\lceil 1/\varepsilon \rceil$ hops have been made.
In particular, the search revisit vertices that have been encountered earlier.
Each time a vertex $v$ is visited,
its estimate of the distance $d_H(s,v)$ is updated to the minimum over all paths explored so far.
In other words, for each $v \in A_{1/\varepsilon}(s)$
the length of the shortest path from $s$ that uses at most $\lceil 1/\varepsilon \rceil$ edges is computed.
We use $d_{1/\varepsilon}(s,v)$ to denote this distance.
Note that $d_{1/\varepsilon}(s,v)$ may overestimate the true distance $d_H(s,v)$ in $H$,
which in turn overestimates $d(s,v)$ in $G$.
Below, we identify some conditions under which the estimate is accurate.
The sets $A_{1/\varepsilon}(s), A_{1/\varepsilon}(t)$ have $O(K^{\lceil 1/\varepsilon \rceil})$ elements and, with the information stored by the DO,
the modified BFSs in $H$ can be emulated in time $O(K^{\lceil 1/\varepsilon \rceil})$.

Recall that $p(v)$ is the pivot closest to $v$
and that the values $d(s,p(v)), d(p(v),t)$ are stored.
The oracle uses $A_{1/\varepsilon}(s), A_{1/\varepsilon}(t)$ to compute
\begin{equation*}
	\widehat{d_1} = \min_{v \in A_{1/\varepsilon}(s) \cap A_{1/\varepsilon}(t)} 
		d_{1/\varepsilon}(s,v) + d_{1/\varepsilon}(v,t)
		\quad\text{and}\ \quad
	\widehat{d_2} = \min_{v \in A_{1/\varepsilon}(s) \cup A_{1/\varepsilon}(t)} d(s,p(v)) + d(p(v),t).
\end{equation*}
It returns the smaller of the two estimates.
The total query time is $O(K^{\lceil 1/\varepsilon \rceil})$.
\vspace*{.5em}

\noindent
\textbf{Stretch.}
For a radius $r$ and vertex $v \in V$, let $\ball(v,r)$ be the set of all vertices that have distance at most $r$ from $v$.
Define $V_{r} = \{v \in V \mid \ball(v,r) \subseteq K[v] \nwspace \}$, 
that is, the set of all vertices $v$
that have at most $K$ vertices within distance $r$.
Recall that for any query $(s,t)$ to the DO, 
we can assume that there is a path between $s$ and $t$ in $G$.
In the following lemma, 
the additive term $W$ can be replaced by the largest weight $w$ of an edge on the path $P$.
This would reduce the stretch in \Cref{lem:correctness-weighted} to $(1{+}\varepsilon, 2w)$.

\begin{lemma}
\label{lem:property-H}
	Consider two vertices $s,t \in V$ with distance $d(s,t)$ and set 
	$r = \frac{\varepsilon}{2} \nwspace d(s,t)+W$.
	Let $P$ be a shortest path between $s$ and $t$ in $G$. 
	\begin{enumerate}
		\item[(i)] If all vertices of $P$ lie in $V_r$ the hop-distance between $s$ and $t$
			in $H$ is at most $\lceil 2/\varepsilon \rceil$.
		\item[(ii)] If all vertices of $P$ lie in $V_r$, 
			then $V(P) \cap A_{1/\varepsilon}(s) \cap A_{1/\varepsilon}(t)$
			is non-empty.
		\item[(iii)] If $P$ contains a vertex from $V{\setminus}V_r$, 
		then $(V(P)\setminus V_r) \cap A_{1/\varepsilon}(s)$ or
		 $(V(P)\setminus V_r) \cap A_{1/\varepsilon}(t)$ is non-empty.
	\end{enumerate}
\end{lemma}

\begin{proof}
	Let $d$ abbreviate the distance $d(s,t)$.
	First, consider the case that all vertices of $P$ lie in $V_{r} = V_{\frac{\varepsilon d}{2}+W}$.
	We define a sequence $\sigma_s$ of vertices in $V(P)$,
	namely, a subsequence of $P$ when directed away from $s$.
	Its first vertex is $x_0 = s$;
	each consecutive vertex $x_{i+1}$ is the element of $V(P[x_{i},t]) \cap \ball(x_i,r)$	
	that maximizes $d(x_i,x_{i+1})$.
	That means $x_{i+1}$ comes after $x_i$ on $P$ when going from $s$ to $t$,
	but it has distance at most $\frac{\varepsilon d}{2}+W$ from $x_i$.
	By the assumption $V(P) \subseteq V_r = \{v \in V \mid \ball(v,r) \subseteq K[v] \nwspace \}$,
	the sequence $\sigma_s$ corresponds to an actual sequence of hops in the auxiliary graph $H$.
	Moreover, for each $x_i$, we have
	$d_{1/\varepsilon}(s,x_i) \le \sum_{j=0}^{i-1} d(x_j,x_{j+1}) = d(s,x_i)$.
	So in this case the estimate $d_{1/\varepsilon}(s,x_i)$ is exact.
	
	Observe that consecutive vertices in $\sigma_s$ (except possibly the last pair involving $t$)
	have graph distance in $G$ of more than $\varepsilon d/2$.
	Indeed, if we had $d(x_i,x_{i+1}) \le \varepsilon d/2$,
	then the next vertex $v$ on $P$ that comes after $x_i$
	has a higher distance, but still satisfies
	$d(x_i,v) = d(x_i,x_{i+1}) + w(x_{i+1},v) \le \frac{\varepsilon d}{2}+W$.
	This shows that $\sigma_s$ reaches from $s$ to $t$ in at most $\lceil 2/\varepsilon\rceil$ hops.
	
	Symmetrically, we define the sequence $\sigma_t$ by perceiving $P$ as directed away from $t$.
	The first vertex is $y_0 = t$ and $y_{i+1}$ maximizes $d(y_i,y_{i+1})$
	over $V(P[s,y_{i}]) \cap \ball(y_i,r)$.
	Let $k$ be the smallest index in $\sigma_s$ such that $d(s,x_{k}) \ge d/2$.
	Then, we have $k \le \lceil 1/\varepsilon \rceil$
	by the above observation on the consecutive distances.
	Analogously, let $\ell$ be the minimum index in $\sigma_t$ with $d(t,y_{\ell}) \ge d/2$.
	The predecessor of $y_{\ell}$ thus satisfies $\frac{d}{2} < d(s,y_{\ell-1}) \le \frac{d}{2} + r$
	and therefore $d(x_k,y_{\ell-1}) \le r$.
	This shows that $x_k$ can be reached from $t$ in $H$ via $\ell \le  \lceil 1/\varepsilon \rceil$
	hops by first following $\sigma_t$ until $y_{\ell-1}$ and then hopping to $x_k$.
	In particular, we have $x_k \in V(P) \cap A_{1/\varepsilon}(s) \cap A_{1/\varepsilon}(t)$.
	
	The last part, where the path $P$ contains a vertex from $V{\setminus}V_r$,
	is structurally similar but somewhat simpler.
	Let vertex $z$ be the minimizer of $\min( d(s,z), d(z,t))$ in $V(P){\setminus}V_r$.
	W.l.o.g.\ $z$ is closer to $s$ than to $t$, whence $d(s,z) \le  d/2$.
	We now define the sequence $\sigma_{s}$ as above but let it end in $z$ instead of $t$.
	Note that, by the minimality of $z$, all vertices of $\sigma_{s}$ except for $z$ itself
	are in $V_r$.
	The same argument as above now shows that $\sigma_{s}$ 
	has at most $\lceil 1/\varepsilon \rceil$ vertices, which correspond to hops in $H$.
	That means, $z \in (V(P)\setminus V_r) \cap A_{1/\varepsilon}(s)$.
\end{proof}

To prove the approximation guarantee,
we use the following straightforward application of Chernoff bounds:
with high probability all vertices $v \in V{\setminus}V_r$ satisfy $d(v,p(v)) \le r$.
For if $\ball(v,r)$ contains more than $K$ elements, it also has a pivot w.h.p.

\begin{lemma}
\label{lem:correctness-weighted}
	With high probability over all queries $(s,t) \in V^2$,
	the distance oracle returns an estimate of $d(s,t)$ with stretch $(1+\varepsilon,2W)$.
\end{lemma}

\begin{proof}
	The oracle correctly answers $+\infty$ if $s$ and $t$ are in different component.
	Let again $P$ be a shortest $s$-$t$-path, $d = d(s,t)$, and $r = \frac{\varepsilon d}{2} + W$.
	First, note that the returned value is never smaller than $d$ since 
	$\widehat{d_1} = d_{1/\varepsilon}(s,v)+d_{1/\varepsilon}(v,t) \ge d(s,v)+ d(v,t) \ge d(s,t)$
	and $\widehat{d_2} = d(s,p(v)) + d(p(v),t)$ even corresponds to an actual path between $s$ and $t$.
	
	First, assume that all vertices of $P$ are in $V_r$.
	There exists some vertex $v \in V(P) \cap A_{1/\varepsilon}(s) \cap A_{1/\varepsilon}(t)$.
	The returned distance is exact since
	\begin{equation*}
		\widehat{d_1} \le d_{1/\varepsilon}(s,v) + d_{1/\varepsilon}(v,t) = d(s,v) + d(v,t) = d.
	\end{equation*}
	The first equality was argued in \Cref{lem:property-H}, the second one is due to $v$
	being on a shortest $s$-$t$-path.
	
	If $V(P) \nsubseteq V_r$, there exists some 
	$v \in (V(P){\setminus}V_r) \cap (A_{1/\varepsilon}(s) \cup A_{1/\varepsilon}(t))$,
	from which we get w.h.p.
	\begin{equation*}
		\widehat{d}(s,t) \le \widehat{d_2} \le d(s,p(v)) + d(p(v),t)
			\le d(s,v) + d(v,t) + 2 \cdot d(v,p(v)) \le d + 2r
			= (1{+}\varepsilon) \nwspace d + 2W. \qedhere
	\end{equation*} 
\end{proof}

\noindent
\textbf{Reporting paths.}
Only small adaptions are needed to make the oracle path-reporting.
Recall that we use an emulated BFS to compute the distance $d_{1/\varepsilon}(s,v)$
for all $v \in A_{1/\varepsilon}(s)$ (length $d_{1/\varepsilon}(v,t)$
for $v \in A_{1/\varepsilon}(t)$) by iteratively updating the current best length
of a path in $H$ with at most $\lceil 1/\varepsilon \rceil$ edges.
A path through $H$ may not correspond to a path through $G$,
namely, if it uses and edge $\{v,v'\}$ with $v' \in K[v]$
that is not actually present in $G$.
However, any shortest $v$-$v'$-path in $G$ exclusively uses vertices from $K[v]$.

We thus store $K[v]$ in the form of a shortest-path tree in $G$ rooted at $v$
that is truncated after $K$ vertices have been reached.
In each step of the emulated BFS, we explore the neighborhood $K[v]$
in the order given by an actual BFS of the shortest-path tree.
Each time some estimate $d_{1/\varepsilon}(s,v')$ is updated
we also store a pointer to the last vertex from which we reached $v'$.
Furthermore, with each distance $d(v,p)$ for an arbitrary pivot $p \in B$,
we store the first edge on a shortest path from $v$ to $p$.
This does not change the space requirement of the oracle by more than a constant factor.

Suppose the minimum reported in \cref{line:global_minimum} of \Cref{alg:query_DO}
is attained by $\widehat{d_1}$ and let $v$ be the minimizing vertex in \cref{line:d1_minimum}.
Following the stored pointers backwards reconstructs a shortest $s$-$v$-path through $A_{1/\varepsilon}(s)$.
Symmetrically, following the pointers forward gives shortest $v$-$t$-path through $A_{1/\varepsilon}(t)$.
Reporting this path in order from $s$ to $t$ can be done by visiting any of the computed edges
at most twice.
If instead $\widehat{d_2}$ is smaller with minimizer $v$ (in \cref{line:d2_minimum}),
we follow the first edge on a shortest path from $s$ to $p(v)$ (which we have stored)
and likewise for each intermediate vertex we encounter this way between $s$ and $p(v)$.
After $p(v)$, we instead follow along a shortest $t$-$p(v)$-path.

\subsection{A Hierarchy of Distance Oracles}
\label{sec:hierarchy}

We now discuss construction of a hierarchy of DOs for general weighted graphs.
The main bottleneck of the space requirement of the DO in \Cref{thm:distance_oracle}
is to store the distance from every pivot to all vertices of the graph.
We next show how to improve this by instead estimating the distance between pivots with another, internal, DO.
In effect, we trade ever smaller space for higher overall stretch.

\hierarchydo*

In order to prove this, we use the following result of Chechik~\cite{Chechik14,Chechik15}.
She gave an improved implementation of the Thorup-Zwick DO~\cite{ThorupZ05}
and also obtained a leaner version of the oracle when restricting the attention to distances between only a subset of the vertices.

\begin{theorem}[Chechik~\cite{Chechik14,Chechik15}]
\label{thm:lean_ThorupZwick}
    Let $G = (V,E)$ be an undirected weighted graph.
    For any positive integer $k$, and vertex set $B \subseteq V$,
    there is a data structure that reports $(B{\times}B)$-distances with
    multiplicative stretch $2k-1$, space $O(|B|^{1+\frac{1}{k}})$, and constant query time.
    The preprocessing time is $O(n^2 + m \sqrt{n})$.
\end{theorem}

\noindent
\textbf{Preprocessing and space.}
For \Cref{thm:hierarchy_DO}, the set of pivots $B$ is sampled from $V$ 
using the probability $C \log (n)/K$. 
(That means, we do not apply the indirect sampling via edges here.)
Let the lists $K[v]$ and closest pivot $p(v)$ be defined as in \Cref{subsec:DOs_near-additive}.
We preprocess the DO of \Cref{thm:lean_ThorupZwick} for pairs of pivots in $B$.
	Let $\widehat{D}(p,q)$ denote the estimate of $d(p,q)$ returned by that DO when queried with $p,q \in B$.
Our data structure stores the list $K[v]$ for every $v \in V$ as well as the $(B{\times}B)$-DO.
This takes space $O(nK + (\frac{n}{K} \log n)^{1+1/k}) = O((\frac{n}{K})^{1+1/k} \, \log^{1+1/k} n)$,
assuming $K = O(n^{1/(2k+1)})$.
\vspace*{.5em}

\begin{algorithm}[t]
\setstretch{1.33}
\vspace*{.25em}
    $A_{4k/\varepsilon}(s) \gets \{v \in V \mid d^{\text{hop}}_H(s,v) \le \lceil 4k/\varepsilon \rceil\}$\; 
    $A_{4k/\varepsilon}(t) \gets \{v \in V \mid d^{\text{hop}}_H(v,t) \le \lceil 4k/\varepsilon \rceil\}$\;  
   \eIf{$t \in A_{4k/\varepsilon}(s)$}{\Return $d_{4k/\varepsilon}(s,t)$\;}
   		{\Return $\min \{d(s,p(u)) + \widehat{D}(p(u),p(v)) + d(p(v),t) \mid u \in A_{4k/\varepsilon}(s); \nwspace v \in A_{4k/\varepsilon}(t)\}$\;}
\caption{Query algorithm of the DO in \Cref{thm:hierarchy_DO} for the query $(s,t)$.\\
	$d^{\text{hop}}_H$ is the hop-distance in $H$, $d_{4k/\varepsilon}$ is the minimum length of
	all paths with at most $\lceil 4k/\varepsilon\rceil$ hops in $H$,
	$p(v) \in B$ is the pivot closest to $v$ in $G$,
	and $\widehat{D}$ is the output of the DO in \Cref{thm:lean_ThorupZwick}.\vspace*{.25em}}
\label{alg:query_hierarchy}
\end{algorithm}

\noindent
\textbf{Query algorithm.}
The query algorithm is shown in \Cref{alg:query_hierarchy}.
Recall that we defined the auxiliary graph $H$ by requiring 
that every vertex is connected to its $K$ closest vertices in $G$.
Define $\delta = \varepsilon/(2k)$ and
let $A_{2/\delta}(s), A_{2/\delta}(t)$ the sets of vertices that have a hop-distance 
at most $\lceil 2/\delta \rceil = \lceil 4k/\varepsilon \rceil$ from $s$ and $t$, respectively, in $H$.
If $t$ is found while computing $A_{2/\delta}(s)$, the oracle returns the distance $d_{2/\delta}(s,t)$,
that is, the minimum length of all $s$-$t$-paths in $H$ with at most $\lceil 2/\delta \rceil$ edges.
Otherwise, it reports
\begin{equation*}
	\widehat{d} = \min_{u \in A_{2/\delta}(s),\, v \in A_{2/\delta}(t)} d(s,p(u)) + \widehat{D}(p(u),p(v)) + d(p(v),t)
\end{equation*}
Evaluating that minimum over all pairs $(u,v)$ takes time 
$O(|A_{2/\delta}(s)| \nwspace	 |A_{2/\delta}(t)|) 
	= O( (K^{\lceil 2/\delta\rceil})^2) = O(K^{2\lceil 4k/\varepsilon \rceil})$,
which dominates the query time.
\vspace*{.5em}

\noindent
\textbf{Stretch.}
Fix two query vertices $s,t \in V$ and let $d = d(s,t)$ be their distance.
Recall that the set $V_{\frac{\delta d}{2} + W}$
consists of all those vertices whose ball with radius $\frac{\delta d}{2} + W$ has size at most $K$.

\begin{lemma}
\label{lem:correctness-hierarchy}
	With high probability over all queries $(s,t) \in V^2$,
	the DO returns an estimate of $d(s,t)$ with stretch $(2k{-}1{+}\varepsilon, 4kW)$.
\end{lemma}

\begin{proof}
	Let $P$ be a shortest $s$-$t$-path in $G$.
	By \Cref{lem:property-H} with $\delta$ in place of $\varepsilon$,
	if all vertices of $P$ are in $V_{\frac{\delta d}{2} + W}$
	then $s$ and $t$ have hop-distance at most $\lceil 2/\delta \rceil$
	in the auxiliary graph $H$, whence $t \in A_{2/\delta}(s)$
	and $d_{2/\delta}(s,t) = d(s,t)$.
	
	Otherwise, the set $V(P){\setminus}V_{\frac{\delta d}{2} + W}$ is non-empty.
    Let $v_s$ be the vertex on $P$
    that is closest to $s$ and does not lie in $V_{\frac{\delta d}{2} + W}$.
    Clearly, we have $v_s \in A_{2/\delta}(s)$.
    For the vertex $v_t \in V(P){\setminus}V_{\frac{\delta d}{2} + W}$ that is
    closest to the other endpoint $t$, we get $v_t \in  A_{2/\delta}(t)$.
    In this case, the output of our oracle is at most
    \begin{align*}
    	\widehat{d} &\le d(s,p(v_s)) + \widehat{D}(p(v_s),p(v_t)) + d(p(v_t),t)
    		 \le d(s,p(v_s)) + (2k{-}1) \nwspace d(p(v_s),p(v_t)) + d(p(v_t),t)\\
    		&\le d(s,v_s) + d(v_s,p(v_s)) + 
    			(2k{-}1) d(p(v_s),v_s) + (2k{-}1) \nwspace d(v_s,v_t) 
    				+ (2k{-}1) \nwspace d(p(v_t),v_t) \, +\\
    		&\quad\  d(p(v_t),v_t) + d(v_t,t)\\
    		&\le (2k{-}1) \nwspace d(s,t) + 2k \, d(v_s,p(v_s)) + 2k \, d(p(v_t),v_t)
    		 \le (2k{-}1) \nwspace d + 4k \left( \frac{\delta d}{2} + W \!\right)\\
    		&= (2k{-}1{+}\varepsilon) \nwspace d + 4kW. \qedhere
    \end{align*}
\end{proof}

\section{Fault-Tolerant Tree Oracles}
\label{sec:framework_DSO_short}

The aim of this section is to prove~\Cref{thm:short_distance_DSO} (restated below),
adding fault tolerance to the construction of distance oracles for short paths.
The core idea is to build FT-trees during preprocessing 
to identify a family of relevant subgraphs of $G$.
For those subgraphs $(\alpha, \beta)$-stretch DOs for hop-short distances are constructed and stored.
The FT-trees, however, are discarded at the end of the preprocessing as they are too large to fit
in subquadratic space.
The main challenge is to utilize the correct DOs at query time.
For this, we \emph{emulate} a root-to-leaf transversal in the discarded FT-tree. 
This motivates the name \emph{fault-tolerant tree oracles}, abbreviated \emph{FT-tree oracles}.

\shortdistanceDSO*

To achieve the above theorem, we build the FT-trees along side $O(Lf \log n)^f$ path-reporting DOs with stretch $(\alpha, \beta)$ during preprocessing. When the construction of the FT-trees is completed, we dump the FT-trees and keep only the path-reporting DOs. We will show that this path-reporting DOs together with some auxiliary data structures will allow us to emulate a search in the FT-trees that were used to construct them. 

\paragraph*{The preprocessing algorithm.}
We now describe how the FT-trees and the path-reporting DOs are built side-by-side (see~\Cref{alg:derandomizing_WY} for the pseudocode). We first define the notion of {\sl $(\alpha, \beta)$-hit-and-miss}, which is somewhat similar to the one provided in Karthik and Parter~\cite{KarthikParter21DeterministicRPC}.\footnote{For simplicity we will write hit-and-miss when $(\alpha, \beta)$ are clear from the context.}
We say that a family $\mathcal{O}$ of DOs  $(\alpha,\beta)$-hits-and-misses a node $x$ in  $FT(s,t)$, if there exists at least one DO $\mathcal{O}_x \in \mathcal{O}$ such that the path $P_{\mathcal{O}_x}(s,t)$ reported by $\mathcal{O}$ when queried with $(s,t)$ is edge-disjoint from $F$, {\sl i.e.}, $E(P_{\mathcal{O}}(s,t)) \cap F = \emptyset$ and $|P_{\mathcal{O}}(s,t)| \le \alpha d(s,t,F) + \beta$.
We build the FT-trees for $G$ level by level. For each level $i$, we precompute a family $\mathcal{O}_i$ of path-reporting DOs for hop-short distances that hits-and-misses all the $i$-th level nodes $x$ of all FT-trees. 
The data structure of our $f$-DSO is then given by the family $\{\mathcal{O}_0,\dots, \mathcal{O}_f\}$. So, we discard all the computed FT-trees. Moreover, we compute an auxiliary data structure that, for every node $x$ in level $i$ of some FT-tree, can deterministically identify the {\em representative} DO in $\mathcal{O}_i$ that hits-and-misses $x$. This auxiliary data structure is necessary to simulate the search on the FT-trees whenever we need to answer to a query $(s,t,F)$.

We now describe how we compute the FT-trees level by level.

For level $0$, $\mathcal{O}_0$ contains the path-reporting DO$^{\leq L}$ $\mathcal{O}_G$ with stretch $(\alpha,\beta)$ that is computed using algorithm $\mathcal{A}$. For every two distinct vertices $s$ and $t$ of $G$, the root node $x$ of $FT(s,t)$ is associated with the path $P_x$ computed by querying the DO$^{\leq L}$ $\mathcal{O}_G$ with the pair $(s,t)$. Moreover, $F_x=\emptyset$. 

The generic level $i$, with $1\leq i\leq f$, of the FT-trees is computed using a two stage algorithm. 

\begin{algorithm}[th!]
   	\caption{The algorithm that builds the $f\text{-DSO}^{\leq L}$ with stretch $(\alpha,\beta)$ of~\autoref{thm:short_distance_DSO}.}
\label{alg:derandomizing_WY}
  	
    \SetKwInOut{Input}{Input}
    \SetKwInOut{Output}{Output}

    \Input{An undirected graph $G = (V,E)$, a positive integer $L$, a sensitivity parameter $f \in \mathbb{N}$, with $f = o(\frac{\log n}{ \log \log n})$, and an algorithm $\mathcal{A}$ that can build a path-reporting DO$^{\leq L}$ with stretch $(\alpha,\beta)$, space $\mathsf{S}$,  query time $\mathsf{Q}$, and preprocessing time $\mathsf{T}$ such that (i) both distance and path-reporting queries are deterministic and (ii) when the DO$^{\leq L}$ is queried on a pair $(s,t)$, the number of edges of the reported path, if a path is returned, is at most the value of the reported distance.}
    \Output{The data structure of a deterministic $f$-$\text{DSO}^{\leq L}$ of $G$ with stretch $(\alpha, \beta)$.}
    
   \BlankLine
    \For(\tcp*[f]{Compute $i$-th level of all FT-trees}){$i = 0,\dots,f$}
    {
        $\mathcal{HM}'_i \gets \emptyset$\;
        \ForEach{$s,t \in V, s\neq t$} 
        {
            \If(\tcp*[f]{Build the root node of $FT(s,t)$}){$i=0$}
            {
                Add the root node $x$ to $FT(s,t)$ with $F_x=\emptyset$\;
            }
            \Else(\tcp*[f]{Build the $i$-th level with $i>0$ from nodes of level $i-1$})
            {
                \ForEach{node $x$ in $(i-1)$-th level of $FT(s,t)$}
                {
                    \ForEach{edge $e$ in $P_x$}
                    {
                        Add to $FT(s,t)$ a child $x_e$ of $x$ with $F_{x_e}=F_x\cup\{e\}$\;
                    }
                }
            }
            \ForEach{node $x$ in $i$-th level of $FT(s,t)$}
            {
                 Use Dijkstra or query a precomputed path-reporting $f$-DSO$^{\leq L}$ $\mathcal{O}$ that returns exact distances ({\sl e.g.}, the $f$-DSO$^{\leq L}$ of Theorem 39 in~\cite{KarthikParter21DeterministicRPC}), to get a shortest path $P'(s,t,F_x)$ from $s$ to $t$ in $G-F_x$\;\label{step:rpc} 
                \If{$|P'(s,t,F_x)|\leq L$}
                {
                    $P'_x\gets P'(s,t,F_x)$\;
                    add $(P'_x,F_x)$ to $\mathcal{HM}_i$\;
                }
                \Else
                {
                    $P'_x \gets \perp$\;
                }
            }
            
        }

        Compute a family of graphs $\mathcal{G}_i$ that hits-and-misses all pairs in $\mathcal{HM}_i$ \label{step:partial_rpc}\;

        Use $\mathcal{A}$ to compute the family $\mathcal{O}_i=\{\mathcal{O}_H \mid H \in \mathcal{G}_i\}$\;

        Build a data structure $\mathcal{D}_i$ that, when queried on a node $x$ in the $i$-th level of  $FT(s,t)$, with $P'_x \neq \perp$, deterministically returns $\mathcal{O}_x\in \mathcal{O}_i$ that hits-and-misses $(P(s,t,F_x),F_x)$, where $P(s,t,F_x)$ is a path from $s$ to $t$ in $G-F_x$ with at most $\alpha L + \beta$ edges\;
        \ForEach{$s,t \in V, s\neq t$} 
        {
            \ForEach{node $x$ in $i$-th level of $FT(s,t)$, with $P'_x \neq \perp$}
            {
                Query $\mathcal{D}_i$ with $x$ to find the oracle $\mathcal{O}_x \in \mathcal{O}_i$
                \label{step:choose_D_x}
                \;

                Query $\mathcal{O}_x$ to compute the path $P(s,t,F_x)$  \label{step:query-do} \;
                $P_x \gets P(s,t,F_x)$\;
            }  
        }
    }
    
    \BlankLine
    \Return $\{(\mathcal{O}_i,\mathcal{D}_i) \mid i\in \{0,...,f\}\}$\;
\end{algorithm}

\begin{algorithm}[th!]
   	\caption{The query algorithm of the $f\text{-DSO}^{\leq L}$ with stretch $(\alpha,\beta)$ of~\autoref{thm:short_distance_DSO}.}
\label{alg:derandomizing_WY_query}
  	
    \SetKwInOut{Input}{Input}
    \SetKwInOut{Output}{Output}

    \Input{A triple $(s,t,F)$, where $s,t \in V$, with $s\neq t$ and $F\subseteq E$, $|F|\leq f$.}
    \Output{A path from $s$ to $t$ in $G-F$ with at most $\alpha L+\beta$ edges, if $d(s,t,F)\leq L$, or $\perp$ otherwise.}
    
   \BlankLine
   $F_x\gets \emptyset$; $i\gets 0$; $\text{status} \gets \text{not-found}$\;
    \While(\tcp*[f]{Emulate the search in $FT(s,t)$}){$\text{status}=\text{not-found}$}
    {
        Query $\mathcal{D}_i$ with $(s,t,F_x)$ to get the representative DO$^{\leq L}$ $\mathcal{O}_x \in \mathcal{O}_i$\;

        Query $\mathcal{O}_x$ with $(s,t)$ to compute $P_x=P(s,t,F_x)$\;

        \If(\tcp*[f]{Check for feasibility}){$P_x=\perp$ or $|P_x|> \alpha L + \beta$}
        {
            $\text{status}\gets \text{not-exist}$\;
        }
        \ElseIf{$E(P_x) \cap F =\emptyset$}
        {
            $\text{status}\gets \text{found}$\;
        }
        \Else
        {
            Let $e \in E(P_x) \cap F$\;
            \label{step:add_e}$F_x \gets F_x \cup\{e\}$\; 
            $i\gets i+1$\;
        }
    }
    \BlankLine
    \textbf{if }{$\text{status}=\text{found}$}
        \textbf{then} \Return $P_x$
        \textbf{else} \Return $\perp$\;
\end{algorithm}

In the first stage we do the following. For every two distinct  vertices $s$ and $t$ of $G$, and for every node $x$ in level $i$ of $FT(s,t)$, we compute a path $P'_x$ from $s$ to $t$ in the graph $G-F_x$ using Dijkstra's algorithm. Once the entire $i$-th level of all FT-trees has been computed,  we compute a family of graphs $\mathcal{G}_i$ that,  according to the definition provided in Karthik and Parter~\cite{KarthikParter21DeterministicRPC}, \emph{hits-and-misses} the family of pairs $(P'_x,F_x)$, for all nodes $x$ in level $i$ of some FT-tree. More precisely, $\mathcal{G}_i$ hits-and-misses $(P'_x,F_x)$ if there is $H_x \in \mathcal{G}_i$ such that $H_x$ contains $P'_x$, but does not contain $F_x$.

In the second stage we use the algorithm $\mathcal{A}$ to build the set of DOs for hop-short distances $\mathcal{O}_i=\{\mathcal{O}_H \mid H\in \mathcal{G}_i\}$ with stretch $(\alpha,\beta)$. We then substitute all the paths associated with the nodes of the $i$-th level of all the FT-trees with 
the corresponding paths computed by querying the representative DOs for hop-short distances in $\mathcal{O}_i$. More precisely, given a node $x$ of the $i$-th level of $FT(s,t)$, we pick a graph $H_x\in \mathcal{G}_i$ that 
hits-and-misses $(P'_x,F_x)$, 
we query the DO$^{\leq L}$ $\mathcal{O}_{H_x}$ with stretch $(\alpha,\beta)$ to compute a path $P_x$ from $s$ to $t$ in $G-F_x$, and, finally, we substitute $P'_x$ with $P_x$ in node $x$. So, at the end of the second stage, each node $x$ contained in the $i$-th level of an FT-tree stores the path $P_x$ computed using the representative DO$^{\leq L}$ $\mathcal{O}_{H_x} \in \mathcal{O}_i$.
This means that the FT-trees are built using the paths $P_x$, and not the paths $P'_x$ that are computed during the first stage. Therefore, when we move from level $i<f$ to level $i+1$, each node $x$ of the $i$-th level has one child node $x_e$ for each edge $e$ of $P_x$, and $F_{x_e}=F_x\cup\{e\}$. Therefore, the branching factor of the nodes in the FT-trees is $\alpha L + \beta$ in the worst case.

\paragraph*{The query algorithm.}
We now describe how we can answer to a query $(s,t,F)$, and use the precomputed family of DOs for hop-short distances $\{\mathcal{O}_0,\dots, \mathcal{O}_f\}$ to emulate the search in $FT(s,t)$ (see~\Cref{alg:derandomizing_WY_query}).

To emulate the exploration of a node $x$ in level $i$ of $FT(s,t)$, we use the auxiliary data structure to identify the representative DO$^{\leq L}$ in  $\mathcal{O}_i$ and then we query it with the pair $(s,t)$ to compute the path $P_x$. We check whether $P_x$ contains some of the failing edges of $F$. If $P_x$ does not contain a failing edge, we return it. Otherwise, if $P_x$ contains some of the failing edges $F$, we select any edge $e$ of $F$ that is also contained in $P_x$ and emulate the traversal of $FT(s,t)$ by emulating the exploration of the child node $x_e$ of $x$ associated with the failure $e$.

\paragraph*{The analysis.} In the above description we used multiple times Dijkstra to compute a shortest path $P'_x$ in $G - F_x$ for every node $x$ of the FT-trees, but this adds a term of $\Otilde(n^2 (\alpha L + \beta)^f m)$ to the preprocessing. To improve upon the preprocessing time and get rid of the $\Otilde(n^2 (\alpha L + \beta)^f m)$ term, we use the path-reporting $f$-DSO$^{\leq L}$ of Theorem 39 in Karthik and Parter~\cite{KarthikParter21DeterministicRPC} with parameters $L$ and $f$ during preprocessing, whose preprocessing time is $O(fL \log n)^{f+1} \cdot APSP^{\leq L}$, where $APSP^{\leq L}$ is the time needed by an algorithm to compute all-pairs shortest paths with at most $L$ edges in a graph $H$ with $n$ vertices and $O(m)$ edges. In our case it is sufficient to use the trivial $\Otilde(mn)$ APSP algorithm as  $APSP^{\leq L}$, so the time to compute this oracle is $O(fL \log n)^{f+1} \cdot O(mn)$.
This oracle answers to a distance query $(s,t,F)$ in time $O(f^2)$, and reports the path $P(s,t,F)$ in additional $O(L)$ time.

We use Theorem 48 of Karthik and Parter to compute the set of graphs $\mathcal{G}_i$ in time $\Otilde(n^2 (\alpha L + \beta)^f) \cdot O(fL \log n)^f)$. Lemma 51 of Karthik and Parter states that, for every node $x$ of level $i$ in our FT-trees, there are $O(f\log n)$ graphs in $\mathcal{G}_i$ that do not contain $F_x$, and at least one of such $O(f\log n)$ graphs hits-and-misses $x$.  In a similar way as in Section 4.3 of \cite{Bilo23CompactDO}, one can construct a data structure $\mathcal{D}_i$ by specifying a subset of columns of an error-correcting code matrix.
The size of the structure is $O(Lf \log n)^f \cdot O(f \log n)$ and given a set $F_x$ of at most $f$ edges, it can compute the (indices of a) set of $O(f \log n)$ graphs that do not contain $F_x$ in time $\Otilde(f^2 L)$. 

$\mathcal{A}$ is the construction algorithm of any DO$^{\leq L}$ with stretch of $(\alpha,\beta)$. As mentioned above, we use $\mathcal{A}$ to convert each graph $H$ in $\mathcal{G}_i$ into a corresponding DO$^{\leq L}$ $\mathcal{O}_H$ with stretch $(\alpha,\beta)$.

The auxiliary data structure we use in our theorem to identify the representative DO$^{\leq L}$ in $\mathcal{O}_i$ of a given node $x$ of level $i$ in $FT(s,t)$ is defined as follows. We use Lemma 51 in Karthik and Parter to identify the $O(f\log n)$ graphs that do not contain $F_x$ in $\Otilde(f^2L)$ time. For each of these $O(f\log n)$ graphs, say $H$, we look at the corresponding DO$^{\leq L}$ $\mathcal{O}_H \in \mathcal{O}_i$ with stretch $(\alpha, \beta)$ and query $\mathcal{O}_H$ with the pair $(s,t)$ to discover the $(\alpha,\beta)$-approximate distance $\hat{d}_H(s,t)$ from $s$ to $t$ in the graph $H$. We point out that, at this particular stage, we only need to retrieve the values of the distances from $s$ to $t$ in the $O(f\log n)$ graphs. Among all distances computed, we take the minimum value and look at the corresponding oracle $\mathcal{O}_H \in \mathcal{O}_i$ associated to a graph $H \in \mathcal{G}_i$. We assume that each oracle in $\mathcal{O}_i$ has a unique identifier. So, in case of ties, we select the DO$^{\leq L}$ with minimum identifier. We use exactly the same deterministic algorithm for building the FT-trees, that is, the path $P_x$ that replaces $P'_x$ is the one reported by querying the oracle $O_H$ with the pair $(s,t)$.

We recall that $\mathsf{Q}$, $\mathsf{S}$, and $\mathsf{T}$ are, respectively, the query time, space, and preprocessing time of the path-reporting DO$^{\leq L}$ of a graph with $n$ vertices and $O(m)$ edges that is built by algorithm $\mathcal{A}$.
\begin{lemma}
\Cref{alg:derandomizing_WY} constructs an $f$-DSO of size $O(fL \log n)^f \cdot \mathsf{S}$ in time $\Otilde(n^2 (\alpha L + \beta)^f) \cdot ( O(fL \log n)^f + \mathsf{Q} ) + O(fL \log n)^f \cdot \mathsf{T}$.
\end{lemma}

\begin{proof}
First, we preprocess the $f$-DSO$^{\leq L}$ $\mathcal{O}_{KP}$ of Karthik and Parter from Theorem 39 in \cite{KarthikParter21DeterministicRPC} with parameters $L$ and $f$ in $O(fL \log n)^{f+1} \cdot mn$ time, as described above.

Then we build, level by level, $n^2$ FT-trees $FT(s,t)$. We claim that each FT-tree contains at most $(\alpha L + \beta)^f$ nodes, as every node $x$ of the tree contains either a path $P_x := \perp$ and then $x$ must be a leaf node, or it contains a path $P_x$ with at most $\alpha L + \beta$ edges that is reported by some DO$^{\leq L}$ with stretch $(\alpha, \beta)$ (if we did not find a path with at most $\alpha L + \beta$ edges then we store $P_x := \perp$ as well). The height of an FT-tree is at most $f$ and the number of its nodes is $O( (\alpha L + \beta)^f)$.

Next, we analyze the time it takes to construct the $i$-th level of the fault tolerant trees. Similar to the above claim, we can observe that there are at most $(\alpha L + \beta)^i$ nodes at level $i$. For each such node $x$, we first query $\mathcal{O}_{KP}$ with $(s,t,F_x)$ to obtain $P'_x$ in $O(f^2 + L)$ time. Then, we collect all the pairs $(P'_x, F_x)$ for every $x$ in the $i$-level of any FT-tree $FT(s,t)$, and run the algorithm of Theorem 48 in Karthik and Parter~\cite{KarthikParter21DeterministicRPC} which computes a set of $O(f L \log n)^i$ graphs $\mathcal{G}_i$ that hits-and-misses all the pairs $(P'_x, F_x)$ corresponding to the nodes $x$ of that level in $\Otilde(n^2 (\alpha L + \beta)^i) \cdot O(f L \log n)^i$ time. Next, we use $\mathcal{A}$ to construct a DO$^{\leq L}$ with stretch $(\alpha, \beta)$ for every graph in $\mathcal{G}_i$ in $O(f L \log n)^i \cdot \mathsf{T}$ time. 
The data structure $\mathcal{D}_i$ does not require additional construction time as it utilizes data structures we have already constructed above.
Finally, for every node $x$ in the $i$-th level of $FT(s,t)$ with $P'_x \ne \perp$ we query $\mathcal{D}_i$ with $x$ to find the representative oracle $\mathcal{O}_x \in \mathcal{H}$ in $O(f^2 L)$ time, then we query $\mathcal{O}_x$ with $(s,t,F_x)$ in $\mathsf{Q}$ time and report the path $P(s,t,F_x)$ in additional $O(\alpha L + \beta)$ time. 

In total, the preprocessing time of~\Cref{alg:derandomizing_WY} is dominated by $\Otilde(n^2 (\alpha L + \beta)^f) \cdot ( O(fL \log n)^f + \mathsf{Q} ) + O(fL \log n)^f \cdot \mathsf{T}$.

The size of the $f$-DSO$^{\leq L}$ is the size $\mathsf{S}$ of the input DO multiplied by the $O(f L \log n)^i$ graphs that are generated in line~\ref{step:rpc} over all levels $0 \le i \le f$, which sums up to $O(f L \log n)^f \cdot \mathsf{S}$. 
\end{proof}

\begin{lemma}
\Cref{alg:derandomizing_WY_query} takes 
$\Otilde(f \cdot (\mathsf{Q} + (\alpha+f^2) L + \beta))$ time.
\end{lemma}

\begin{proof}
The algorithm simulates the traversal of the query $(s,t,F)$ in the FT-tree $FT(s,t)$. Starting with $F_x \gets \emptyset$, at each iteration an $(\alpha, \beta)$-approximate shortest path $P_x$ from $s$ to $t$ in the graph $G-F'$ is computed by querying $\mathcal{D}_i$ to find the representative DO$^{\leq L}$ $\mathcal{O}_x$ that hits-and-misses $(P_x, F_x)$, and querying $\mathcal{O}_x$ with $(s,t)$ to find the path $P_x$. 
The query to $\mathcal{D}_i$ takes $O(f^2L + f \log n \cdot \mathsf{Q})$ time as it includes using the hit-and-miss error correcting codes in $O(f^2 L)$ time to find $O(f \log n)$ DOs for hop-short distances that hit-and-miss $(P'_x, F_x)$, querying each of these oracles with $(s,t)$ in $O(f \log n \cdot \mathsf{Q})$ time, and obtaining the representative DO$^{\leq L}$ $\mathcal{O}_x$, which is the first to return the minimum $s$-to-$t$ distance. Then querying $\mathcal{O}_x$ with the pair $(s,t)$ takes additional $\mathsf{Q}$ time and reporting an $(\alpha, \beta)$-approximate shortest path $P_x$ using $\mathcal{O}_x$ takes additional $O(\alpha L + \beta)$ time. In total, it takes $O(f^2L + f \log n \mathsf{Q} + \alpha L + \beta)$ to compute the path $P_x$. If $|P_x| > L$ then the query stops and returns that there is no path of length at most $L$ from $s$ to $t$ in $G-F$. Otherwise, let $e \in E(P_x) \cap F$, the algorithm adds $e$ to $F_x$ and continues to the next iteration. 

As $F_x \subseteq F$ and $|F| \le f$, the algorithm must finish after at most $f+1$ iterations, and as each iteration takes $O(f^2L + f \log n \mathsf{Q} + \alpha L + \beta)$ time, in total the runtime of the query is $\Otilde(f \cdot (\mathsf{Q} + (\alpha+f^2) L + \beta))$. 
\end{proof}

\begin{lemma} [Correctness]
\label{lem:correctness_query_short-distance_DSO}
Given a query $(s,t,F)$,~\Cref{alg:derandomizing_WY_query} computes either a path $P_x$ from $s$ to $t$ in $G-F$ with at most $\alpha L+\beta$ edges such that $d(s,t,F)\leq |P_x|\leq \alpha d(s,t,F)+\beta$ or it returns $\perp$. If $d(s,t,F)\leq L$, then the algorithm always returns a path $P_x$.
\end{lemma}

\begin{proof}
We prove by induction on $0 \le i \le f$ that~\Cref{alg:derandomizing_WY_query} emulates the search in the FT-tree $FT(s,t)$ that was built by~\Cref{alg:derandomizing_WY}. This is enough to show correctness. 

For $i=0$, i.e., $F_x=\emptyset$, there is only one deterministic data structure $\mathcal{O} \in \mathcal{O}_0$ which is an $(\alpha, \beta)$-stretch DO$^{\leq L}$ of $G$. Querying $\mathcal{O}$ with the pair $(s,t)$ either returns a path $P_x$ from $s$ to $t$ in $G$ such that $|P_x|\geq d(s,t)$ or it returns $\perp$. Moreover, if $d(s,t)\leq L$, then the oracle indeed returns a path such that $|P_x|\leq \alpha d(s,t) + \beta \leq \alpha L + \beta$ as it guarantees a stretch of $(\alpha,\beta)$. As the oracle is deterministic, $|P_x|$ corresponds exactly to the path that was stored in the root of $FT(s,t)$ by~\Cref{alg:derandomizing_WY}.

For the inductive step, by the induction hypothesis assume that, a node $x$ in level $i-1$ of $FT(s,t)$ stores exactly the same information that was associated with it by~\Cref{alg:derandomizing_WY}. 

The proof breaks into the following two cases.
\begin{description}
    \item[Case 1:] node $x$ stores $\perp$. In this case, it must be that $d(s,t,F_x) > L$. As $F_x\subseteq F$, we have that $d(s,t,F)\geq d(s,t,F_x)>L$ and~\Cref{alg:derandomizing_WY_query} 
    correctly returns $\perp$.
    \item[Case 2:] node $x$ stores a path $P_x$ with at most $\alpha L + \beta$ edges such that $d(s,t,F_x)\leq |P_x|\leq \alpha d(s,t,F_x) + \beta$. Then, either $E(P_x)\cap F = \emptyset$, and the algorithm correctly returns $P_x$ as $d(s,t,F)\leq |P_x|\leq \alpha d(s,t,F_x)+\beta \leq \alpha d(s,t,F)+\beta$, or $E(P_x)\cap F \neq \emptyset$. In this latter case, let $e \in E(P_x) \cap F$. In line~\ref{step:add_e} at the end of the $(i-1)$-iteration of~\Cref{alg:derandomizing_WY_query} we add $e$ to $F_x$. Let $x$ be the node of level $i$ that corresponds to the new $F_x$. Clearly,~\Cref{alg:derandomizing_WY} added this node to $FT(s,t)$. We show that the information computed by~\Cref{alg:derandomizing_WY_query} for $x$ coincides with the information that~\Cref{alg:derandomizing_WY} stored in $x$.

    According to Theorem 39 in~\cite{KarthikParter21DeterministicRPC}, querying  their exact $f$-DSO$^{\leq L}$ with $(s,t,F_x)$ results with a shortest path $P'_x$ from $s$ to $t$ in $G-F_x$. According to Theorem 48 and Lemma 51 in~\cite{KarthikParter21DeterministicRPC}, the family of graphs $\mathcal{G}_i$ computed in Step~\ref{step:partial_rpc} of~\Cref{alg:derandomizing_WY} has the property 
    that there are $O(f \log n)$ graphs among $\mathcal{G}_i$ which exclude $F_x$ and at least one of these graphs, say $G_j$, also contains $P'_x$. Recall that $\mathcal{O}_i$ is the set of the $(\alpha, \beta)$-stretch DOs for hop-short distances constructed for all graph is $\mathcal{G}_i$.
    We query all the DOs for hop-short distances  corresponding to the $O(f \log n)$ graphs that exclude $F_x$ ($\mathcal{O}_{\mathcal{G}_j}$ included) with the pair $(s,t)$ to $(\alpha,\beta)$-approximate the value $d(s,t,F-x)$, we keep the minimum-index DO$^{\leq L}$ $\mathcal{O}_x$ that returned the minimum among the computed values as in~\Cref{alg:derandomizing_WY}, and we query it to report the path $|P_x|$ from $s$ to $t$ in $G-F_x$. As all the queried oracles are deterministic, if all of them ($\mathcal{O}_{\mathcal{G}_j}$ included) return $\perp$, then $d(s,t,F_x)>L$,~\Cref{alg:derandomizing_WY} stores $\perp$ in $x$, and~\Cref{alg:derandomizing_WY_query} correctly returns $\perp$. Therefore, we can assume that  $\mathcal{O}_x$ returns a path $|P_x|$. As the graph $\mathcal{G}_x \in \mathcal{G}_i$ corresponding to $\mathcal{O}_x$ excludes $F_x$, we have that $|P_x|$ is a path from $s$ to $t$ in $G-F_x$. Therefore, $d(s,t,F_x)\leq |P_x|$. Let $\hat{d}_x$ and $\hat{d}_j$ be the distances returned by querying $\mathcal{O}_x$ and $\mathcal{O}_j$ with the pair $(s,t)$, respectively. By the choice of $\mathcal{O}_x$ and because $\mathcal{O}_j$ has a stretch of $(\alpha,\beta)$, we have $\hat{d}_x \leq \hat{d}_j \leq \alpha d_{\mathcal{G}_j}(s,t,F_x)+\beta =\alpha |P'_x|+\beta$. Furthermore, as the $s$-to-$t$ path $|P_x|$ reported by $\mathcal{O}_x$ is no longer than $\hat{d}_x$, it follows that $|P_x|\leq \alpha |P'_x| + \beta$. As a consequence, if $|P_x|>\alpha L + \beta$, then $|P'_x|> L$ and~\Cref{alg:derandomizing_WY_query} correctly returns $\perp$ as~\Cref{alg:derandomizing_WY} stored $\perp$ in $x$. If $|P_x|\leq \alpha L + \beta$, then~\Cref{alg:derandomizing_WY} stores $P_x$ in the node $x$ (lines~\ref{step:choose_D_x}~and~\ref{step:query-do}) and, by construction,~\Cref{alg:derandomizing_WY_query} correctly simulates the search in $FT(s,t)$.  
\end{description}
\end{proof}

\section{General Distance Sensitivity Oracles}
\label{sec:framework_DSO}

This section proves \Cref{thm:general_DSO} (restated below),
which takes an $f$-DSO$^{\leq L}$ for short distances 
and turns it into a general $f$-DSO.
Let $f$ be the sensitivity of the input DSO and $\widehat{d^{\le L}}(s,t,F)$
the returned estimate with stretch $(\alpha,\beta)$.
Recall that we always assume $\widehat{d^{\le L}}(s,t,F) \ge d(s,t,F)$;
but $\widehat{d^{\le L}}(s,t,F) \le \alpha \cdot d(s,t,F) + \beta$
only needs to hold for those queries for which
$s$ and $t$ have a replacement path on at most $L$ edges.
The input data structure can be the one from \Cref{thm:short_distance_DSO}
but this is not required.
In the current setting, we assume that $G$ is undirected, unweighted, and has unique shortest paths.
The overall ideas are similar to that of Bilò et al.~\cite{Bilo24ApproxDSOSubquadraticTheoretiCS}.
However, their result focus strictly on a multiplicative stretch of $3$.
When implementing the general ideas, we divert more and more
in order to enable our reduction to work with input DSOs with general stretch $(\alpha,\beta)$.
We highlight the key differences at the respective places.

\generaldso*

\subsection{Preprocessing and Space}
\label{subsec:general_DSO_preproc}

The first step of our preprocessing algorithm is to construct and store 
the assumed input $f$-DSO$^{\leq L}$ with stretch $(\alpha,\beta)$,
taking time $\mathsf{T}_L$ and storage space $\mathsf{S}_L$.
We then choose an approximation parameter $\varepsilon > 0$
that is used to control the increase
of the multiplicative stretch, namely, the resulting DSO has stretch
$(\alpha (1{+}\varepsilon), \beta)$.
For technical reasons, we assume $\varepsilon \le 12/7$.
Note, however, that we do not require it to be constant.

Define the \emph{granularity} as an integer parameter $\lambda = \lambda(\varepsilon,L)$.
Its value will be made precise later,\footnote{%
	The granularity will turn out to be $\lambda = \varepsilon L/9$.
}
for now it is enough to require $\lambda \le L$.
To ease notation, we assume that $\lambda$ is even;
otherwise, every occurrence of $\lambda/2$ below can be replaced by $\lfloor \lambda/2 \rfloor$.
\vspace*{.5em}

\noindent
\textbf{Pivots.}
After constructing the input DSO for hop-short distances,
the preprocessing algorithm samples and stores two random subsets $B_1,B_2 \subseteq V$ of vertices.
These sets have vastly differing densities to fulfill their purpose.\footnote{%
	Set $B_1$ has a low density in $V$ and is used to hit \emph{dense} balls,
	while $B_1$ is dense and supposed to hit \emph{sparse} balls instead.
	To avoid confusion, we do not use this distinction when referring to $B_1,B_2$.%
}
The input DSO for short distances can estimate a replacement path
as long as the endpoints are sufficiently close in $G-F$.
For general replacement distances, we use the pivots in $B_1$ and $B_2$ 
to stitch together short paths.
We will do so by constructing so-called FT-trees with or without granularity $\lambda$ between pivots.
As it turns out, the FT-tree with granularity have a very low query time
(see the discussion after \Cref{lem:weaker_guarantee}).
Unfortunately, this advantage is bought by a much higher space requirement.
This is why we choose the set $B_1$ to be very small and build
FT-trees with granularity only between pairs of vertices in $B_1 \times B_1$.
However, since there are only few \emph{pivots of the first type}, for a given vertex $u \in V$
there may be none in the vicinity of $u$ if the ball around that vertex in $G-F$ is too sparse.
We thus have to provide a fall back option.
We sample another set $B_2$ at a much higher frequency to ensure
that there is always some \emph{pivot of the second type} close to $u$.
Since we only need them if the ball around $u$ is sparse, it will be sufficient
to store and check \emph{all} pivots from $B_2$ around $u$.

Each vertex of $V$ is included in the set $B_1$ independently with probability
$C_1 f \log_2(n)/L^{f}$ for a constant $C_1 > 0$.
An application of standard Chernoff bounds, see e.g.\ 
\cite{GrandoniVWilliamsFasterRPandDSO_journal,RodittyZwick12kSimpleShortestPaths},
shows that w.h.p.\ $|B_1| = \Otilde(fn/L^{f})$.
Let $u \in V$ and $F \subseteq E$, $|F| \le f$, a set of failing edges
and $\ball_{G-F}(u,\lambda/2)$ be the ball of all vertices 
reachable from $u$ in $G-F$  within $\lambda/2$ hops. 
We call the ball \emph{dense},
if it has $|\ball_{G-F}(u,\lambda/2)| > L^f$ elements; and \emph{sparse} otherwise.
Chernoff bounds also show that, with high probability over all pairs $(u,F)$,
whenever some $\ball_{G-F}(u,\lambda/2)$ is dense, it intersects $B_1$.

The properties of the other set $B_2$ are a bit more involved.
We define the notions of a \emph{trapezoid} and a path being \emph{far away} from all failures
as introduced by Chechik et al.~\cite{ChCoFiKa17}.

\begin{definition}[$\frac{\varepsilon}{6}$-trapezoid]
\label{def:trapezoid}
  Let $F \subseteq E$ a set of edges, $u,v \in V$ two vertices,
  and $P$ a $u$-$v$-path in $G-F$.
  The $\frac{\eps}{6}$\emph{-trapezoid} around $P$ in $G-F$ is
  \begin{equation*}
    \tr^{\eps/6}_{G-F}(P) = \big\lbrace\nwspace z \in V{\setminus}\{u,v\} \mid
      \exists y \in V(P) \colon d_{G-F}(y,z)
      \le \frac{\varepsilon}{6}
        \cdot \min(\nwspace \length{P[u..y]}, \length{P[y..v]} \nwspace) \big\rbrace. 
  \end{equation*}
  The path $P$ is \emph{far away} from all failures in $F$ if 
  $\tr^{\varepsilon/6}_{G-F}(P) \cap V(F) = \emptyset$;
  otherwise $P$ is \emph{close} to $F$.
\end{definition}

Each vertex is sampled for $B_2$ independently with probability $C_2 f \log_2(n)/\lambda$
for some constant $C_2 > 0$.
We call the elements of $B_2$ \emph{pivots of the second type}.
Again, we have w.h.p.\ $|B_2| = \Otilde(f n/\lambda)$ such pivots.
Moreover, by choosing $C_2$ sufficiently large, we achieve the following covering property of $B_2$.
With high probability over all triples of vertices $u,v,w \in V$ and failure sets $F \subseteq E$ with $|F| \le f$, the following holds.

\begin{itemize}
	\item If the unique shortest $u$-$v$-path in $G$ has at least $L$ edges,
		then it contains a vertex from $B_2$.
	\item If the $u$-$v$-replacement path in $G-F$
		has at least $\lambda/2$ edges and is far away from all failures in $F$,
		then it contains a vertex of $B_2$.
	\item If the concatenation of the $u$-$v$-replacement path 
		and the $v$-$w$-replacement path in $G-F$ 
		has at least $\lambda$ edges and is far away from all failures in $F$,
		then it contains a vertex of $B_2$ among its first and last $\lambda/2$ vertices.
\end{itemize}

\noindent
The second and third item together imply that w.h.p.\ any concatenation of 
\emph{at most} two replacement paths that is hop-long and far away from $F$
has a pivot of the second type within distance $\lambda/2$ from either endpoint.
\vspace*{.5em}

\noindent
\textbf{Pivot trees.}
We need a mechanism in place that allows us to check at query time whether there
is a pivot of the first type that is sufficiently close to $v$; and, if not,
gives access to a set of not too many pivots of the second type in the vicinity.
The problem is that there are too many graphs $G-F$ to preprocess them directly.
Here, we depart from \cite{Bilo24ApproxDSOSubquadraticTheoretiCS}.
They used an $(L,f)$-replacement path coverings (RPC)~\cite{WY13,KarthikParter21DeterministicRPC}
as a proxy for the $G{-}F$.
RPCs are families of spanning subgraphs of $G$
such that for each failing set $F$, $|F| \le f$,
and pair of vertices with a replacement path with at most $L$ edges,
there exists one graph in the family that contains the path but no edge of $F$.
Their construction realizes this with $(3L)^{f+o(1)}$ graphs.\footnote{
	This needs some reconstruction of the results in 
	\cite[Section 4.2 \& 5.3]{Bilo24ApproxDSOSubquadraticTheoretiCS} 
	as they use constant $f$.
}
We do not need the full power of replacement paths coverings to find the pivots of the two types.
We give a much simpler solution based on what we call \emph{pivot trees}. 
The advantage is that they only need to be computed for the original graph $G$.
Even when incorporating the size of the pivot trees,
this allows us to reduce the space for this part of the data structure
by roughly a factor $3^{f+o(1)} L^{1+o(1)}$. 

Consider $\ball_{G}(u,\lambda/2)$ around $u$ in the original graph $G$.
It can be computed using a breath-first search.
We modify it to stop as soon at it finds the first pivot from $B_1$ (if there is one).
Note that $\ball_{G}(u,\lambda/2) \cap B_1 = \emptyset$ 
implies $|\ball_{G}(u,\lambda/2)| \le L^f$ w.h.p.
That means, the BFS explores all of $\ball_{G}(u,\lambda/2)$ in time $O(L^{2f})$
or finds a pivot of the first type even before that.

Suppose first that the search does not find such a pivot.
With high probability, the set $\ball_{G}(u,\lambda/2)$ is sparse
and contains at most $\Otilde(fL^f/\lambda)$ pivots of the \emph{second} type from $B_2$.
We store this set $\ball_{G}(u,\lambda/2) \cap B_2$.
Now assume that we have $\ball_{G}(u,\lambda/2) \cap B_1 \neq \emptyset$.
Let $p_G(u) \in B_1$ be the pivot closest to $u$ that was found by the search.
Let $P$ the shortest $u$-$p_G(u)$-path in $G$.
It has at most $\lambda/2$ edges.
For each $e \in E(P)$, we create a child node in which we conduct the same computation
for the ball $\ball_{G-e}(u,\lambda/2)$.
Note that the new closest pivot $p_{G-e}(u) \in B_1$ may not be the same as $p_G(u)$.
This is the major difference to the trees used in \Cref{sec:framework_DSO_short}
that always store paths between the same two endpoints.
It still holds that the shortest $u$-$p_{G-e}(u)$-path in $G{-}e$ has at most $\lambda/2$ edges.
We iterate this construction until the pivot tree reached depth $f$.
In each child node, a new edge from the previous path from $u$ to its closest pivot is removed.

We build such a pivot tree for every $u \in V$.
Each one has $O(\lambda^f)$ nodes.
They store either a path with $\lambda/2$ edges
or some set $\ball_{G-F}(u,\lambda/2) \cap B_2$.
Using $2 \le f$ and $\lambda \le L$, we get $\lambda/2 = \Otilde(fL^f/\lambda)$.
The total size of all trees is thus $\Otilde(f\lambda^{f-1}L^f n)$.
They can be computed in time $O(\lambda^f L^{2f}n)$.

Additionally, for every pivot $p \in B_2$ of the second type, the preprocessing computes
the shortest-path tree $T_p$ rooted at $p$ in the original graph $G$.
For each $T_p$, the data structure of Bender and Farach-Colton~\cite{BenderFarachColton00LCARevisited}
is constructed that supports lowest common ancestor (LCA) queries in constant time.
Preparing all LCA-data structures together takes time $O(|B_2| \nwspace (m+n)) = \Otilde(f mn/\lambda)$
and storing them takes space $O(|B_2| \nwspace n) = \Otilde(f n^2/\lambda)$.

W.h.p.\ the preprocessing algorithm spent time 
$\mathsf{T}_L  + O(\lambda^f L^{2f}n) + 
		\Otilde(f mn/\lambda)$ 
so far and stored information taking up space
\begin{equation*}
	\mathsf{S}_L + |B_1| + |B_2| + \Otilde(f\lambda^{f-1}L^f n)
		+ \Otilde\!\left(f \frac{n^2}{\lambda} \right)
		= \mathsf{S}_L + \Otilde(f\lambda^{f-1}L^f n)
		+ \Otilde\!\left(f \frac{n^2}{\lambda} \right).
\end{equation*}
\vspace*{.5em}

\noindent
\textbf{Fault-tolerant trees again.}
We have already utilized the fault-tolerant trees of Chechik et al.~\cite{ChCoFiKa17} 
in our proof of \Cref{thm:short_distance_DSO}.
There, every edge of the path included in a node formed its own segment.
The main construction in \cite{ChCoFiKa17} extended this to segments of exponentially increasing length.
Bilò et al.~\cite{Bilo24ApproxDSOSubquadraticTheoretiCS} later introduced a hybrid form
in which the first and last $\lambda$ edges are on their own and longer segments only 
appear in the middle part.
We build on the latter version.
The main idea behind the switch from single edges to longer segments is to reduce
the storage space of an FT-tree.
Any segment is going to be identified by its endpoints, so reducing their number lowers the space.
However, this also leads to a more coarse-grained picture at query time
as failing whole segments can lead to much more than $f$ failures
overloading the underlying $f$-DSO.
The segments thus need to be structured in a way that even if they fail completely
certain paths are still guaranteed to survive.
We make this statement precise in \Cref{lem:key_property} 
after describing the query algorithm in \Cref{subsec:general_DSO_query}.

For the construction, we first have to define certain types of concatenated paths.
Afek, Bremler-Barr, Kaplan, Cohen, and Merritt~\cite[Theorems~1 \& 2]{Afek02RestorationbyPathConcatenation_journal}
showed that every replacement path in $G-F$ consists of at most $|F|+1$ shortest paths
of the original graph $G$, possibly interleaved with $|F|$ edges in case $G$ is weighted.
We call a path of this structure $|F|$-\emph{decomposable}.

\begin{definition}[$\ell$-decomposable paths]
\label{def:decomposable} 
	If $G$ is unweighted, an $\ell$\emph{-decomposable path} is
	a concatenation of at most $\ell+1$ shortest paths of $G$.
	If $G$ is edge-weighted, an $\ell$\emph{-decomposable path} may have
	an additional edge between any pair of consecutive constituting shortest paths.
\end{definition}

\noindent
It is easy to see that any $\ell$-decomposable path 
is also $\ell'$-decomposable 
for any $\ell' \ge \ell$.
Also, any subpath of an  $\ell$-decomposable path (with granularity $\lambda$)
is again $\ell$-decomposable 
%
The $(2f{+}1)$-decomposable paths 
are of special interest to us
since they comprise all concatenations of up to two replacement paths in $G-F$.
Let $A \subseteq E$ be a set of edges, possibly much more than $f$,
and let $s,t \in V$ be two vertices.
We use $d^{(2f{+}1)}(s,t,A)$ to denote the length of the shortest $(2f{+}1)$-decomposable path
from $s$ to $t$ in $G-A$; or $+\infty$ if no such path exists.

\begin{definition}[$\ell$-expath with granularity $\lambda$]
\label{def:l-expath_granularity}
  An $\ell$-\emph{expath with granularity} $\lambda$ is a path
  $P_a \circ P_b \circ P_c$
  such that $P_a$ and $P_c$ contain at most $\lambda$ edges each, 
  while $P_b$ is a concatenation of $2\log_2(n){+}1$
  $\ell$-decomposable paths
  such that, for every $0 \le i \le 2\log_2 n$,
  the the $i$-th $\ell$-decomposable path has at most $\min(2^i, 2^{2\log_2(n) - i})$ edges.
\end{definition}

\noindent
\Cref{def:l-expath_granularity} includes the case where some or all $\ell$-decomposable paths in $P_b$
are empty, only the maximum number of edges is bounded.
Bilò et al.~\cite[Section 7]{Bilo24ApproxDSOSubquadraticTheoretiCS} 
gave and $\Otilde(fm)$-time algorithm to compute the shortest $(2f{+}1)$-expath with granularity $\lambda$ in any subgraph of $G$
when given access to the original pair-wise distances in $G$.
They come annotated with the constituting $(2f{+}1)$-decomposable paths and, in turn, 
with their shortest paths in $G$ (and interleaving edges).

The last things we need before defining FT-trees are that of netpoints, segments, and parts.

\begin{definition}[Netpoints with granularity $\lambda$]
\label{def:netpoints}
  Let $P = (v_1, v_2, \ldots, v_\ell)$ be a path.
  If $|P| \leq \lambda$, then the \emph{netpoints of $P$ with granularity $\lambda$}
  are all vertices in $V(P)$.
  Otherwise, define $p_{\text{left}}$ to be all pairs of consecutive vertices $v_j,v_{j+1} \in V(P)$ 
  with $\frac{\lambda}{2} \le j \le \ell-\frac{\lambda}{2}$ for which there exists an integer $i \ge 0$
  such that $\length{P[v_{\frac{\lambda}{2}}..v_j]} < (1{+}\frac{\eps}{24})^i 
  	\le \length{P[v_{\frac{\lambda}{2}}..v_{j+1}]}$.
  Analogously, let $p_{\text{right}}$ be all vertices $v_j,v_{j-1} \in V(P)$ such that
  $\length{P[v_j..v_{\ell-\frac{\lambda}{2}}]} < (1{+}\frac{\eps}{24})^i 
  	\le \length{P[v_{j-1}..v_{\ell-\frac{\lambda}{2}}]}$
  for some $i$.
  The \emph{netpoints of P with granularity $\lambda$} are all vertices in
  $\{v_1,\ldots,v_{\frac{\lambda}{2}}\} \cup p_{\text{left}} \cup p_{\text{right}} 
  	\cup \{v_{\ell-\frac{\lambda}{2}},\ldots, v_\ell\}$.
\end{definition}

The intuition is as follows.
The first and last $\lambda/2$ vertices on the path $P$ are always netpoints.
Now consider the central subpath $P[v_{\frac{\lambda}{2}} .. v_{\ell-\frac{\lambda}{2}}]$.
Let $x$ be a power of $1{+}\frac{\eps}{24}$.
Among all vertices on the subpath that have distance at most $x$
from $v_{\frac{\lambda}{2}}$ mark the one furthest away as well as its successor.
Doing this for all possible powers $x$ gives the vertices in  $p_{\text{left}}$.
The same marking scheme starting from $v_{\ell-\frac{\lambda}{2}}$ on the other end of the subpath
gives $p_{\text{right}}$.

\begin{definition}[Segments and parts]
\label{def:segments}
  Let $P$ be any path.
  A \emph{segment} of $P$ is a subpath between consecutive netpoints with granularity $\lambda$.
  For any edge $e \in E(P)$, the segment of $P$ containing $e$ is denoted $\seg_{\lambda}(e,P)$.
  Suppose $P$ is an $(2f{+}1)$-expath with granularity $\lambda$.
  Its defining subpath $P_b$ 
  consists of $2 \log_2(n) + 1$ $(2f{+}1)$-decomposable paths
  that in turn each consists of up to $2f+1$ shortest paths $P_i$.
  A \emph{part} of $P$ is a maximal subpath of any of the $P_i$ that does not cross netpoints.
\end{definition}

The first and last $\lambda/2$ edges of a path each form their own segment.
In the central subpath, there are single-edge segments, but also longer ones
whose length grows exponential in $1{+}\frac{\eps}{24}$.
However, marking the netpoints both ends of the subpath ensures that these do not become too large.
The subdivision of the segments into parts is to align these building blocks with the
structure of expaths.
This is of course only necessary for segments with more than one edge.

We use $[x,y]$ for a part that reaches from vertex $x$ to $y$.
By definition, $[x,y]$ is the unique shortest $x$-$y$-path in $G$.
There are fewer then $\lambda + 2\log_{1+\frac{\varepsilon}{24}} (n) = \lambda + O(\log(n)/\varepsilon)$
segments on any path.
Moreover, there exists a universal constant $D > 0$ such that any $(2f{+}1)$-expath with granularity $\lambda$
has at most $\lambda + D f \log_2(n)/\varepsilon$ many parts.
We sometime also use an alternative upper bound on the number of parts
that is independent of $\varepsilon > 0$.
Namely, it is $\Otilde(fn)$
since each of the $O(\log n)$ decomposable paths in $P_b$ 
consists of $O(f)$ shortest (i.e., simple) paths. 

We finally turn to the fault-tolerant trees with granularity $\lambda$.
Each such tree belongs to a pair of vertices $u,v \in V$.
In turn, each node $\nu$ in the tree $FT_{\lambda}(u,v)$ is associated with a path $P_\nu$,
which is a shortest $(2f{+}1)$-expath with granularity $\lambda$ from $u$ to $v$
in the graph $G-A_{\nu}$, where $A_{\nu} \subseteq E$ is a set of edges
(with cardinality possibly much more than $f$).
We have $A_{\nu} = \emptyset$ in the root.
The path $P_{\nu}$ is computed with the algorithm of Bilò et al.~\cite{Bilo24ApproxDSOSubquadraticTheoretiCS}
and, with additional linear scan,
annotated with its netpoints, segments, and parts.
For each part $[x,y]$, we store two pointers.
One to the closest netpoint that comes before $x$ (including $x$ itself),
and one to the closest netpoint after $y$ (including).
The further processing of $[x,y]$ depends on the number of its edges.
If it has more than $L$ edges it contains some pivot $p \in B_1$ w.h.p.,
where we use that $L \ge \lambda$.
We store (the ID of) $p$ with $[x,y]$.
Otherwise, if $[x,y]$ has at most $L$ edges, we store the original graph distance $d(x,y)$.
Computing and post-processing $P_{\nu}$ takes time 
$\Otilde(fm) + O(n) + \Otilde(fn) = \Otilde(fm)$.
Here, the third term is the alternative upper bound on the number of parts.
The stored information is $O(\lambda + f \log(n)/\varepsilon)$, a constant number of words per part.

To construct the FT-tree with granularity, we create a child node $\mu$ of $\nu$ for each segment $S$ of $P_{\nu}$.
We set the new set of edges in the child as $A_{\mu} = A_{\nu} \cup E(S)$.
We continue with this process until depth $f$.
If the vertices $u$ and $v$ become disconnected in one of the graphs $G-A_{\nu}$,
we mark this by $\nu$ being a leaf that does not store any path.
In total, the tree $FT_{\lambda}(u,v)$ has $O(\lambda^f) + O(\log n/\varepsilon)^f$ nodes
and thus be computed in time 
$\Otilde(fm) \cdot \big(O(\lambda^f) + O(\frac{\log n}{\varepsilon})^f \big)$
and stored using space 
$(\lambda + O(f \frac{\log n}{\varepsilon}))
	\cdot \big(O(\lambda^f) + O(\frac{\log n}{\varepsilon})^f \big)$.
We leave the expressions as they are for now.
They will be simplified after fixing $\lambda$ and
the number of trees.

The description above holds verbatim also for granularity $0$ instead of $\lambda$,.
We refer to this as FT-trees \emph{without granularity} etc.
\vspace*{.5em}

\noindent
\textbf{Finishing the preprocessing.}
From now on, we limit the granularity to $\lambda = \Theta(\varepsilon L)$.
Recall that w.h.p.\ we have $\Otilde(fn/L^f)$ pivots of the first type.
We construct an FT-tree with granularity $\lambda$ for every pair in $B_1 \times B_1$.
Further, an FT-tree without granularity is built for each pair a pivot of the second type
and an arbitrary vertex, that is, for $B_2 \times V$.
W.h.p.\ the former trees take space
\begin{align*}
	&|B_1|^2 \cdot \left(\lambda + O\!\left(f \frac{\log n}{\varepsilon} \right) \right)
		\left( O(\lambda^f) + O\!\left(\frac{\log n}{\varepsilon}\right)^f \right)\\
		&= \Otilde\!\left( f^2 \frac{n^2}{L^{2f}} \right) \cdot
			\left( O(\lambda^{f+1}) +  \lambda \cdot O\!\left( \frac{\log n}{\varepsilon}\right)^f
				+ O\!\left(f \lambda^f\right) \cdot O\!\left(\frac{\log n}{\varepsilon} \right) + \Otilde(f) \cdot 					O\!\left(\frac{\log n}{\varepsilon} \right)^{f+1} \right)\\
		&= \Otilde\!\left(f^2 \frac{n^2 \lambda^{f+1}}{L^{2f}} \right) 
			+ \Otilde\!\left(f^2 \frac{n^2 \lambda}{L^{2f}} \right) 
				O\!\left( \frac{\log n}{\varepsilon}\right)^f
			+ \Otilde\!\left(f^3 \frac{n^2 \lambda^f}{L^{2f}} \right) 
				O\!\left(\frac{\log n}{\varepsilon} \right) 
			+ \Otilde\!\left(f^3 \frac{n^2}{L^{2f}} \right)
				O\!\left(\frac{\log n}{\varepsilon} \right)^{f+1}\\
		&= \Otilde\!\left(f^2 \frac{n^2}{L^{f-1}} \right) 
			+ \Otilde\!\left(f^2 \frac{n^2}{L^{2f-1}} \right)\!
				O\!\left( \frac{\log n}{\varepsilon}\right)^f
			+ \Otilde\!\left(f^2 \frac{n^2}{L^{f-1}} \right)\!
				O\!\left(\frac{\log n}{\varepsilon} \right)  
			+ \Otilde\!\left(f^2 \frac{n^2}{L^{2f-1}} \right)\!
				O\!\left(\frac{\log n}{\varepsilon} \right)^{f+1}\!.
\end{align*}
We used $\lambda \le L$, $\lambda = \Omega(\varepsilon L)$, and $f \le L$
for the transformations in the last line.
For the FT-trees between pair in $B_2 \times V$ the space is
\begin{equation*}
	n |B_2| \cdot O\!\left(f \frac{\log n}{\varepsilon} \right)
		O\!\left(\frac{\log n}{\varepsilon} \right)^f
		= \Otilde\!\left(f^2 \frac{n^2}{\lambda} \right) 
			\cdot O\!\left(\frac{\log n}{\varepsilon} \right)^{f+1}
		= \Otilde\!\left(f^2 \frac{n^2}{L} \right) 
			\cdot O\!\left(\frac{\log n}{\varepsilon} \right)^{f+2}.
\end{equation*}
We assumed $f \ge 2$, which gives $n^2/L^{f-1} \le n^2/L$.
Therefore, all terms for both sets of FT-trees are dominated by 
$\Otilde(f^2 n^2/L) \cdot O(\log n/\varepsilon)^{f+2}$.

Recall that we also store the DSO for hop-short distances,
the pivot trees, and the LCA data structures.
The latter is also dominated by the FT-trees.
The total space is
\begin{equation*}
	\mathsf{S}_L + \Otilde(f\lambda^{f-1}L^f n)
		+ \Otilde\!\left(f^2 \frac{n^2}{L}\right) \cdot
			O\!\left(\frac{\log n}{\varepsilon}\right)^{f+2}.
\end{equation*}
The statement in \Cref{thm:general_DSO} follows from this via $\lambda \le L$.

A similar calculation shows that the construction time for the FT-trees is
\begin{align*}
	&\Otilde(fm) \cdot
		\left( |B_1|^2 \cdot \left( O(\lambda^f) 
			+ O\!\left(\frac{\log n}{\varepsilon}\right)^f \right)
		+ n |B_2| \cdot O\!\left(\frac{\log n}{\varepsilon}\right)^f\right)\\
		&\quad= \Otilde(fm) \cdot
			\left(\Otilde\!\left(f^2 \frac{n^2}{ L^{f}} \right)
			+ \Otilde\!\left(f^2 \frac{n^2}{L^{2f}} \right) 
				O\!\left(\frac{\log n}{\varepsilon}\right)^f 
			+ \Otilde\!\left( f\frac{n^2}{\lambda}\right) 
				O\!\left(\frac{\log n}{\varepsilon}\right)^f\right)\\
		&\quad=	\Otilde(fm) \cdot
			\Otilde\!\left(f \frac{n^2}{\lambda} \right) O\!\left(\frac{\log n}{\varepsilon}\right)^f
		 = \Otilde\!\left(f^2 \frac{mn^2}{L} \right) O\!\left(\frac{\log n}{\varepsilon}\right)^{f+1}.
\end{align*}
The last summary uses $\lambda = \Omega(\varepsilon L)$ again.
Considering the preprocessing of the input DSO for hop-short distances, the pivot trees, and LCA structures,
the resulting time is
$\mathsf{T}_L  + O(\lambda^f L^{2f}n) + \Otilde(f^2 mn^2/L) \cdot O(\log n/\varepsilon)^{f+1}
	= \mathsf{T}_L  + O(L^{3f}n) + \Otilde(f^2 mn^2/L) \cdot O(\log n/\varepsilon)^{f+1}$.

\subsection{Query Algorithm and Time}
\label{subsec:general_DSO_query}

The resulting $f$-DSO (for arbitrary distances) is queried with a triple $(s,t,F)$
with $s,t \in V$ being two vertices and $F \subseteq E$ is a set of at most $f$ edges.
The task is to approximate the replacement distance $d(s,t,F)$, 
which is the distance between $s$ and $t$ in $G-F$.
We first describe the high-level algorithm, before going into the details.
For the further description, let $u,v \in V$ be any two vertices.
Recall that we use $\widehat{d^{\le L}}(u,v,F)$
for the estimate of the $u$-$v$-distance in $G-F$ 
when restricted to paths with at most $L$ edges.
It is reported by the short-distance DSO in time $\mathsf{Q}_L$.
In this section, we further use $FT_{\lambda}(u,v,F)$ for the result we get
from querying the FT-tree with granularity $FT_{\lambda}(u,v)$ with the set $F$.
Analogously, we use $FT_{0}(u,v,F)$ for the value of an FT-tree without granularity.
We will see later that the query time for FT-trees with or without granularity,
which we respectively denote by $\mathsf{Q}_{FT}^{(\lambda)}$ and $\mathsf{Q}_{FT}^{(0)}$,
are asymptotically larger than $\mathsf{Q}_L$ and than $f^2$.
We use this fact for simplifications in the $O$-notation below.

When given a query $(s,t,F)$,
an auxiliary weighted complete graph $H$ on the vertex set $V(H) = \{s,t\} \cup V(F)$ is constructed,
where $V(F)$ is the set of endpoints of edges in $F$.
For each $\{u,v\} \in E(H) = \binom{V(H)}{2}$, the query algorithm computes a weight $w_H(u,v)$.
The eventual answer of our DSO is the distance $d_H(s,t)$ in $H$.
\vspace*{.5em}

\noindent
\textbf{Computing the weights $w_H(u,v)$.}
The query algorithm computes $w_H(u,v)$ with one of two methods,
depending on whether there are pivots of the first type available 
in the vicinity of $u$ and $v$ in $G-F$.
The respective pivot trees are used for this decision.
It is checked, whether the stored path $P$ from $u$ to the closest pivot $p_{G}(u) \in B_1$ contains
any edge of $F$.
If not, then $p_{G}(u) \in \ball_{G-F}(u, \lambda/2)$
and it is also the closest pivot of the first type in $G{-}F$.
Otherwise, we recurse to an arbitrary child node corresponding to some edge in $E(P) \cap F$.
After at most $f$ levels,
the algorithm either learns a pivot from $B_1$ that is close to $u$ in $G{-}F$,
or it arrives in a leave node corresponding to some $F' \subseteq F$ 
such that $\ball_{G-F'}(u, \lambda/2)$ is too sparse to intersect $B_2$.
In the former case, let $p_u \in B_1$ be the closest pivot.
In the latter case, the query algorithm gets access to the pivots of the second type
in $\ball_{G-F'}(u, \lambda/2) \cap B_2$.
Note that we have $\ball_{G-F'}(u, \lambda/2) \cap B_2 \supseteq \ball_{G-F}(u, \lambda/2) \cap B_2$.
That means, all pivots of the second type that are close to $u$ are recovered
(and maybe some more). 
Each check for the emptyness of  $E(P) \cap F$ takes time $O(f)$,
so the whole time spend in the pivot tree is $O(f^2)$.
The query algorithm does the same for the other vertex $v$.
We say the computation of the weight $w_H(u,v)$ is in the \emph{dense ball case}
if two pivots $p_u,p_v \in B_1$ are obtained.
Otherwise, we say it is in the \emph{sparse ball case}.

First, assume we are in the dense ball case.
We stored an FT-tree with granularity  $\lambda$ for the pair $(p_u,p_v) \in B^2_2$.
Recall that we use $FT_\lambda(p_u,p_v,F)$ for the value returned by that FT-tree
when queried with the failure set $F$.
The weight is defined as
\begin{equation}
\label{eq:dense_balls}
	w_{H}(u,v) = \min\!\left( \widehat{d^{\le L}}(u,v,F),\  FT_\lambda(p_u,p_v,F) + \lambda \right).
\end{equation}
It is computable in time $O(f^2 + \mathsf{Q}_L + \mathsf{Q}_{FT}^{(\lambda)}) = O(\mathsf{Q}_{FT}^{(\lambda)})$.

In the sparse ball case, there exists a vertex $x \in \{u,v\}$
such that the algorithm recovered $\ball_{G-F'}(u, \lambda/2) \cap B_2$.
Let $y \in \{u,v\}{\setminus}\{x\}$ be the other endpoint.
The weight is
\begin{equation}
\label{eq:sparse_balls}
	w_{H}(u,v) = \min\bigg( \widehat{d^{\le L}}(u,v,F), 
		\min_{p \in \ball_{G-F'}(u, \lambda/2) \cap B_2} 
			\left(\widehat{d^{\le L}}(x,p,F) + FT_0(p,y,F)\right)\bigg).
\end{equation}
That means, the query algorithm minimizes the sum $\widehat{d^{\le L}}(x,p,F) + FT_0(p,y,F)$
over all pivots $p$ of the second type in the sparse ball around $x$,
it then compares the result with the estimate $\widehat{d^{\le L}}(u,v,F)$ of the input DSO
for hop-short distances.
The edge weight is the smaller of the two values.
Since w.h.p.\ $|\ball_{G-F'}(u, \lambda/2) \cap B_2| = O(fL^f/\lambda)$,
the weight can be obtained in time 
$\Otilde(f^2 + \mathsf{Q}_L + f\frac{L^{f}}{\lambda} \nwspace ( \mathsf{Q}_L + \mathsf{Q}_{FT}^{(0)})) = 
	O(f\frac{L^{f}}{\lambda} \nwspace \mathsf{Q}_{FT}^{(0)})$.

The auxiliary graph $H$ has $O(f^2)$ edges, so obtaining all weights takes 
$\Otilde(f^2 ( \mathsf{Q}_{FT}^{(\lambda)}+ f \frac{L^{f}}{\lambda} \mathsf{Q}_{FT}^{(0)}))$ time.
This is also the order of the whole query time of our general DSO
since the $\Otilde(f^2)$ time to compute the distance $d_H(s,t)$ with Dijkstra's algorithm
is immaterial.
\vspace*{.5em}

\noindent
\textbf{Querying the FT-trees.}
So far, we used the query procedure of the FT-trees as a black box.
We now describe how the value $FT_{\lambda}(u,v,F)$ is computed,
the case without granularity follows immediately by setting $\lambda = 0$.
The overall structure of the query algorithm is similar to \cite{Bilo24ApproxDSOSubquadraticTheoretiCS},
but we generalize it to incorporate the stretch $(\alpha,\beta)$ of the
the underlying DSO for hop-short distances.
Recall that $D$ is the universal constant such that any 
$(2f{+}1)$-expath with granularity $\lambda$
has at most $\lambda + D f \log_2(n)/\varepsilon$ many parts.

\begin{lemma}
\label{lem:weaker_guarantee}
	Let $\nu$ be a node of the FT-tree $FT_{\lambda}(u,v)$ in which an $(2f{+}1)$-expath $P_{\nu}$
	from $u$ to $v$ is stored.
	There is an algorithm that either certifies that
	$d(u,v,F) \le \alpha \cdot |P_\nu| 
		+ D f \frac{\log_2(n)}{\varepsilon} \cdot \beta$
	or finds the segment $\seg_{\lambda}(e, P_{\nu})$ for some edge $e \in F \cap E(P_{\nu})$.
	The algorithm runs in time $O(f\lambda + f \frac{\log n}{\varepsilon} (\mathsf{Q}_L + f))$.
\end{lemma}

\begin{proof}
	The algorithm first probes each parts of $P_\nu$,
	whether it can certify that it contains an edge from $F$.
	Recall that any part $[x,y]$ is the unique shortest path 
	from $x$ to $y$ in the original graph $G$.
	
	If the part consists of only a single edge, it is trivial to check whether
	$\{x,y\} \in F$ in time $O(f)$.
	In particular, this is the case for the first and last $\lambda/2$ edges of $P_{\nu}$.
	If $[x,y]$ has at most $L$ edges, we stored the original distance $d(x,y)$ with it.
	The algorithms queries the underlying DSO for hop-short distances to get the estimate
	$\widehat{d^{\le L}}(x,y,F)$.
	The oracle has stretch $(\alpha,\beta)$.
	If $\widehat{d^{\le L}}(x,y,F) > \alpha \cdot d(x,y) + \beta$, 
	then the part must contain a failing edge.
	
	Otherwise, $[x,y]$ has more than $L$ edges.
	We stored a pivot $p \in B_2$ of the second type that lies on $[x,y]$.
	By the uniqueness of paths $[x,y]$ is the concatenation of the unique shortest path from $x$ to $p$
	in $G$ and the one from $p$ to $y$.
	Using the LCA data structure for the shortest-path tree $T_p$ rooted at $p$,
	the algorithm can check in time $O(f)$ 
	whether any edge $e \in F$ also lies on $[x,y]$.
	
	Let $\mathcal{P}$ be the collection of all parts of $P_\nu$ that are only a single edge,
	$\mathcal{L}$ the ones with more than $L$ edges,
	and $\mathcal{R}$ the remaining parts. 
	The path $P_{\nu}$ is the concatenation of the parts in 
	$\mathcal{P} \cup \mathcal{L} \cup \mathcal{R}$ (in a suitable order).
	Let instead $P$ be the path in which the parts in $\mathcal{P} \cup \mathcal{L}$ remain the same,
	but each part $[x,y] \in \mathcal{R}$ is exchanged by the replacement path $P(x,y,F)$ in $G-F$.
	
	Checking all the at most $\lambda + O(\log(n)/\varepsilon)$ many parts takes total time 
	$O(f\lambda + f (\log n) (\mathsf{Q}_L + f)/\varepsilon)$.
	If a part $[x,y]$ is found to contain an edge from $F$,
	the corresponding segment is given by the pointer to the closest netpoints before $x$
	and after $y$.
	Note that we do not need to know the exact failing edge
	since the whole part lies in the same segment.
	Otherwise, the algorithm verified that the path $P$ defined above lies in $G{-}F$ and has length
	\begin{multline*}
		|P| = \sum_{[x,y] \in \mathcal{P} \cup \mathcal{L}} d(x,y) + 
			\sum_{[x,y] \in \mathcal{R}} d(x,y,F)
			\le \sum_{[x,y] \in \mathcal{P} \cup \mathcal{L}} d(x,y) + 
				\sum_{[x,y] \in \mathcal{R}} (\alpha \cdot d(x,y) + \beta)\\
			\le \left( \sum_{[x,y] \in \mathcal{P} \cup \mathcal{L} \cup \mathcal{R}} 
				\alpha \cdot d(x,y) \right) + |\mathcal{R}| \nwspace \beta
			\le \alpha \nwspace |P_{\nu}| 
				+ D f \frac{\log_2(n)}{\varepsilon} \cdot \beta. \qedhere
	\end{multline*}
\end{proof}

After this setup, we describe the query algorithm of FT-trees.
When queried with the failure set $F$, $FT_{\lambda}(u,v)$ is traversed starting with the root.
In each node $\nu$, the algorithm first checks whether $u$ and $v$ are connected,
that is, whether a path $P_{\nu}$ is present.
If there is no path, the query returns with the answer $+\infty$.
Otherwise, the algorithm from \Cref{lem:weaker_guarantee} is run with set $F$.
If a segment with a failing edge is found, the traversal recurses on the child node of $\nu$
that corresponds to that segment;
otherwise, the value 
$\alpha \nwspace |P_\nu| + D f \frac{\log_2(n)}{\varepsilon} \nwspace \beta$
is reported.
If a leaf $\nu^*$ of $FT_{\lambda}(u,v)$ is reached that way in which a path $P_{\nu*}$ is stored,
the output is $|P_{\nu^*}|$.

The height of the tree $FT_{\lambda}(u,v)$ is $f$, so the total query time is
$\mathsf{Q}_{FT}^{(\lambda)} = O(f^2\lambda + f^2\, \frac{\log n}{\varepsilon} (\mathsf{Q}_L + f))$,
or $\mathsf{Q}_{FT}^{(0)} = O(f^2 \, \frac{\log n}{\varepsilon} (\mathsf{Q}_L + f))$
without granulartiy.
From $\lambda = \Theta(\varepsilon L)$ and $f \ge 2$, we get that
$f\frac{L^{f}}{\lambda} \mathsf{Q}_{FT}^{(0)}$ dominates $\mathsf{Q}_{FT}^{(\lambda)}$,
Therefore, the total query time of our (general) DSO is

\begin{multline*}
	\Otilde\!\left(f^2 \left( \mathsf{Q}_{FT}^{(\lambda)}
		+ f \frac{L^{f}}{\lambda} \mathsf{Q}_{FT}^{(0)} \right)\! \right)
	= \Otilde\!\left(f^3 \frac{L^{f}}{\lambda} \mathsf{Q}_{FT}^{(0)} \right)\\
	= \Otilde\!\left(f^5 \frac{L^{f}}{\lambda} \frac{\log n}{\varepsilon} (\mathsf{Q}_L + f) \right)
	= \Otilde\!\left(f^5 \frac{L^{f-1}}{\varepsilon^2} (\mathsf{Q}_L + f) \right).
\end{multline*}

Intuitively, a parent-child traversal during the query of an FT-tree 
simulates the failure of the whole segment $\seg_{\lambda}(e, P_{\lambda})$
instead of only the single edge $e \in F$.
Storing information only for the segments, or more accurately for the parts,
reduces the space needed to store an FT-tree but in turn makes the query answer less precise.
Let $\nu$ be the node that produces the output.
The set $A_{\nu}$ of edges that are absent in $\nu$ may be much larger than $F$.
However, due to definition of segments, the elements in $A_{\nu}$ are heavily clustered.
Therefore, if a path is \emph{decomposable and far away from} $F$ 
(see \Cref{def:decomposable,def:trapezoid}) 
it even avoids all of $A_{\nu}$.
More precisely, the FT-trees with their associated query algorithm
have the following key property \cite[Lemmas~5.9 \& 6.6]{Bilo24ApproxDSOSubquadraticTheoretiCS}

\begin{lemma}[\cite{Bilo24ApproxDSOSubquadraticTheoretiCS}]
\label{lem:key_property}
	Let $u,v \in V$ be two vertices and $F \subseteq E$ a set of at most $f$ edges,
	and $P$ be the shortest $(2f{+}1)$-decomposable $u$-$v$-path
	that is far away from $F$.
	\begin{enumerate}[(i)]
		\item Let $p_u \in \ball_{G-F}(u,\lambda/2) \cap B_1$ and 
			$p_v \in \ball_{G-F}(v,\lambda/2) \cap B_1$, and $\nu$ be the output-node
			of $FT_{\lambda}(p_u,p_v)$ when queried with $F$.
			Then, the path $P(p_u,u,F) \circ P \circ P(v,p_v,F)$ exists in  $G-A_{\nu}$.
		\item Let $\nu'$ be the output-node of $FT_0(u,v)$ when queried with $F$.
			Then, the path $P$ exists in $G-A_{\nu'}$.
	\end{enumerate}
\end{lemma}

Observe that \Cref{lem:key_property}~\emph{(i)} applies to the tree $FT_{\lambda}(p_u,p_v)$
although its construction and query algorithm is completely independent of the endpoints $u$ and $v$
of the path $P$.
The only relation is that $p_u$ and $u$ (respectively, $p_v$ and $v$)
are connected by a path on at most $\lambda/2$ edges in $G-F$.
The granularity $\lambda$ of $FT_{\lambda}(p_u,p_v)$ means
that the first and last $\lambda/2$ edges of any path $P_{\nu}$
stored in a node $\nu$ of the FT-tree form their own segment.
Failing a segment in this pre-/suffix is equivalent to failing a single edge.
Intuitively, this cannot affect paths
that lie within $\ball_{G-F}(u,\lambda/2)$ (respectively, in $\ball_{G-F}(v,\lambda/2)$).
Conversely, for the exponentially increasing segments in the middle of $P_{\nu}$,
the safety area that ensures the survival of the path $P$
is given by the trapezoid $\tr_{G-F}^{\varepsilon/6}(P)$.
If this area is free of failures, meaning that $P$ is far away from $F$,
no segment $\seg_{\lambda}(e, P_{\lambda})$ can reach $P$.
Making this intuition rigorous, however, requires some technical machinery.
This is the reason to use expaths $P_{\nu}$ (instead of mere decomposable paths) and
for the constants in the definition of trapezoids (\Cref{def:trapezoid})
and netpoints (\Cref{def:netpoints}) differing by a multiplicative factor $4$.
The details are discussed in \cite{Bilo24ApproxDSOSubquadraticTheoretiCS}.

\subsection{Stretch}
\label{subsec:general_DSO_stretch}

We are left to show that our DSO has stretch $(\alpha (1{+}\varepsilon), \beta)$,
where $(\alpha,\beta)$ is the stretch of the input DSO for hop-short distances
and $\varepsilon > 0$ is the approximation parameter chosen at the beginning of the preprocessing.
The fine-tuned analysis in this section is the main improvement over \cite{Bilo24ApproxDSOSubquadraticTheoretiCS}.
It allows us to generalize $f$-DSOs to the setting of a positive additive stretch $\beta > 0$
and multiplicative stretch $\alpha \neq 3$.

We first prove bounds on the values 
$\widehat{d^{\le L}}(x,y,F), FT_{0}(x,y,F)$, and $FT_{\lambda}(x,y,F)$ 
that are used in the computation of the weight $w_H(u,v)$.
We then combine all the intermediate bounds in \Cref{lem:stretch_general_DSO} to prove the stretch.
We start by showing that none of the different values underestimate the true replacement distance.

\begin{lemma}
\label{lem:never_underestimate}
	Let $u,v,x,y \in V$ be vertices.
	Then, it holds that
	\begin{enumerate}[(i)]
		\item $\widehat{d^{\le L}}(x,y,F), FT_{0}(x,y,F), FT_{\lambda}(x,y,F) \ge d(x,y,F)$;
		\item $w_{H}(u,v) \ge d(u,v,F)$.
	\end{enumerate}	
\end{lemma}

\begin{proof}
	We first prove \emph{(i)}.
	This is immediate for $\widehat{d^{\le L}}(x,y,F)$ as the estimate for the 
	hop-bounded replacement distance is always at least the true replacement $d(x,y,F)$.
	
	If the FT-tree $FT_{\lambda}(x,y)$, when queried with set $F$, answers $+\infty$,
	this is clearly at least the replacement distance.
	Now suppose the query procedure stops in some node $\nu$ and the value 
	$\alpha \cdot |P_\nu| + D f \frac{\log_2(n)}{\varepsilon} \cdot \beta$
	is returned.
	Then, the algorithm in \Cref{lem:weaker_guarantee} certifies
	that this value is at least $d(x,y,F)$.
	Otherwise, the query reaches a leaf
	$\nu^*$ at depth $f$.
	The tree traversal thus found a segment for each failing edge in $F$.
	The path $P_{\nu^*}$ from $x$ to $y$ is thus disjoint from $F$,
	giving $FT_{\lambda}(x,y,F) = |P_{\nu^*}| \ge d(x,y,F)$.
	The same argument also holds for $FT_0(x,y,F)$.
	
	We now turn to \emph{(ii)}.
	First, assume that $w_H(u,v)$ is computed by \Cref{eq:dense_balls}
	using the pivots $p_u \in \ball_{G-F}(u,\lambda/2) \cap B_1$
	and $p_v \in \ball_{G-F}(v,\lambda/2) \cap B_1$.
	We need to show that $FT_{\lambda}(p_u,p_v,F) + \lambda \ge d(u,v,F)$.		
	If $p_u$ and $p_v$ are disconnected in $G{-}F$,
	we have $FT_{\lambda}(p_u,p_v,F) = +\infty$ by \emph{(i)} and we are done.
	Otherwise, there exists some replacement path $P(p_u,p_v,F)$.
	Let $P_{p_u}$ be the shortest path from $u$ to $p_u$ in $G{-}F$,
	and $P_{p_v}$ the shortest path from $p_v$ to $v$ in $G{-}F$.
	Note that both paths have at most $\lambda/2$ edges by the choice of $p_u,p_v$.
	In summary, $P_{p_u} \circ P(p_u,p_v,F) \circ P_{p_v}$ is some path from $u$ to $v$ in $G{-}F$.
	From Part~\emph{(i)}, we get that
	\begin{equation*}
		FT_{\lambda}(p_u,p_v,F) + \lambda \ge 
		\lambda/2 + d(p_u,p_v,F) + \lambda/2 \ge |P_{p_u}| + d(p_u,p_v,F) + |P_{p_v}| \ge d(u,v,F).
	\end{equation*}
	
	Finally, if $w_H(u,v)$ is computed via \Cref{eq:sparse_balls},
	then it equals $\widehat{d^{\le L}}(u,v,F)$ or
	$\widehat{d^{\le L}}(x,p,F) + FT_0(p,y,F)$ 
	for some $p \in B_2$ and $y \in \{u,v\}{\setminus}\{x\}$.
	The former is not smaller than $d(u,v,F)$ again using Part~\emph{(i)};
	the latter is at least $d(x,p,F) + d(p,y,F) \ge d(u,v,F)$
	by the triangle inequality.
\end{proof}

\Cref{lem:key_property} shows that certain structured paths survive
until the output node of an FT-tree, even though full segments are failed
in each parent-child transition.
This key property lends some importance to the class of $(2f{+}1)$-decomposable paths
that are far away from all failures in $F$, that is, 
whose trapezoid $\tr^{\varepsilon/6}_{G-F}(P)$
does not contain any endpoint of an edge in $F$ (\Cref{def:trapezoid}).
We use it to derive upper bounds on the return values
$FT_0(u,v,F)$ and $FT_{\lambda}(u,v,F)$.
For this, we need to expand the definition of the decomposable distance $d^{(2f+1)}(u,v,F)$.
We define $d^{(2f+1)}_{\varepsilon/6}(u,v,F)$
to be the length of the shortest $u$-$v$-path that is $(2f{+}1)$-decomposable
\emph{and far away from} $F$; 
or $d^{(2f+1)}_{\varepsilon/9}(u,v,F) = +\infty$ if no such path exists.
Note that $d^{(2f+1)}_{\varepsilon/6}(u,v,F) \ge d^{(2f+1)}(u,v,F) \ge d(u,v,F)$.
Further, recall that $D > 0$ is the universal constant used in the output value of FT-trees.

\begin{lemma}
\label{lem:FT_upper_bound}
	Let $u,v \in V$ be two vertices and $F \subseteq E$ a set of $|F| \le f$ edges.
	\begin{enumerate}[(i)]
		\item Let $p_u \in \ball_{G-F}(u, \lambda/2) \cap B_1$
			and $p_v \in \ball_{G-F}(v, \lambda/2) \cap B_1$ be pivots of the first type.
			Then, $FT_{\lambda}(p_u,p_v,F) + \lambda \le \alpha 
				\cdot d^{(2f+1)}_{\varepsilon/6}(u,v,F) + \alpha \lambda 
				+ D f \frac{\log_2(n)}{\varepsilon} \cdot \beta$.
		\item It holds that $FT_{0}(u,v,F) \le \alpha \cdot d^{(2f+1)}_{\varepsilon/6}(u,v,F) 
			+ D f \frac{\log_2(n)}{\varepsilon} \cdot \beta$.
	\end{enumerate}
\end{lemma}

\noindent
Similarly to \Cref{lem:key_property},
is important to note that the left-hand side of the inequality in \Cref{lem:FT_upper_bound}~\emph{(i)}
is in terms of the pair $(p_u,p_v)$ while the right-hand side is in terms of $(u,v)$.

\begin{proof}[Proof of \Cref{lem:FT_upper_bound}]
	Let $P$ be the shortest $(2f{+}1)$-decomposable $u$-$v$-path in $G{-}F$
	that is far away from $F$,
	that is, $|P| = d^{(2f+1)}_{\varepsilon/6}(u,v,F)$.
	If no such path exists, the lemma holds vacuously.
	
	We first prove Part~\emph{(ii)}.
	Let $\nu'$ the node of $FT_0(u,v)$ that produces the output when the tree is queried with $F$.
	By \Cref{lem:key_property}~\emph{(ii)}, $P$ only uses edges in $E{\setminus}A_{\nu'}$.
	Since $P$ is $(2{+}1)$-decomposable, it is also an $(2f{+}1)$-expath.
	By definition, $P_{\nu'}$ is the \emph{shortest} $(2f{+}1)$-expath without granularity 
	from $u$ to $v$ in $G-{A_{\nu'}}$,
	whence $|P_{\nu'}| \le |P|$.
	The query algorithm of $FT_0(u,v)$ returns either $|P_{\nu'}|$ or 
	$\alpha \cdot |P_{\nu'}| + D f \frac{\log_2(n)}{\varepsilon} \cdot \beta$
	depending on whether $\nu'$ is a leaf.
	Either way the return value is bounded by 
	$\alpha \cdot |P| + D f \frac{\log_2(n)}{\varepsilon} \cdot \beta =
		\alpha \cdot d^{(2f+1)}_{\varepsilon/6}(u,v,F)
			+ D f \frac{\log_2(n)}{\varepsilon} \cdot \beta$.
	
	For Part~\emph{(i)}, let $\nu$ be the output-node of $FT_{\lambda}(p_u,p_v)$ with query $F$.
	The argument is similar as before, 
	but we have to account for the different endpoints of $P_{\nu}$ and $P$.
	Consider the shortest (replacement) paths $P(p_u,u,F)$ and $P(v,p_v,F)$ in $G{-}F$.
	By \Cref{lem:key_property}~\emph{(i)}, the concatenation $Q = P(p_u,u,F) \circ P \circ P(v,p_v,F)$
	lies in $G{-}A_{\nu}$.
	Since $p_u \in \ball_{G-F}(u, \lambda/2) \cap B_1$ and $p_v \in \ball_{G-F}(u, \lambda/2) \cap B_1$,
	the paths $P(p_u,u,F)$ and $P(v,p_v,F)$ have at most $\lambda/2$ edges.
	Together with $P$ being $(2f{+}1)$-decomposable,
	this shows that $Q$ is an $(2f{+}1)$-expath with granularity $\lambda$ from $p_u$ to $p_v$.
	It follows that $|P_{\nu}| \le |Q|$ and
	\begin{multline*}
		FT_{\lambda}(p_u,p_v,F) + \lambda \le \alpha \cdot |Q| 
			+ D f \frac{\log_2(n)}{\varepsilon} \cdot \beta
		\le \alpha \cdot (|P| + \lambda) 
			+ D f \frac{\log_2(n)}{\varepsilon} \cdot \beta\\
		= \alpha \cdot d^{(2f+1)}_{\varepsilon/6}(u,v,F) + \alpha \cdot \lambda
				+ D f \frac{\log_2(n)}{\varepsilon} \cdot \beta.	\qedhere	
	\end{multline*}
\end{proof}

We now turn the above bounds on the return values of the FT-trees
into bounds on the edge weight $w_H(u,v)$ in $H$ 
depending on the existence of certain $u$-$v$-paths in $G-F$.
First, note that if there is a hop-short path from $u$ to $v$ in $G-F$,
then also the replacement path $P(u,v,F)$ has at most $L$ edges.
It is easy to see that in this case $w_H(u,v) \le \alpha \cdot d(u,v,F) + \beta$
since $\widehat{d^{\le L}}(u,v,F)$ is part of the minimization in both \Cref{eq:sparse_balls,eq:dense_balls} and the underlying DSO for hop-short distances
has stretch $(\alpha,\beta)$.
We are thus concerned with hop-long $u$-$v$-paths.
We starting with the ``dense ball'' case.
It is an easy corollary of \Cref{lem:FT_upper_bound}~\emph{(i)}.

\begin{corollary}
\label{cor:weight_long_far_away_dense}
	Define $\delta = \lambda/L$.
	Let $F \subseteq E$, $|F| \le f$, be a set of edges and $u,v \in V(F) \cup \{s,t\}$ vertices
	such that $w_H(u,v)$ is computed using \Cref{eq:dense_balls}.
	Suppose there exists an $u$-$v$-path $P$ in $G-F$
	that is the concatenation of at most two replacement paths, 
	hop-long, and far away from all failures in $F$.
	Then, with high probability it holds that 
	$w_H(u,v) \le \alpha (1{+}\delta) \cdot |P| +
		 D f \frac{\log_2(n)}{\varepsilon} \cdot \beta$.
\end{corollary}

\begin{proof}
	Replacement paths are $f$-decomposable.
	Therefore, the concatenation $P$ of at most two replacement paths is
	$(2f{+}1)$-decomposable.
	Since $P$ is also assumed to be far away from $F$,
	we get $|P| \ge d_{\varepsilon/6}^{(2f+1)}(u,v,F)$.
	Moreover, the path is hop-long, that is, $|P| \ge L$.
	
	There exists $p_u \in \ball_{G-F}(u, \lambda/2) \cap B_1$
	and $p_v \in \ball_{G-F}(u, \lambda/2) \cap B_1$ such that
	$w_H(u,v) \le FT_{\lambda}(p_u,p_v,F) + \lambda$ 
	since the weight is computed via \Cref{eq:dense_balls}.
	Combining \Cref{lem:FT_upper_bound}~\emph{(i)} with the lower bounds on $|P|$ gives
	\begin{multline*}
		w_H(u,v) \le \alpha \cdot d^{(2f+1)}_{\varepsilon/6}(u,v,F) + \alpha \cdot \lambda 
			+ D f \frac{\log_2(n)}{\varepsilon}\beta
			\le \alpha \cdot |P| + \alpha \cdot \delta L 
			+ D f \frac{\log_2(n)}{\varepsilon} \beta\\
			\le \alpha  (1{+}\delta) \cdot |P| 
				+ D f \frac{\log_2(n)}{\varepsilon} \beta. \qedhere
	\end{multline*}
\end{proof}

\begin{lemma}
\label{lem:weight_long_far_away_sparse}
	Let $F \subseteq E$, $|F| \le f$, be a set of edges and $u,v \in V(F) \cup \{s,t\}$ vertices
	such that $w_H(u,v)$ is computed using \Cref{eq:sparse_balls}.
	Suppose there exists an $u$-$v$-path $P$ in $G-F$
	that is the concatenation of at most two replacement paths, 
	hop-long, and far away from all failures in $F$.
	Then, with high probability, it holds that $w_H(u,v) \le \alpha \cdot |P| +
		 (D f \frac{\log_2(n)}{\varepsilon} {+} 1) \cdot \beta$.
\end{lemma}

\begin{proof}
	Let $x \in \{u,v\}$ be a vertex for which the pivot tree
	did not return a sufficiently close pivot of the first type when processing set $F$.
	Let $F' \subseteq F$ 
	be the collection of edges corresponding to the parent-child traversals of that pivot tree.
	With high probability, $\ball_{G-F'}(x,\lambda/2)$ is sparse
	and the algorithm gains access to the set $\ball_{G-F'}(x,\lambda/2) \cap B_2$
	of $\Otilde(f L^f/\lambda)$ pivots of the \emph{second type}.
	Let $y \in \{u,v\}{\setminus}\{x\}$ be the other vertex.
	\Cref{eq:sparse_balls} implies that 
	\begin{equation*}
		w_H(u,v) \le \min_{p \in \ball{G-F'}(x,\lambda/2) \cap B_2}
			\widehat{d^{\le L}}(x,p,F) + FT_0(p,y,F).
	\end{equation*}
	
	Recall the covering properties of the set $B_2$ (below \Cref{def:trapezoid}).
	Since $P$ is the concatenation of at most two replacement paths, hop-long,
	and far away from $F$, it has a pivot of the second type 
	within distance $\lambda/2$ of either endpoint w.h.p.
	Let $p^* \in B_2$ be the one around $x$.
	The set $\ball_{G-\textbf{\textit{F}}}(x,\lambda/2)$,
	is contained in $\ball_{G-\textbf{\textit{F'}}}(x,\lambda/2)$.
	That means, the pivot $p^*$ is considered when computing $w_H(u,v)$, giving
	\begin{equation*}
		w_H(u,v) \le \widehat{d^{\le L}}(x,p^*,F) + FT_0(p^*,y,F).
	\end{equation*}

	The prefix $P[x..p^*]$ has at most $\lambda/2 \le L$ edges,
	that means, the first term $\widehat{d^{\le L}}(x,p^*,F)$ approximates its length 
	with stretch $(\alpha,\beta)$.
	The suffix observes $|P[p^*..y]|  \ge d_{\varepsilon/6}^{(2f+1)}(p^*,y,F)$.
	As a subpath of $P$, the suffix itself is the concatenation of at most 
	two replacement paths, hence $(2f{+}1)$-decomposable, and far away from $F$.
	Using \Cref{lem:FT_upper_bound}~\emph{(ii)}, we get that
	\begin{align*}
		\widehat{d^{\le L}}(x,p^*,F) + FT_0(p^*,y,F) &\le \alpha \cdot |\nwspace P[x..p^*] \nwspace | + \beta
			 + \alpha \cdot d_{\varepsilon/6}^{(2f+1)}(p^*,y,F) 
				+ D f \frac{\log_2(n)}{\varepsilon} \cdot \beta\\
			&\le \alpha \Big(|\nwspace P[x..p^*] \nwspace | + |\nwspace P[p^*..y] \nwspace | \Big) + 
				\left(D f \frac{\log_2(n)}{\varepsilon} + 1 \right) \cdot \beta\\
			&= \alpha \cdot |P| + \left(D f \frac{\log_2(n)}{\varepsilon} + 1 \right) 
				\cdot \beta. \qedhere
	\end{align*}
\end{proof}

We have collected all the lemmas for the different cases 
in the proof of the stretch below (\Cref{lem:stretch_general_DSO}).
In order to tie them together, we use an result by Chechik et al.~\cite[Lemma~2.6]{ChCoFiKa17}
about paths $P$ that are \emph{not} for away from $F$.
It states the existence of a detour between the endpoints of $P$ that is not much longer
but enjoys the property of being far away from all failures.
Namely, one follows $P$ until a certain vertex $y$.
One then leaves the path
to reach an endpoint $z$ of some failing edge inside of $\tr^{\varepsilon/6}_{G-F}(P)$.
The detour can then be completed by going from $z$ to the other end of $P$. 

\begin{lemma}[\cite{ChCoFiKa17}]
\label{lem:xyz-lemma}
  Let $F \subseteq E$, $|F| \le f$, be a set of edges, $u,v \in V(F) \cup \{s,t\}$ vertices,
  and $P = P(u,v,F)$ their replacement path.
  If $\tr^{\varepsilon/6}_{G-F}(P) \cap V(F) \neq \emptyset$,
  then there are vertices $x \in \{u,v\}$, $y \in V(P)$,
  and $z \in \tr^{\varepsilon/6}_{G-F}(P) \cap V(F)$ with the following properties.
  \begin{enumerate}[(i)]
    \item $\length{P[x..y]} \le \length{P}/2$;
    \item $d(y,z,F) \le \frac{\varepsilon}{6} \cdot d(x,y,F)$;
    \item $\tr^{\varepsilon/6}_{G-F}\!\Big(P[x..y] \circ P(y,z,F) \Big) \cap V(F) = \emptyset$.
  \end{enumerate}
  Thus, the path $P[x..y] \circ P(y,z,F)$ is far away from $F$
  and has length at most $(1 {+} \tfrac{\varepsilon}{6}) \cdot d(x,y,F)$.
\end{lemma}

Finally, we prove the $(\alpha (1{+}\varepsilon), \beta)$ stretch
of our $f$-DSO.
This is the main result of this section.
Recall that the auxiliary graph $H$ has vertex set $V(H) = V(F) \cup \{s,t\}$
and weights on its edges.
The output of the DSO is $d_H(s,t)$.
At the heart of the proof of \Cref{lem:stretch_general_DSO} is an induction over
the pairwise distances in $H$.
Extra care must be taken to control the additive stretch
that accumulates over this induction.
The next lemma also completes the proof of the randomized part of \Cref{thm:general_DSO}. 
The space, query time, and preprocessing time above where all proven under the assumptions
of $\lambda = \Theta(\varepsilon L)$ and $\lambda \le L$.
We now need to make the leading constant precise and set $\lambda = \varepsilon L/9$.

\begin{lemma}
\label{lem:stretch_general_DSO}
	Define $\lambda = \varepsilon L/9$.
	Let $(\alpha,\beta)$ be the stretch of the input $f$-DSO for distances at most $L$,
	where $\beta = o(\frac{\varepsilon^2 L}{f^3 \log n})$.
	Then, w.h.p.\
	$d(s,t,F) \le d_H(s,t) \le \alpha (1{+}\varepsilon) \cdot d(s,t,F) + \beta$.
\end{lemma}

\begin{proof}
	We first cover the easy cases.
	No edge weight $w_H(u,v)$ for  $u,v \in V(F) \cup \{s,t\}$
	underestimates $d(u,v,F)$ by \Cref{lem:never_underestimate}~\emph{(ii)}.
	Therefore, the graph distance $d_H(s,t)$ is not smaller than the replacement distance $d(s,t,F)$.
	If $P(s,t,F)$ has at most $L$ edges, the upper bound is immediate as
	\begin{equation*}
		d_H(s,t) \le w_H(s,t) \le \widehat{d^{\le L}}(s,t,F) 
			\le \alpha \cdot d(s,t,F) + \beta,
	\end{equation*}
	independently of how the weight is computed.
	This is a better stretch than what we claimed, but only holds under the assumption that 
	$P(s,t,F)$ is hop-short.
	The remainder of the proof consists of showing that for hop-long replacement paths we have
	\begin{equation}
	\label[ineq]{eq:stretch_hop-long}
		d_H(s,t) \le \alpha (1{+}\varepsilon) \cdot d(s,t,F).
	\end{equation}
	While it may look as if \Cref{eq:stretch_hop-long} is independent of $\beta$,
	it will be proven by subsuming all the additive stretch in the additional $1+\varepsilon$
	factor, using that $P(s,t,F)$ is hop-long.
	
	Let $X = D f \frac{\log_2(n)}{\varepsilon}$.
	Consider all sets $\{u,v\}$ with one or two vertices of $H$,
	that is, we include the case $u = v$.
	ordered ascendingly by the true replacement distance $d(u,v,F)$
	(ties are broken arbitrarily).
	We use $\rank(u,v)$ to denote the rank of the set $\{u,v\}$ in that order.
	Note the rank is symmetric as $H$ is undirected, that is, $\rank(u,v) = \rank(v,u)$.
	We claim that, for \emph{all} $\{u,v\}$,
	\begin{equation}
	\label[ineq]{eq:induction}
		d_H(u,v) \le \alpha \left(1{+} \frac{2}{3}\varepsilon\right) \cdot d(u,v,F) 
			+ \rank(u,v) (X {+} 1) \cdot \beta.
	\end{equation}
	We prove this by induction over the order.
	It is clear for all singleton sets where $u = v$.
	Assume that the assertion holds for all sets coming before $\{u,v\}$.
	There are three major cases depending on whether $P(u,v,F)$ is hop-short or hop-long and
	whether is far away from all failures in $F$.
	Only \Cref{case:third_case} makes actual use of the induction hypothesis,
	the others are proven directly.
	
	\begin{case}
	\label{case:first_case}
		The replacement path $P(u,v,F)$ is hop-short.
	\end{case}
	
	We have
	$d_H(u,v) \le \widehat{d^{\le L}}(u,v,F) \le \alpha \cdot d(u,v,F) + \beta
		\le \alpha (1{+} \frac{2}{3}\varepsilon)d(u,v,F) 
			+ \rank(u,v) (X {+} 1) \beta$.
			
	\begin{case}
	\label{case:second_case}
		The replacement path $P(u,v,F)$ is hop-long and far away from $F$.
	\end{case}
	
	There are two subcases.
	First, suppose that $w_H(u,v)$ is computed via \Cref{eq:dense_balls}.
	Applying \Cref{cor:weight_long_far_away_dense} with $\delta = \lambda/L = \varepsilon/9$
	gives $d_H(u,v) \le w_H(u,v) \le \alpha (1{+}\frac{\varepsilon}{9}) 
		\cdot d(u,v,F) + X \beta$ w.h.p.
	Otherwise, the weight is computed via \Cref{eq:sparse_balls}.
	\Cref{lem:weight_long_far_away_sparse} shows that
	$d_H(u,v) \le w_H(u,v) \le \alpha \cdot d(u,v,F) + (X{+}1) \cdot \beta$ w.h.p.	
	Both bounds are never larger than what we claim in \Cref{eq:induction}.
	
	\begin{case}
	\label{case:third_case}
		The replacement path $P(u,v,F)$ is hop-long but close to $F$.
	\end{case}
	
	Let $P = P(u,v,F)$.
	By \Cref{lem:xyz-lemma}, there exists an endpoint $x \in \{u,v\}$ of $P$, 
	some other vertex $y \in V(P)$ on the path, and
	and an endpoint of a failing edge
	$z \in \tr_{G-F}^{\varepsilon/6}(P) \cap V(F)$ in the trapezoid
	such that the concatenation $Q = P[x..y] \circ P(y,z,F)$ is far away from all failures in $F$.
	$P(y,z,F)$ denotes the replacement path from $y$ to $z$.
	Further note that $P[x..y]$ is the replacement path from $x$ to $y$,
	since it is a subpath of $P(u,v,F)$.
	In summary $Q$ is the concatenation of two replacement paths and far away from $F$.
	 \Cref{lem:xyz-lemma} also gives $|P(y,z,F)| = d(y,\textbf{\textit{z}},F) \le \frac{\varepsilon}{6} d(x,\textbf{\textit{y}},F)$.
	So $Q$ is only slightly larger than $|P[x..y]| = d(x,y,F)$,
	namely, of length $(1{+}\frac{\varepsilon}{6}) \nwspace d(x,y,F)$.
	For notational convenience, let $w$ be the other endpoint of $P$ in $\{u,v\}{\setminus}\{x\}$.
	
	On our way to establish \Cref{eq:induction} also in this case,
	we claim the following bound on the weight of the edge $\{x,z\}$ in $H$,
	\begin{equation}
	\label[ineq]{eq:xz_weight}
		w_H(x,z) \le \alpha 
			\left(1 + \frac{\varepsilon}{9} \right) \left(1 + \frac{\varepsilon}{6} \right) \cdot
				d(x,y,F) + (X{+}1) \cdot \beta.
	\end{equation}
	Note that the left-hand side is in terms of the pair $(x,z)$
	while the right-hand side is in terms of $(x,y)$.
	Translating between the two is done via 
	$d(x,y,F) \le (1+\frac{\varepsilon}{6}) \nwspace d(x,y,F)$
	by the choice of vertex $z$.
	Intuitively, \Cref{eq:xz_weight} states that following $Q$ from $x$ to $z$
	and continuing to the other endpoint $w$ from there
	is not much more expensive than going directly from $x$ to $w$ along $P$.
	
	We distinguish the same three cases as above but now with respect to path $Q$.
	In the first one, $Q$ is hop-short.
	As in \Cref{case:first_case}, we get
	\begin{equation*}
		w_H(x,z) \le d^{\le L}(x,z,F) \le \alpha \cdot d(x,\textbf{\textit{z}},F) + \beta \le 
	 		\alpha \left(1+\frac{\varepsilon}{6} \right) \cdot d(x,\textbf{\textit{y}},F) + \beta.
	\end{equation*}
	Now assume $Q$ is hop-long and $w_H(x,z)$ is computed using \Cref{eq:dense_balls}.
	This is handled by \Cref{cor:weight_long_far_away_dense}
	which also applies to $Q$ 
	as it is the concatenation of two replacement paths and far away from $F$.
	Reecall that $\delta = \varepsilon/9$.
	We get
	\begin{align*}
		w_H(x,z) &\le \alpha \left(1{+}\delta \right) \cdot d(x,\textbf{\textit{z}},F) 
			+ (X{+}1) \cdot \beta\\
			&\le \alpha \left(1 + \frac{\varepsilon}{9} \right) 
				\left(1 + \frac{\varepsilon}{6} \right) \cdot d(x,\textbf{\textit{y}},F) 
					+ (X{+}1) \cdot \beta.
	\end{align*}
	Finally, assume $w_H(x,z)$ is computed using \Cref{eq:sparse_balls}.
	From \Cref{lem:weight_long_far_away_sparse}, we get
	$w_H(x,z) \le \alpha \cdot |Q| + (X{+}1) \cdot \beta 
		\le \alpha (1+\frac{\varepsilon}{6}) \cdot d(x,y,F) + (X{+}1) \cdot \beta$,
	which implies \Cref{eq:xz_weight}.
	
	Now that we have a bound on the weight $w_H(x,z)$,
	we prove \Cref{eq:induction} in \Cref{case:third_case}.
	Recall that this case assumes that $P = P(u,v,F)$
	is hop-long but \emph{close} to some failure in $F$.
	Also recall that $\{x,w\} = \{u,v\}$ is the same pair of vertices.
	The other vertex $z$ lies in the $\tfrac{\varepsilon}{6}$-trapezoid of $P$.
	Since $\tfrac{\varepsilon}{6} \le \tfrac{12}{7 \cdot 6} < 1$, 
	\emph{every} element of $\tr_{G-F}^{\varepsilon/6}(P)$
	is closer to $w$ in $G-F$ than $x$ is.
	In particular, we have $d(z,w,F) < d(x,w,F) = d(u,v,F)$,
	whence $\rank(z,w)$ is strictly smaller than $\rank(x,w) = \rank(u,v)$.
	We apply the induction hypothesis to $\{z,w\}$.
	It states that
	$d_H(z,w) \le \alpha (1+\frac{2}{3}\varepsilon) \cdot d(z,w,F) + \rank(z,w)(X+1) \cdot \beta$.
	\Cref{eq:xz_weight} gives
	\begin{multline*}
		d_H(u,v) = d_H(x,w) \le w_H(x,z) + d_H(z,w)\\
			\le \alpha \!\left(1 + \frac{\varepsilon}{9} \right)\! 
				\left(1 + \frac{\varepsilon}{6} \right) \nwspace d(x,y,F) + (X{+}1) \nwspace \beta
				+  \alpha \!\left(1+\frac{2}{3}\varepsilon \right) \nwspace d(z,w,F) 
					+ \rank(z,w)(X{+}1) \nwspace \beta.
	\end{multline*}
	
	We first aggregate the multiples of $\beta$ and then those of $\alpha$.
	We have 
	$(\rank(z,w)+1)(X{+}1) \cdot \beta \le \rank(x,w)(X{+}1) \cdot \beta 
		= \rank(u,v)(X{+}1) \cdot \beta$		
	For the other terms, we use that $\varepsilon \le 12/7$ implies
	\begin{equation*}
		\left(1+\frac{\varepsilon}{9}\right) \!\left(1+\frac{\varepsilon}{6}\right)
			+ \left(1 + \frac{2}{3} \varepsilon\right) \frac{\varepsilon}{6} =
		\frac{7}{54} \varepsilon^2 + \frac{4}{9} \varepsilon + 1	
	 	\le 1 + \frac{2}{3} \varepsilon.
	\end{equation*}
	We omit the factor $\alpha$ in the following calculation.
	\begin{align*}
		&\left(1 + \frac{\varepsilon}{9} \right) \left(1 + \frac{\varepsilon}{6} \right) d(x,y,F)
			+ \left(1+\frac{2}{3}\varepsilon \right) d(z,w,F)\\
			&\qquad\le \left(1 + \frac{\varepsilon}{9} \right) 
				\left(1 + \frac{\varepsilon}{6} \right) d(x,y,F)
					+ \left(1+\frac{2}{3}\varepsilon \right) \Big( d(z,y,F) + d(y,w,F) \Big)\\
			&\qquad\le \left(1 + \frac{\varepsilon}{9} \right) 
				\left(1 + \frac{\varepsilon}{6} \right) d(x,y,F)
				+ \left(1+\frac{2}{3}\varepsilon \right) \frac{\varepsilon}{6} \nwspace d(x,y,F) 
					+ \left(1+\frac{2}{3}\varepsilon \right) d(y,w,F)\\
			&\qquad= \left(\left(1+\frac{\varepsilon}{9}\right) \!\left(1+\frac{\varepsilon}{6}\right)
				+ \left(1 + \frac{2}{3} \varepsilon\right) \frac{\varepsilon}{6} \right) d(x,y,F)
					+ \left(1+\frac{2}{3}\varepsilon \right) d(y,w,F)\\
			&\qquad \le \left(1+\frac{2}{3}\varepsilon \right) d(x,y,F)
				+ \left(1+\frac{2}{3}\varepsilon \right) d(y,w,F) 
			= \left(1+\frac{2}{3}\varepsilon \right) d(x,w,F)
			= \left(1+\frac{2}{3}\varepsilon \right) d(u,v,F).
	\end{align*}
	Recombining the terms depending on $\alpha$ and $\beta$ 
	completes \Cref{case:third_case} and establishes \Cref{eq:induction}.
	
	We now use this to also prove \Cref{eq:stretch_hop-long}, which states that 
	$d_H(s,t) \le \alpha (1{+}\varepsilon) \cdot d(s,t,F)$ is true
	for those queries for which $P(s,t,F)$ is hop-long.
	This is the last part of the proof that is still open.
	\Cref{eq:induction} applied to the pair $\{s,t\}$ states that
	\begin{equation*}
		d_H(s,t) \le \alpha \left(1{+} \frac{2}{3}\varepsilon\right) \cdot d(s,t,F) 
			+ \rank(s,t) (X {+} 1) \cdot \beta.
	\end{equation*}
	It is thus enough to prove 
	$\rank(s,t) (X {+} 1) \beta \le  \alpha \nwspace \frac{\varepsilon}{3} \nwspace d(s,t,F)$.
	
	Observe that the replacement path being hop-long implies that $d(s,t,F) \ge L$.
	There are $\binom{|V(F) \cup \{s,t\}|}{2} + |V(F) \cup \{s,t\}| \le 2f^2 + 5f + 3 \le 6f^2$ 
	sets $\{u,v\}$ in $H$,
	which is an upper bound on $\rank(s,t)$.
	The last estimate uses $f \ge 2$.
	The restriction $\beta = o(\frac{\varepsilon^2 L}{f^3 \log n})$
	implies $7Df^3  \nwspace \frac{\log_2(n)}{\varepsilon} \beta \le \frac{\varepsilon}{3}L$
	for all sufficiently large $n$.
	Together with $\alpha \ge 1$, we get
	\begin{multline*}
		\rank(s,t) (X {+} 1) \,\beta \le  6f^2
			\left(Df \nwspace \frac{\log_2(n)}{\varepsilon} + 1\right) \beta
			 \le 7Df^3  \nwspace \frac{\log_2(n)}{\varepsilon} \beta\\
			 \le \frac{\varepsilon}{3} \, L
			 \le \alpha \nwspace \frac{\varepsilon}{3} \, L
			 \le \alpha \frac{\varepsilon}{3} \, d(s,t,F). \qedhere
	\end{multline*}
\end{proof}

\section{Derandomization}
\label{sec:derandomization}

We now derandomize the steps leading up to \Cref{thm:chained}.
The process of making distance oracles fault-tolerant
for short distances (\Cref{thm:short_distance_DSO}) is already deterministic.
We are thus left with the distance oracle in \Cref{thm:distance_oracle}
and the pivots in \Cref{thm:general_DSO}.
We use the critical-path framework of Alon, Chechik and Cohen~\cite{AlonChechikCohen19CombinatorialRP}.
This means computing a deterministic hitting set for the namesake 
\emph{critical (replacement) paths}
and use it in the construction of the oracle.
The crucial issue is to make the set of paths 
as small as possible and easily identifiable in order to keep the derandomization efficient.
Once assembled, the paths are given as input to the straightforward greedy algorithm 
that obtains the hitting set by iteratively selecting the vertex
contained in the most yet unhit paths.
We treat the paths are mere sets of vertices.
This allows us to seamlessly extend the approach to the sets $\ball_{G-F}(u,\lambda/2)$
when derandomizing \Cref{thm:general_DSO}. 

The following folklore lemma states the properties of the greedy algorithm.
A formal proof can be found in~\cite{AlonChechikCohen19CombinatorialRP};
King~\cite{King99FullyDynamicAPSP} gave an alternative algorithm with similar parameters.

\begin{lemma}
\label{lem:greedy_derand}
	Let $V$ be a set of $|V| = n$ vertices,
	and $M \le n$ and $q$ be positive integers (possibly depending on $n$).
	Let $S_1, \dots, S_q \subseteq V$ be sets of vertices that,
	for all indices $1 \le k \le q$, satisfy $|S_k| = \Omega(M)$.
	The \textup{\textsf{GreedyPivotSelection}} algorithm runs in time $\Otilde(qM)$
	and outputs a set $B \subseteq V$ of cardinality $|B| = O(n \log(q)/M)$ such that,
	for every $k$, it holds that $B \cap S_k \neq \emptyset$.
\end{lemma}

\noindent
Preferably, the time needed for the derandomization does not increase the
asymptotic bounds of \Cref{thm:distance_oracle,thm:general_DSO}.
We thus need to find the critical vertex sets as quickly as possible
while also keeping their number $q$ small.
If, however, we choose the minimum size $M$ too small,
we end up with too many pivots.

\subsection{Derandomizing the Near-Additive Distance Oracle}
\label{subsec:derand_DO}

The derandomization of the distance oracles is a good introduction
how we employ \Cref{lem:greedy_derand}.
We focus the discussion on \Cref{thm:distance_oracle}.
The same ideas also apply to \Cref{thm:hierarchy_DO}.
The only random choice in the construction is which elements to include in the set $B$ of pivots.
Recall that the input graph $G$ is assumed to be weighted.
For any non-negative radius $r$, the set $V_r$ contains those vertices $u$
that have at most $K$ vertices within (weighted) distance $r$ around $u$.
Here, $K$ is the parameter from \Cref{thm:distance_oracle}.
The only time we use a property of $B$ other than its size is
right before \Cref{lem:correctness-weighted}.
We require that any vertex $v \in V{\setminus}V_r$ has a pivot $p(v) \in B$ 
with $d(v,p(v)) \le r$.

Let $K[v]$ be the set of the $K$ closest vertices to $v$.
If $v$ is connected to fewer than $K$ vertices, then $K[v]$ is the entire component.
We claim that it is enough to find a hitting set for all the $K[v]$ 
that have exactly $K$ elements.
To see this, first note that disconnected vertices, i.e., $r = +\infty$,
are handled directly by the oracle.
Suppose $v \in V{\setminus}V_r$ for some finite $r$.
The set $\ball(v,r)$ has more than $K$ vertices.
On the one hand, this gives $|K[v]| = K$ as more than $K$ vertices are connected to $v$.
On the other hand, the closest $K$ among them all have distance at most $r$.
Hitting either one of them is sufficient.

In terms of the parameters of \Cref{lem:greedy_derand},
there are $q \le n$ critical sets that we need to hit.
The lower bound on their size translates to choosing $M = K$.
We thus get a hitting set $B$ of size $O(n \log(q)/M) = \Otilde(n/K)$,
the same as we had for \Cref{thm:distance_oracle}.
The sets $K[v]$ are computed anyway during the preprocessing
so this incurs no extra charge.
\textsf{GreedyPivotSelection} runs in time $\Otilde(qM) = \Otilde(nK) = \Otilde(n^{3/2})$
which is dominated by the APSP computation in the preprocessing.

\subsection{Derandomizing the General Distance Sensitivity Oracle}
\label{subsec:derand_pivots}

We now turn to \Cref{thm:general_DSO}.
In particular, we restrict our attention to unweighted graphs.
The only source of randomness are the pivot sets $B_1$ and $B_2$.
We recall their covering properties.
The pivots of the first type in $B_1$ are used to hit all the sets $\ball_{G-F}(u,\lambda/2)$
provided that they are dense, that is, contain more than $L^f$ vertices.
The set $B_2$ hits the (unique) shortest paths in $G$ on at least $L$ edges,
as well as the first and last $\lambda/2$ vertices
of all hop-long concatenations of up to two replacement paths in any $G{-}F$.
In the construction of the oracle, we only used pivots of the second type 
on concatenations that are far away from $F$.
However, we do not use this information for the derandomization.

The set $B_2$ is much easier to make deterministic than $B_1$,
which is why we start with the pivots of the second type.
Recall that we compute all-pairs shortest paths in $G$ during preprocessing
(to get expaths in additional time $\Otilde(fm)$).
While doing so, we collect all paths that have length at least $M_2 = \lambda/(4f{+}2)$.
Evidently, there are at most $q_2 = O(n^2)$ such paths.
\Cref{lem:greedy_derand} with parameters $M_2$ and $q_2$
gives a set $B_2$ with $\Otilde(f n/\lambda)$ pivots.
This is the same number that was guaranteed with high probability for the random construction,
and used throughout \Cref{sec:framework_DSO}.
The running time of \textsf{GreedyPivotSelection} is $\Otilde( n^2 \lambda/f)$.

To show that it is enough to hit all shortest paths in $G$ of length $\lambda/(4f{+}2)$,
we use that concatenations of at most two replacement paths are $(2f{+}1)$-decomposable~\cite{Afek02RestorationbyPathConcatenation_journal}.
Consider such a concatenation with at least $\lambda$ edges in total,
and let $P$ be its prefix (or suffix) with $\lambda/2$ edges.
$P$ itself is $(2f{+}1)$-decomposable, that is,
a concatenation of at most $2f+1$ shortest paths in $G$.
One of them must have at least $\lambda/(4f{+}2)$ edges.
The same argument shows that $B_2$ intersects all replacement paths
with at least $\lambda/2$ edges.
The hop-long shortest path in $G$ are hit by design.

If we were to use the same approach for the pivots of the first type,
we would get way too many vertices.
Conversely, brute-forcing the dense balls in the graphs $G{-}F$ 
for \emph{all} the failure sets $F \subseteq E$ with at most $|F| \le f$ edges
takes too much time.
Instead, we use the pivot trees introduced in \Cref{subsec:general_DSO_preproc}
also for the derandomization.
Recall that we build a pivot tree with at most $f+1$ levels for each vertex $u \in V$.
Each node stores a path of at most $\lambda/2$ edges,
or a set of  $\Otilde(f L^f/\lambda)$ pivots of the second type
in case $\ball_{G-F'}(u,\lambda/2)$ is sparse.
The set $F'$ corresponds to the edges 
that are associated with the parent-child traversals
on the path from the root to the current node.

We interleave the level-wise construction of the pivot trees with derandomization phases.
Consider all the sets $\ball_{G-F'}(u,\lambda/2)$ that were computed to create the nodes
of the current level of all trees.
In the subsequent derandomization phase, it is not enough to hit only the dense balls among them.
The reason is as follows.
Whenever there is no pivot of the first type in the vicinity of $u$,
we store all pivots of the second type that are close to $u$.
The greedily selected pivots in $B_2$ may not be as evenly spread as the random ones,
potentially resulting in even a sparse ball
receiving much more than $\Otilde(fL^f/\lambda)$ many pivots of the second type.
We must therefore hit all balls that have too large of an intersection with $B_2$.

In the first level,
we construct the root nodes of the pivot trees for all $u \in V$
by computing the sets $\ball_{G}(u,\lambda/2)$ in the original graph $G$ via breath-first searches.
We use a slightly different stopping criteria than the one in \Cref{subsec:general_DSO_preproc}.
The BFS halts once the whole ball is explored or $L^f$ vertices have been visited 
(whatever happens first).
This way, we can keep the computing time per ball at $O(L^{2f})$.
To apply \Cref{lem:greedy_derand}, let $S_1, \dots, S_{q_1}$ be all balls
whose computation was stopped at $L^f$-vertices 
or which have an intersection with $B_2$ of size more than $f L^f/\lambda$.
This means that each of them has at least $M_1 = f L^f/\lambda$ elements.
We get a hitting set of $O(n \log(q_1) \lambda/(f L^f))$ vertices.
For each $u \in V$, if $\ball_{G}(u,\lambda/2)$ is among the sets $S_k$,
we store the path from $u$ to the closest selected element.
Otherwise, we store $\ball_{G}(u,\lambda/2) \cap B_2$,
which has size $O(f L^f/\lambda)$.
In the latter case, we mark the node as a leaf.
In each subsequent level up to height $f$,
we iterate this process in that for each non-leaf node, we create a child for each edge on the stored path, and compute the ball around $u$ in the corresponding graph $G{-}F'$.
In the first level, there are at most $O(n)$ balls given to \textsf{GreedyPivotSelection},
but this can grow up to $q_1 = O(n \lambda^f)$ in the last level.

As the new set of pivots of the first type $B_1$, we take the union of all the pivots generated in
the derandomization phases.
These are 
$O(f \cdot n \log(n \lambda^f) \cdot \frac{\lambda}{f L^f}) = \Otilde(f \nwspace n \lambda/L^f)$
and thus a factor $O(\lambda)$ larger than what we had using randomization.
We prove that this does not increase 
the space and preprocessing time by more than constant factors.

We redo the calculations in the paragraph ``Finishing the preprocessing'' 
(just before \Cref{subsec:general_DSO_query}) with the new value of $|B_1|$.
We use the assumptions $\lambda \le L$ and $f \le L$ for simplifications.
\begin{align*}
	&|B_1|^2 \cdot \left(\lambda + O\!\left(f \frac{\log n}{\varepsilon} \right) \right)
		\left( O(\lambda^f) + O\!\left(\frac{\log n}{\varepsilon}\right)^f \right)\\
		&= \Otilde\!\left(f^2 \frac{n^2 \lambda^{f+3}}{L^{2f}} \right) 
			+ \Otilde\!\left(f^2 \frac{n^2 \lambda^3}{L^{2f}} \right) 
				O\!\left( \frac{\log n}{\varepsilon}\right)^f
			+ \Otilde\!\left(f^3 \frac{n^2 \lambda^{f+2}}{L^{2f}} \right) 
				O\!\left(\frac{\log n}{\varepsilon} \right) 
			+ \Otilde\!\left(f^3 \frac{n^2 \lambda^2}{L^{2f}} \right)
				O\!\left(\frac{\log n}{\varepsilon} \right)^{f+1}\\
		&= \Otilde\!\left(f^4 \frac{n^2}{L^{f-3}} \right) 
			+ \Otilde\!\left(f^2 \frac{n^2}{L^{2f-3}} \right) 
				O\!\left( \frac{\log n}{\varepsilon}\right)^f
			+ \Otilde\!\left(f^2 \frac{n^2}{L^{f-3}} \right)
				O\!\left(\frac{\log n}{\varepsilon} \right)  
			+ \Otilde\!\left(f^2 \frac{n^2}{L^{2f-3}} \right)
				O\!\left(\frac{\log n}{\varepsilon} \right)^{f+1}\!.
\end{align*}
The space of the FT-trees without granularity between the pairs in $B_2 \times V$ stays the same.
For the derandomization, we assume $f \ge 4$, giving $n^2/L^{f-3} \le n^2/L$. 
All terms (for both kinds of FT-trees) are thus at most 
$\Otilde(f^2 \nwspace n^2/L) \cdot O(\log (n)/\varepsilon )^{f+2}$.
This is the same as in the randomized construction.

The term in the preprocessing time that depends on the new size of $B_1$ is
\begin{gather*}
	|B_1|^2 \cdot \left( O(\lambda^f) 
			+ O\!\left(\frac{\log n}{\varepsilon}\right)^f \right)
		= \Otilde\!\left(f^2 \frac{n^2 \lambda^{f+2}}{L^{2f}} \right)
			+ \Otilde\!\left(f^2 \frac{n^2 \lambda^2}{L^{2f}} \right) 
				O\!\left(\frac{\log n}{\varepsilon}\right)^f\\
		= \Otilde\!\left(f^2 \frac{n^2}{L^{f-2}} \right)
			+ \Otilde\!\left(f^2 \frac{n^2}{L^{2f-2}} \right) 
				O\!\left(\frac{\log n}{\varepsilon}\right)^f
		= \Otilde\!\left(f \frac{n^2}{L^{f-3}} \right)
			+ \Otilde\!\left(f \frac{n^2}{L^{2f-3}} \right) 
				O\!\left(\frac{\log n}{\varepsilon}\right)^f.
\end{gather*}
The corresponding term for the FT-trees without granularity is 
$O(f n/\lambda) \cdot O(\log(n)/\varepsilon)^{f} = O(f n/L) \cdot O(\log(n)/\varepsilon)^{f+1}$.
Since $n^2/L^{f-3} \le n^2/L$, this term still dominates this part of the preprocessing.

We still have to account for the time spend on the derandomization itself.
Recall that, for $B_2$, this is 
$\Otilde( n^2 \lambda/f)
	= \Otilde( n^2 \nwspace \varepsilon L/f)$.
Regarding $B_1$, the time for the $f+1$ derandomization phases 
can be estimated as being an $O(f)$-factor larger than the one for the last phase.
\Cref{lem:greedy_derand} gives
$\Otilde(f \cdot n \lambda^f \cdot \frac{f L^f}{\lambda}) 
	= \Otilde( f^2 L^{2f-1} n)$.
We combine the times with those needed to precompute the short-distance DSO, the pivot trees,
and LCA data structures just like in \Cref{subsec:general_DSO_preproc}.
Using that $f^2 L^{2f-1} < L^{3f}$, we arrive at
\begin{gather*}
	\mathsf{T}_L + O(L^{3f}n) 
		+ \Otilde\!\left( \frac{n^2 \nwspace \varepsilon L}{f} \right)
		+\Otilde\!\left( f^2 L^{2f-1} \nwspace n \right)
		+ \Otilde(fm) \cdot \Otilde\!\left(f \frac{n^2}{L}\right)
			\cdot O\!\left(\frac{\log n}{\varepsilon} \right)^{f+1}\\
		= \mathsf{T}_L + O(L^{3f}n) + \Otilde\!\left( \frac{n^2 \nwspace \varepsilon L}{f} \right)
			+ \Otilde\!\left(f^2 \frac{mn^2}{L} \right) 
			\cdot O\!\left(\frac{\log n}{\varepsilon} \right)^{f+1}.
\end{gather*}
If $L = \Otilde(\sqrt{f^3 \nwspace m/\varepsilon})$, 
we have $\Otilde(n^2 \varepsilon L/f) = \Otilde(f^2 \nwspace mn^2/L)$.
Therefore, we get the same asymptotic preprocessing time as with randomization.

\section{Proof of \texorpdfstring{\Cref{thm:chained}}{Theorem 5}}
\label{sec:chained}

We chain the new distance oracle of \Cref{cor:distance_oracle}
with the reductions in \Cref{thm:short_distance_DSO,thm:general_DSO} 
to obtain a deterministic $f$-DSO for unweighted graphs with near-additive
stretch and subquadratic space.

\chained*

\begin{proof}
	The static distance oracle of \Cref{cor:distance_oracle} 
	is instantiated for unweighted graphs, i.e., with maximum edge weight $W = 1$,
	and parameter $c = \gamma \ell/(\ell{+}1)$.
	Since $\gamma \le (\ell{+}1)/2$, we have $c \le \ell/2$.
	The DO has stretch $(\alpha,\beta) = (1{+}\frac{1}{\ell}, 2)$, 
	space $\mathsf{S} = \Otilde(n^{2-\frac{\gamma}{\ell+1}})$,
	and query time $\mathsf{Q} = O(n^{\gamma t/(1{+}\ell)})$.
	The time to compute it is $\mathsf{T} = O(mn)$
	when a BFS from every vertex is used for APSP.
	
	We then plug this oracle into \Cref{thm:short_distance_DSO} with a cut-off parameter 
	$L = n^{\frac{c}{\ell (f{+}1)}} = n^{\frac{\gamma}{(\ell{+}1) (f{+}1)}}$
	that distinguishes between hop-long and hop-short distances.
	Note that this implies $L^{f} = n^{(1-\frac{1}{f+1}) \frac{\gamma}{\ell{+}1}}$.
	The assumption $f = o(\log (n)/\log\log n)$ implies
	that $O(f)^f = O(\log n)^f = n^{o(1)}$.
	We use this for the space of the resulting hop-short distance \emph{sensitivity} oracle
	\begin{equation*}
		\mathsf{S}_L = O(fL \log n)^f \cdot \mathsf{S} 
			= L^f \nwspace n^{2-\frac{\gamma}{\ell+1} + o(1)}
			= n^{2- \frac{\gamma}{(\ell+1)(f+1)} + o(1)}.
	\end{equation*}
	Since $\alpha L + \beta + f^2 L = \Otilde(n^{\frac{\gamma}{(\ell{+}1) (f{+}1)}})$
	is negligible compared to $\mathsf{Q} = O(n^{\gamma \ell/(\ell+1)})$,
	the query time of the intermediate data structure is 
	$\mathsf{Q}_L = \Otilde(f \mathsf{Q}) = \Otilde(n^{\gamma \ell/(\ell+1)})$.
	We defer the analysis of the preprocessing time to the end of the proof
	as it uses similar simplifications as introduced next.
	
	The condition 
	$\varepsilon = \omega(\sqrt{\log(n)}/n^{\frac{\gamma}{2(\ell+1)(f+1)}})$ bounding how fast
	$\varepsilon$ can go to 0 (if any) ensures that
	$\beta = 2 = o(\frac{\varepsilon^2 L}{f^3 \log n})$.
	Also, we have 
	$L = n^{\frac{\gamma}{(\ell{+}1) (f{+}1)}} = O(n^{1/(2(f+1))}) 
		= \Otilde(\sqrt{f^3 \nwspace m/\varepsilon})$.
	We can thus apply the deterministic version of \Cref{thm:general_DSO}
	to the $f$-DSO for short distances
	to get a general $f$-DSO with slightly increased stretch $((1{+}\frac{1}{\ell})(1{+}\varepsilon), 2)$.
	
	We first discuss its space.
	We insert our choice of $L$ into the second and third term 
	of the space bound in \Cref{thm:general_DSO}.
	\begin{gather*}
		\Otilde(f L^{2f-1} n) = \Otilde\!\left(n^{ \frac{2f-1}{f+1} \frac{\gamma}{\ell+1}} \cdot n \right)
			= \Otilde\!\left(n^{ 1+ \frac{2f-1}{f+1} \frac{\gamma}{\ell+1}} \right).\\
		\Otilde\!\left(f^2 \frac{n^2}{L}\right) \cdot O\!\left(\frac{\log n}{\varepsilon}\right)^{f+2}
			= \frac{n^{2-\frac{1}{f+1} \frac{\gamma}{\ell{+}1}+ o(1)}}{\varepsilon^{f+2}}.
	\end{gather*}
	They have a factor $n^{-1 + \frac{2f}{f+1} \frac{\gamma}{\ell+1}+ o(1)}/\varepsilon^{f+2}$
	between them.
	Since $\gamma/(\ell{+}1) \le 1/2 < (f{+}1)/(2f)$,
	the latter term dominates.
	With high probability, the total space of our general DSO is
	\begin{equation*}
		\mathsf{S}_L + \Otilde\!\left(f^2 \frac{n^2}{L}\right) 
			\cdot O\!\left(\frac{\log n}{\varepsilon}\right)^{f+2}
			= n^{2- \frac{\gamma}{(\ell+1)(f+1)}+ o(1)}
				+ \frac{n^{2- \frac{\gamma}{(\ell+1)(f+1)}+ o(1)}}{\varepsilon^{f+2}}
			= \frac{n^{2- \frac{\gamma}{(\ell+1)(f+1)}+ o(1)}}{\varepsilon^{f+2}}.
	\end{equation*}
	Its query time is
	$\Otilde(f^5 \frac{L^{f-1}}{\varepsilon^2} \cdot (\mathsf{Q}_L + f))
		= \Otilde(n^{(1-\frac{2}{f+1}) \frac{\gamma}{(t{+}1)}} 
			\cdot n^{\frac{\gamma \nwspace \ell}{\ell+1}}/\varepsilon^2)
		= \Otilde(n^{\gamma}/\varepsilon^2)$.
		
	The only thing left to prove is the preprocessing time.
	Recall that the static DO is computed in time $\mathsf{T} = O(mn)$.
	Consider the first term of the preprocessing time in \Cref{thm:short_distance_DSO}
	The $O(f L \log n)^f = n^{(1-\frac{1}{f+1}) \frac{\gamma}{(\ell{+}1)} + o(1)}$
	part is negligible compared to $\mathsf{Q} = O(n^{\frac{\gamma \ell}{\ell+1}})$.
	Also, we have $(\alpha L + \beta)^f = O(L)^f = n^{(1-\frac{1}{f+1}) \frac{\gamma}{(\ell{+}1)} + o(1)}$.
	The first step of the transformation thus runs in time
	\begin{align*}
		\mathsf{T}_L &= \Otilde(n^2 \cdot (\alpha L + \beta)^f) \cdot ( O(fL \log n)^f + \mathsf{Q} )
				+ O(fL \log n)^f \cdot \mathsf{T}\\
				&= n^{2+(1-\frac{1}{f+1}) \frac{\gamma}{(\ell{+}1)} + o(1)} \cdot \mathsf{Q}
					+ n^{(1-\frac{1}{f+1}) \frac{\gamma}{(\ell{+}1)} + o(1)} \cdot \mathsf{T}\\
				&= n^{2+(1-\frac{1}{f+1}) \frac{\gamma}{(\ell{+}1)} + \frac{\gamma \ell}{\ell+1} + o(1)}
					+ mn^{1+(1-\frac{1}{f+1}) \frac{\gamma}{\ell{+}1} + o(1)}\\
				&= n^{2+\gamma + o(1)}
					+ mn^{1+\frac{\gamma}{\ell{+}1} + o(1)}.
	\end{align*}
	
	The running time of the reduction in \Cref{thm:general_DSO} is
	\begin{equation*}
		\mathsf{T}_L +  O(L^{3f}n) + \Otilde\!\left(f^2 \frac{mn^2}{L} \right) 
			\cdot O\!\left(\frac{\log n}{\varepsilon} \right)^{f+1}.
	\end{equation*}
	The second term is 
	$O(L^{3f}n) = O(n^{1+\frac{3f}{f+1}\frac{\gamma}{\ell+1}})$.
	With our implicit assumption of $m \ge n$,
	this is dominated by  
	$\Otilde(f^2 mn^2/L) \cdot O(\log (n)/\varepsilon)^{f+1}
		= mn^{2-\frac{1}{f+1} \frac{\gamma}{\ell{+}1} + o(1)}/\varepsilon^{f+1}$.
	Therefore, the total preprocessing time in our case simplifies to
	\begin{align*}
		\mathsf{T}_L +\frac{mn^{2-\frac{1}{f+1} \frac{\gamma}{\ell{+}1} + o(1)}}{\varepsilon^{f+1}}
			&= n^{2+\gamma + o(1)} + mn^{1+\frac{\gamma}{\ell{+}1} + o(1)}
				+ \frac{mn^{2- \frac{\gamma}{(\ell+1)(f+1)} + o(1)}}{\varepsilon^{f+1}}\\
			&= n^{2+\gamma + o(1)} + \frac{mn^{2- \frac{\gamma}{(\ell+1)(f+1)} + o(1)}}
				{\varepsilon^{f+1}}. \qedhere
	\end{align*}
\end{proof}

\section*{Acknowledgements} 

\begin{minipage}{0.99\textwidth}
\begin{wrapfigure}{l}{4.75cm}
\flushleft
\vspace*{-2em}
\hspace*{-.25em}
	\includegraphics[scale=.17]{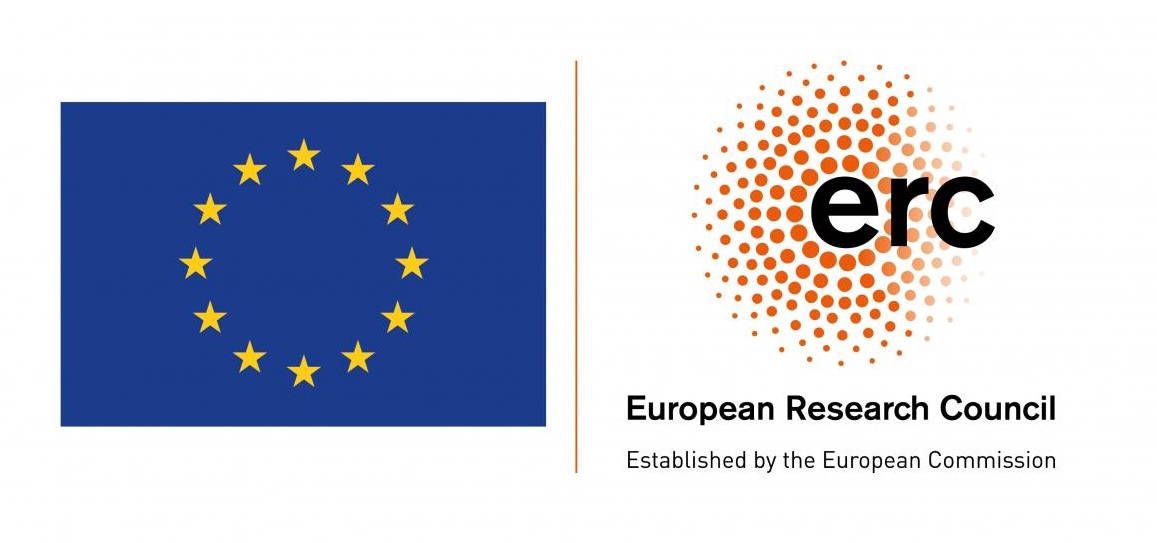}
\end{wrapfigure}

This project received funding from the
European Research Council (ERC) under the European Union's Horizon
2020 research and innovation program,
grant agreement No.~803118 
``The Power of Randomization in Uncertain Environments (UncertainENV)''.
The third author is partially supported by Google India Research Awards.
\end{minipage}

\bibliographystyle{alphaurl}
\bibliography{DSO_bib}

\end{document}